\documentclass[12pt]{article}

\usepackage{graphicx}
\usepackage{amsmath}
\usepackage{amssymb}
\usepackage{amsthm,mathtools,xcolor}
\usepackage{cancel}%used for editorial stuff like striking out fonts 
\usepackage[numbers,sort&compress]{natbib}

\usepackage[colorlinks = false]{hyperref}

\usepackage{subcaption}
\usepackage[margin=1in]{geometry}
\usepackage{enumitem}

\newcommand{\footremember}[2]{%
    \footnote{#2}
    \newcounter{#1}
    \setcounter{#1}{\value{footnote}}%
}
\newcommand{\footrecall}[1]{%
    \footnotemark[\value{#1}]%
}

\newcommand{\be}{\begin{equation}}
\newcommand{\ee}{\end{equation}}
\newcommand{\ba}{\begin{array}}
\newcommand{\ea}{\end{array}}
\newcommand{\bea}{\begin{eqnarray}}
\newcommand{\eea}{\end{eqnarray}}

\DeclareMathOperator*{\argsup}{arg\,sup}
\newcommand{\ra}{\rangle}
\newcommand{\la}{\langle}
\newcommand{\D}{\nabla}
\newcommand{\ip}[2]{\langle #1,#2\rangle}
\newcommand{\ipp}[2]{\langle \langle #1,#2\rangle \rangle}

\newcommand{\norm}[1]{\lVert #1\rVert}
\newcommand{\abs}[1]{\lvert #1\rvert}

\newcommand{\TT}{\mathbb{T}}

\newcommand{\calB}{{\cal B }}
\newcommand{\calH}{{\cal H }}

\newcommand{\calN}{{\cal N }}

\newcommand{\calF}{{\cal F }}

\newcommand{\calJ}{{\cal J }}

\newcommand{\calV}{{\cal V }}
\newcommand{\calE}{{\cal E }}

\newcommand{\calS}{{\cal S }}

\newcommand{\calM}{{\cal M }}

\newcommand{\calW}{{\cal W}}

\newcommand{\EE}{\mathbb{E}}
\newcommand{\ZZ}{\mathbb{Z}}
\newcommand{\CC}{\mathbb{C}}

\newcommand{\RR}{\mathbb{R}}
\newcommand{\HH}{\mathbb{H}}
\newcommand{\bm}{}

\newcommand{\multi}[1]{\mathbf{{#1}}}
\newcommand{\nn}{\multi{n}}
\newcommand{\mm}{\multi{m}}

\newcommand{\herm}{\mathrm{He}}
\newcommand{\polylog}[1]{\mathrm{polylog}{(#1)}}
\newcommand{\voracle}[1]{{\mathsf{Val}_{#1}}}
\newcommand{\poracle}[1]{{\mathsf{Pos}_{#1}}}

\newtheorem{prop}{Proposition}
\newtheorem{lemma}{Lemma}

\newtheorem{corol}{Corollary}
\newtheorem{conj}{Conjecture}
\newtheorem{theorem}{Theorem}
\newtheorem*{theorem*}{Theorem}
\newtheorem*{lemma*}{Lemma}
\newtheorem*{corol*}{Corollary}
\newtheorem{rem}{Remark}

\usepackage{acro}

\DeclareAcronym{nse}{
  short=NSE,
  long=Navier--Stokes equations
}

\title{Quantum algorithms for stochastic nonlinear differential equations}

\author{Sergey Bravyi\footremember{ibm}{IBM Quantum, IBM T.J. Watson Research Center, Yorktown Heights, NY 10598 (USA)}
\and
Adam Byrne \footremember{Dublin}{IBM Quantum, IBM Research Europe, Trinity Business School, Dublin,  D02 F6N2 (Ireland)} \footremember{TCD}{School of Mathematics, Trinity College Dublin, College Green, Dublin 2 (Ireland)}
\and 
Mykhaylo Zayats\footrecall{Dublin}\and Sergiy Zhuk\footrecall{Dublin}}

\begin{document}
\maketitle

\begin{abstract}
Stochastic nonlinear dynamics underlie many models in engineering and computational physics, yet accurate high‑dimensional simulation remains challenging. We present a quantum algorithm for a broad class of $N$-dimensional stochastic differential equations with dissipation and  quadratic drift.
The algorithm applies to  strongly nonlinear systems with all-to-all interactions, thereby 
extending the scope of previously known quantum algorithms that were limited to weak nonlinearity and  sparse systems.
 For norm‑preserving drifts---a condition satisfied by key fluid dynamics discretizations---our method approximates expectation values of low‑order correlation functions with rigorous error bounds at a cost polynomial in $\log{(N)}$ and linear  in the evolution time. Our main technical advance is a subroutine for simulating an auxiliary system of $N$ interacting quantum harmonic oscillators with cost polylogarithmic in $N$. Finally, we formulate turbulence models, including Navier–Stokes and damped Euler equations, within this framework, opening a route to quantum simulation of strongly nonlinear SDEs governing turbulence and nonlinear wave dynamics.
\end{abstract}

\clearpage
\tableofcontents

\section{Introduction}

Numerical simulation of nonlinear dynamical systems, such as turbulent flows, remains a central challenge in computational science. Despite their formal simplicity, quadratic nonlinearities render these systems highly sensitive to modelling errors and uncertainties in initial conditions, hindering accurate prediction. The problem is compounded in large-scale models: nonlinear partial differential equations span multiple spatial and temporal scales, yielding high-dimensional discretizations with many degrees of freedom~\cite{yang2021grid}. These are typically cast as Ordinary Differential Equations (ODE) for a state vector $X(t)$ with $N$ components: e.g. $X(t)$ describes coordinates of a fluid velocity field in Fourier basis, its time derivative is a quadratic  function evaluated at $X(t)$, and larger $N$ provides a more accurate description of the fluid. The objective is to infer properties of $X(t)$ from a given initial state $X(0)$.

\subsection*{Motivation}

Quantum algorithms offer a potential route beyond the aforementioned  limitations. Liu et al.~\cite{liu2021efficient} showed that certain quadratic dissipative ODEs can be solved with runtime logarithmic in $N$ and polynomial in time $t$, hence exponentially faster than classical ODE solvers which scale at least linearly in $N$ as they need to store $X(t)$ in memory. However, this advantage is established only for weakly nonlinear ODEs where dynamics is dominated by linear dissipative term, typical of low–Reynolds-number flows~\cite{mikel-sanz-PRR25}. For strongly nonlinear dissipative ODEs with extra condition of norm preservation, existing approaches do provide runtime logarithmic in $N$ but incur exponential scaling in time $t$~\cite{leyton2008quantum}. The latter arguably rules out most practical applications. Whether generic strongly nonlinear ODEs of practical relevance 
admit quantum algorithms with complexity polynomial in both $t$ and $\log(N)$ remains an open question.

\subsection*{Our approach}
The present work offers a new perspective on quantum simulation of nonlinear dynamics: certain ODEs can be made more computationally tractable by adding {\em random noise}.
We show that additive white noise converts a deterministic nonlinear evolution into a stochastic process that admits efficient quantum simulation, while remaining costly for known classical stochastic methods — making quantum advantage plausible in this setting. This perspective is also physically motivated. Noise is ubiquitous in real systems—arising from environmental coupling, uncertain initial conditions, or imperfect model parameters—and, unless actively suppressed, stochastic dynamics often provide a more realistic description than idealized deterministic models.

This perspective leads to a concrete quantum algorithm. We show that a class of stochastic differential equations (SDEs) can be simulated efficiently on a quantum computer. Specifically, systems of $N$ dissipative SDEs with quadratic, divergence-free drifts, driven by additive white noise and Gaussian noise in the initial condition $X(0)$, admit simulation with runtime linear in $t$, polynomial in $\log N$ and the drift strength, and inverse polynomial in the dissipation rate and target error.
The runtime depends exponentially on the norm of $X(0)$ weighted by a ratio of dissipation and Gaussian noise rates; however, this dependence is mild when $X(0)$ has small norm, and the noise level is away from zero.
The divergence-free condition is analogous to the norm-preservation assumption of Leyton and Osborne~\cite{leyton2008quantum} and is satisfied in many turbulence and photonics models.

In contrast, some quadratic ODEs without noise or dissipation are known to be hard to simulate, both classically and on a quantum computer~\cite{abrams1998nonlinear,childs2016optimal,brustle2025quantum,krovi26}. 
The key obstruction is that nonlinear dynamics
can feature a butterfly effect in which  a pair of nearby state vectors diverges exponentially fast, even if the evolution is norm-preserving. 
This effect can be strong enough to enable the solution of NP-hard problems in polynomial time~\cite{abrams1998nonlinear,childs2016optimal,brustle2025quantum}, suggesting  that an efficient simulation of such systems
is unlikely. Our work demonstrates that adding random noise mitigates the butterfly effect and, in certain cases, makes the dynamics tractable for quantum computers, though not necessarily for classical ones.

\subsection*{Main results}
Let us define a general model for a dissipative nonlinear dynamics with a random noise.
To fix notation, consider first $N$-dimensional Ornstein-Uhlenbeck process, a linear SDEs of the form
$dX_i(t)=-\lambda_i X_i(t)\,dt+\sqrt{q}\,dW_i(t)$.
Here $t\ge 0$ is the evolution time, 
$q>0$ is the noise rate, and $dW_i(t)$ is the infinitesimal Wiener noise: in a small time step
$dt$, the variable $X_i(t)$ receives a random kick of typical size
$\sqrt{q\,dt}$. 
The damping term $-\lambda_i X_i(t)$ describes dissipation: it pulls $X_i(t)$  towards zero, while the
noise spreads it out. The two effects balance in a Gaussian steady-state distribution 
\be
\label{mu}
\mu(x)\propto \exp\!\left[-q^{-1}\sum_{i=1}^N \lambda_i x_i^2\right]
\ee
where $\lambda_i$ is the dissipation rate for the variable $X_i(t)$. 
Let us order the variables such that  $\lambda_1=\min_i \lambda_i$ is the smallest dissipation rate. 

We consider $N$-dimensional nonlinear SDEs of the form 
\be
\label{SDE}
dX_i(t) = -\lambda_i X_i(t) dt  + f_i(X(t))dt  + 
\sqrt{q}\, dW_i(t), 
\ee
where 
\[
f_i(x) =  
\sum_{j,k=1}^N c_{ijk}  x_j x_k
\]
are quadratic {\em drift functions}. Assume without loss of generality that $c_{ijk}=c_{ikj}$ for all indices.
The system is said to have 
 {\em drift strength} $J$ if  any
two-dimensional slice of the tensor $c$ has norm at most $J$.
Thus, for any  index $i$ one must have  $\sqrt{\sum_{j,k=1}^N |c_{ijk}|^2}\le J$ and
$\sqrt{\sum_{j,k=1}^N |c_{jik}|^2}\le J$. 
We allow dense drift functions meaning that each variable can be coupled to all other variables. The only restriction is that each {\em pair} of variables can be coupled only to a few other variables. Formally, we say that the drift has $s$ channels if for any fixed pair $(j,k)$ there are at most $s$ indices $i$ such that $c_{ijk}\ne 0$. Below we assume that $s$ is a constant independent of $N$.
We call Eq.~(\ref{SDE}) {\em divergence-free} if
\be
\label{div_free}
\mathrm{div}(\mu(x) f(x))=0
\ee
for all $x\in \RR^N$, where $\mu(x)$ is defined by Eq.~(\ref{mu}).
That is, the nonlinear drift acts like an incompressible flow when space is weighted by the Gaussian density $\mu$. For example, we prove that drifts representing Navier-Stokes and Euler equations in Fourier basis  are divergence-free and have $s\le 4$ channels.

Our quantum algorithm, illustrated in Fig.~\ref{fig:diagram}, departs radically from state-of-the-art quantum ODE solvers~\cite{liu2021efficient,koopman1931hamiltonian,neumann1932operatorenmethode,kowalski1997nonlinear,engel2021linear,tanaka2023polynomial,novikau2025quantum,shi2023koopman}, which aim to track the full state vector $X(t)$. Instead, our goal is to approximate a given observable function $u_0\, : \, \RR^N\to \RR$ evaluated at the solution $X(t)$ for a given initial vector $X(0)$ and averaged over noise.
For example, $u_0(x)=x_i$ and $u_0(x)=x_i x_j$ are observables describing the first and the second moments of the solution.
In general, we allow linear observables $u_0(x)=\sum_{i=1}^N b_i x_i$
with a constant norm of the coefficient vector  $b$ 
and nonlinear monomial observables
$u_0(x)=\prod_{i=1}^N x_i^{d_i}$ with a constant total degree.

Given the target initial condition $x\in \RR^N$ and the evolution time $t\ge 0$, define an expected value 
\be
\label{v(t,x)}
v(t,x)=\EE_{\mathrm{noise}} \left[ u_0(X(t))\right]
\ee
where the  observable function is evaluated at the 
solution $X(t)$ of the SDE Eq.~(\ref{SDE}) with the initial condition 
\be
\label{ic-noisy}
X(0)=x+z
\ee
for a random normally distributed vector $z\in \RR^N$ that models noise in the initial condition.
The expected value is taken over both Wiener noise in the SDE and noise in the initial condition.
Each component $z_i$ is drawn independently from  the normal distribution 
with the zero mean and variance $p_i$.
The parameter $p_i$ can be viewed as the initialization noise rate for the $i$-th variable.
We choose $p_i=q/(2\lambda_i)$, as this choice greatly simplifies the technical analysis but we expect that more general noise models can be considered.
We assume that the target initial condition $x$ is a sparse vector with $O(1)$ nonzero components.

Our main result is a quantum algorithm for estimating observables of high-dimensional
quadratic SDEs. The following theorem describes the algorithm and its runtime accounting for both gate complexity and the number of queries to the oracles specifying dissipation rates,  drift functions, and the observable.
\begin{theorem} 
\label{thm:main}
There exists a quantum algorithm taking as input the target initial condition $x\in \RR^N$, evolution time $t$, an error tolerance $\epsilon$, and the observable function $u_0$ as above. The algorithm outputs a real number approximating the expected value of the observable defined in Eq.~(\ref{v(t,x)})
with an additive error $\epsilon$. The algorithm has runtime scaling linearly with $t$,
polynomially with  $\log{N}$,  $J$,  $\lambda_1^{-1}$, $\epsilon^{-1}$, and 
exponentially with the quantity $\sum_{i=1}^N x_i^2/p_i$, where $p_i=q/(2\lambda_i)$ is the noise rate for the $i$-th variable.
\end{theorem}
The algorithm estimates the time-evolved expectation of a scalar observable.
Equivalently, it solves the backward Kolmogorov equation (bKE), the dual of the
Fokker--Planck equation governing the evolution of probability densities. The
analysis therefore rests on the bKE, a standard object in stochastic process
theory~\cite{gikhman2004theory,da2004kolmogorov}. In the present setting this
analysis is nontrivial: the drift in Eq.~\eqref{SDE} is quadratic, and hence
unbounded and non-Lipschitz. As a result, even the existence theory for bKEs with
quadratic drifts in Hilbert spaces equipped with Gaussian measure appears not to
be fully settled in the literature~\cite{flandoli1998kolmogorov,Flandoli2019KolmogorovGaussian}.

\begin{figure}
\centering\includegraphics[width=16cm]{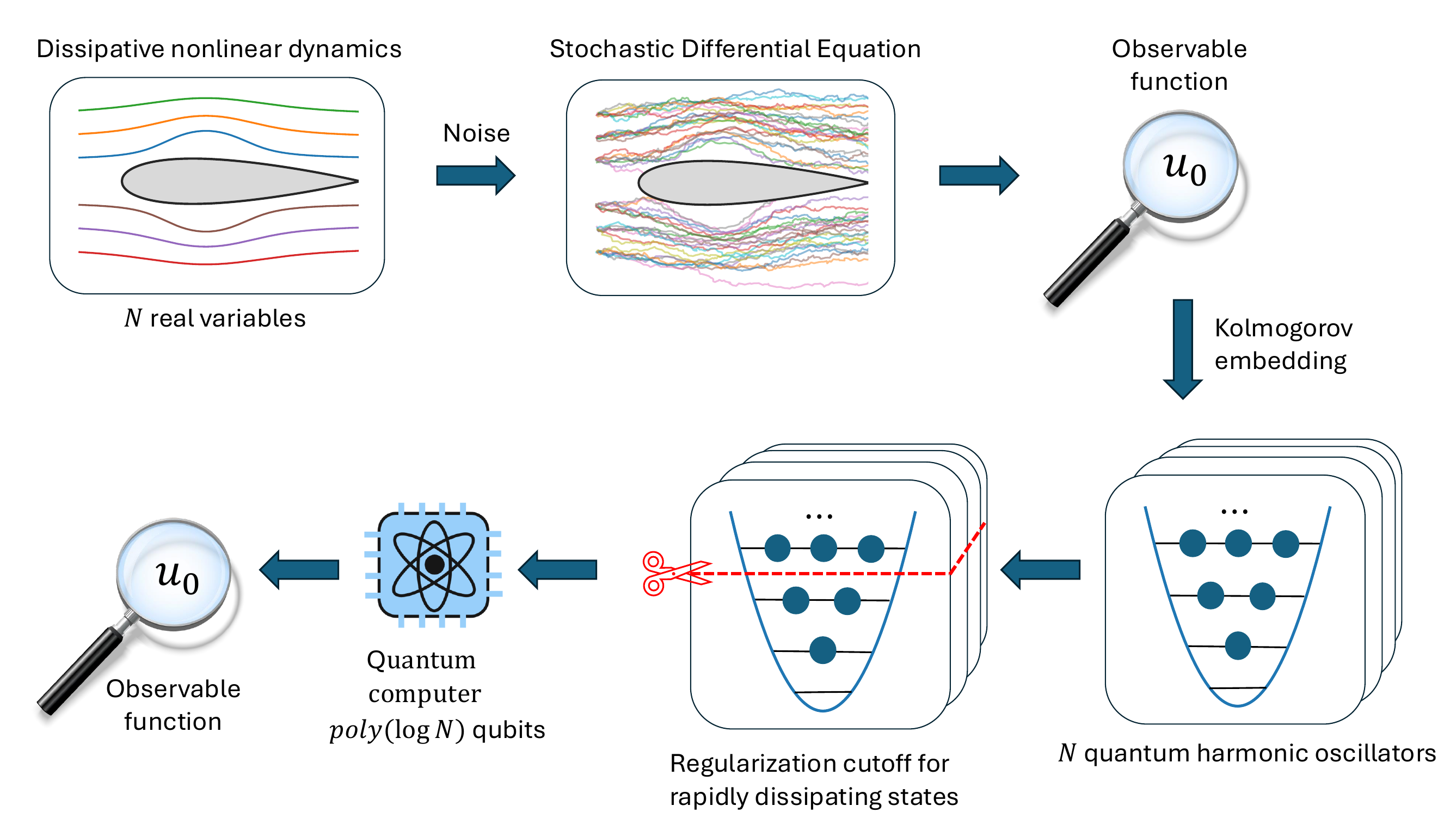}
        \caption{
        {\em Quantum simulation workflow.} A quadratic ODE for a state vector $X(t)\in\RR^N$ is converted into a stochastic differential equation (SDE) by adding random noise.
        Rather than tracking the full state vector, the quantum algorithm estimates the noise-averaged value of a scalar observable $u_0:\RR^N\to\RR$ evaluated at the solution $X(t)$ for a given target initial condition $x\in \RR^N$.
        As a function of $t$ and $x$, this expected value satisfies the Kolmogorov equation.
        Embedding this equation into the Hilbert space of $N$ quantum harmonic oscillators and applying a regularization procedure then enables efficient quantum simulation.}
    \label{fig:diagram}
\end{figure}

We overcome this difficulty by embedding the bKE into a system of $N$ quantum
harmonic oscillators using second quantization. In this representation, the time
evolution is generated by a second-quantized Kolmogorov differential operator,
which we call the \emph{Kolmogorian}, in analogy with the Hamiltonian of quantum
mechanics. We shall see that for quadratic SDEs, the Kolmogorian is cubic in bosonic creation and
annihilation operators. Unlike standard particle-preserving (or number-conserving) Hamiltonians, it can change the total particle number. This feature makes direct simulation inefficient unless the dynamics can be confined to a finite
computational subspace. 
The central technical ingredient of our work is a regularization procedure that
uses dissipation and noise to truncate the Kolmogorian to a finite-dimensional
subspace while preserving the target bKE evolution up to a rigorously controlled
error. This enables simulation of a system of $N$ quantum harmonic oscillators with cost polynomial in $\log N$.

After approximately preparing the second-quantized bKE solution $\psi(t)$, the
desired observable is obtained from the overlap
\[
    v(t,x)=\langle \psi_{\mathrm{out}}(x) \mid \psi(t)\rangle .
\]
Here $\psi_{\mathrm{out}}(x)$ is a readout state, given by a tensor product of
$N$ coherent states encoding the initial condition $x\in\mathbb{R}^N$, as formally defined below.
The norm of this readout state grows exponentially with
$\sum_i x_i^2/p_i$, which accounts for the exponential dependence on this
quantity in Theorem~\ref{thm:main}.

\subsection*{Impact and applications}
We answer affirmatively an open question whether a quantum computer can 
simulate strongly nonlinear $N$-dimensional dynamics with cost polynomial $\log{N}$ and the evolution time. Indeed, the algorithm of Theorem~\ref{thm:main} 
achieves the desired simulation cost scaling when the initial condition $x$ has sufficiently small norm and the noise is sufficiently strong so that
$\sum_{i=1}^N x_i^2/p_i=O(1)$.
We identify practically relevant turbulence models for which this condition holds.

Our main technical tool, the backward Kolmogorov equation, has already played a central role in the analysis of major open problems, including the Navier--Stokes existence and uniqueness Millennium Prize Problem~\cite{fefferman2006existence,ClayNavierStokes,dalang2015stochastic}. We believe that it can play a similarly transformative role in the emerging theory of quantum algorithms for nonlinear differential equations. The framework developed here suggests a shift in perspective: instead of reconstructing the full state vector $X(t)$~\cite{liu2021efficient,koopman1931hamiltonian,neumann1932operatorenmethode,kowalski1997nonlinear,engel2021linear,tanaka2023polynomial,novikau2025quantum,shi2023koopman}, future quantum ODE and SDE solvers may focus on the direct estimation of scalar observables of $X(t)$. 

The regularization procedure introduced in this work---projection of the Kolmogorian onto a finite-dimensional subspace with rigorously controlled error---provides a concrete mechanism for making this paradigm algorithmic. We expect this construction to serve as a foundational building block for future quantum algorithms for nonlinear stochastic dynamics, enabling scalable simulations in regimes where full state reconstruction is neither necessary nor feasible. Beyond algorithm design, the Kolmogorian framework also creates a new bridge between stochastic analysis, turbulence theory, and computational complexity. It has already led to natural conjectures on the classical simulability of the Euler and Navier--Stokes equations; for instance, we conjecture below that in high-dimensional SDEs with strong dissipation, local observables may be efficiently estimated by Monte Carlo-like classical algorithms. In this sense, the impact of the present work extends beyond a single quantum algorithm: it introduces a framework that may guide the future classification, simulation, and complexity analysis of nonlinear dynamical systems on both quantum and classical computational platforms. 

\paragraph{Potential use cases.}
The class of SDEs addressed by our method includes randomly forced turbulence models, such as discretizations of the two-dimensional Navier--Stokes equations~\cite{Foias_Manley_Rosa_Temam_2001} and damped Euler equations in two and three spatial dimensions. In these examples, incompressibility is required to satisfy our divergence-free condition, and the advection term must be regularized to suppress rapidly oscillating Fourier modes. Additional examples include canonical toy models of turbulence, such as the randomly forced Lorenz 96 model~\cite{lorenz1996predictability}, the Burgers--Hopf model~\cite{frank2018detectability}, and the Orszag--McLaughlin system~\cite{orszag1980evidence}. Beyond turbulence, the same framework applies to nonlinear Markov diffusions, including overdamped Langevin dynamics for global optimization~\cite{trillos2023optimization}, as well as randomly forced ODEs arising in photonics, for example the three-wave mixing model~\cite{saltiel2005multistep} for electromagnetic waves interacting in nonlinear media.

As a concrete use case, we consider a modified damped inviscid Euler--Bardina turbulence model in two spatial dimensions, referred to below as the dEB model~\cite{BardinaFerzigerReynolds1980SGS,LaytonLewandowski2006WellPosedTurbulenceModel}. The dEB model is well posed in both two and three dimensions~\cite{CaoLunasinTiti2006Bardina}. Numerically, its viscous variant correlates well with subgrid-scale Reynolds stresses; however, because of its weaker dissipation, it typically requires augmentation by eddy-viscosity models to better match full-scale simulations of homogeneous turbulence~\cite{BardinaFerzigerReynolds1983NASA}. This weaker dissipation is precisely what makes the dEB model a useful test case for our purposes. Below, we conjecture that models with stronger dissipation, such as that induced by the standard Stokes operator, may allow efficient extraction of local information---for example, local correlations between neighboring spatial locations---by Monte Carlo-like classical algorithms. We show that the discrete dEB model can be simulated efficiently by our quantum algorithm. For illustration, Fig.~\ref{fig:dEB} shows a classical Euler--Maruyama simulation of $v(t,x)$, as discussed below.

In summary, our work establishes a route toward practical quantum algorithms for nonlinear SDEs with direct relevance to turbulence and nonlinear wave dynamics to name a few. Indeed, we prove that our algorithm eﬃciently simulates a quadratic SDE which formally represents a projection of two-dimensional dEB model. 
This suggests that Kolmogorian simulation algorithms could support both theoretical and numerical studies of turbulence. In a quantum-centric supercomputing setting~\cite{Bravyi2022JAP_FutureSCQubits}, a quantum solver for the dEB model, or related turbulence models, could augment state-of-the-art classical solvers by providing high-resolution data otherwise out of reach, potentially in combination with eddy-viscosity closures~\cite{BardinaFerzigerReynolds1983NASA} such as the dynamic Smagorinsky model~\cite{Germano1991DynamicSGS,Lilly1992GermanoModification}. As large-scale fault-tolerant quantum computers become available, the proposed framework offers a principled path toward both new theoretical insight and computational advances across these domains.

\subsection*{Previous work}

Most previous work on quantum algorithms for nonlinear differential equations has focused on the Carleman embedding. The original algorithm of Liu et al.~\cite{liu2021efficient} has since been extended in several directions. Krovi~\cite{krovi2023improved}
 and Costa, Schleich, Morales, and Berry~\cite{costa2025further} exponentially improved the scaling with the error tolerance and achieved nearly-linear scaling with the evolution time using higher-order methods and rescaling. Liu et al.~\cite{liu2023efficient} presented an efficient quantum algorithm for nonlinear reaction-diffusion equations and proved convergence of the Carleman method under a less stringent condition than in the original work~\cite{liu2021efficient}, with the strength of the nonlinearity measured using an $\ell_\infty$ norm rather than an $\ell_2$ norm. Surana, Gnanasekaran, and Sahai~\cite{surana2024efficient} extended Carleman-style quantum simulation from quadratic to higher-degree polynomial vector fields. Wang, Jia, et al.~\cite{WangJiaVeerapaneniDing2026} relaxed convergence conditions for the Carleman embedding using pivot-shift and rescaling transformations. Finally, Wu, Wang, and Li~\cite{wu2025quantum} presented a convergence analysis for Carleman embedding beyond the dissipative setting, assuming a no-resonance condition motivated by the theory of dynamical systems.  

As far as we are aware, rigorous guarantees for quantum algorithms based on Carleman embedding require the effective nonlinearity strength to remain below a threshold determined by the problem parameters. In contrast, our approach does not impose a small-nonlinearity condition: for divergence-free quadratic drifts satisfying our assumptions, the algorithm remains efficient for arbitrarily strong nonlinearity, with a polynomial runtime dependence on the drift strength.

Our approach  is conceptually related  to the line of work based on 
Koopman--von Neumann embeddings 
 which represent nonlinear classical dynamics through linear evolution in an 
enlarged observable space. Quantum algorithms based on Koopman-von Neumann embedding were 
studied by Joseph~\cite{joseph2020koopman} as well as Engel, Smith, and Parker~\cite{engel2021linear},
see also~\cite{novikau2025quantum,shi2023koopman,tanaka2023polynomial}.
Similar to Carleman-based methods, the main difficulty here is not the formal linearization itself but the construction of a finite-dimensional embedding with controlled error and efficient quantum implementation. 

We note that existing quantum algorithms based on Koopman--von Neumann embeddings primarily address deterministic dynamics, whereas in our setting additive random forcing is the key mechanism that provides an invariant Gaussian measure and enables efficient regularization with controlled error.

{\bf Connection to the preprint~\cite{bravyi2025quantum}}
The present work builds on several ideas and technical results from our earlier preprint~\cite{bravyi2025quantum}, which introduced the Kolmogorian simulation framework for quadratic SDEs, developed an efficient readout procedure, and proved regularization error bounds. However, the algorithm of~\cite{bravyi2025quantum} is limited to sparse SDEs, in which each variable is coupled to only a few other variables. This sparsity assumption is rather restrictive for discretized  turbulent PDEs, especially when combined with the divergence-free condition. Consequently, the framework of~\cite{bravyi2025quantum} does not directly apply to the dense models considered here.

The main contribution of the present work is to remove the sparsity assumption and to identify a concrete discretized turbulence model, namely the dEB model, that fits into the resulting framework. This requires two major technical advances: extending the regularization error bounds of~\cite{bravyi2025quantum} from sparse to dense Kolmogorians, and developing efficient block encodings for dense regularized Kolmogorians. Lastly, we employ Kolmogorian approach to reveal regimes in which discretized Euler and Navier–Stokes equations might be efficiently simulated by classical algorithms. These changes substantially broaden the scope of the original framework and motivate the present work as a separate submission rather than an update of~\cite{bravyi2025quantum}.

\section{Technical overview}
\label{sec:methods}

This section sketches the main technical tools used in the proof of Theorem~\ref{thm:main}. We first describe how the time evolution of noise-averaged observables can be reformulated as a linear Kolmogorov equation in the Fock space of $N$ quantum harmonic oscillators. We then summarize our regularization procedure, which enables projection of the Kolmogorov equation onto a finite-dimensional subspace with controlled error, and sketch the main steps of the resulting quantum algorithm. We explain how this framework applies to the damped Euler-Bardina model and how this example is used in our numerical experiments. Finally, we discuss implications of our results for classical algorithms for solving ODEs and SDEs. The precise definitions, oracle models, error bounds, block encodings, and formal proofs of the results described above are given in the sections that follow.

\subsection{Kolmogorov equation and regularization}

The first step of our quantum algorithm is the Kolmogorov embedding, see Fig.~\ref{fig:diagram}.
It converts a nonlinear dynamics into a linear one, albeit defined in an infinite-dimensional Hilbert space.

We consider the backward Kolmogorov equation describing the time evolution of the expected value
\[
u(t,x) =
  \EE_{W} u_0(X(t))
    \quad \mbox{subject to $X(0)=x$}.
 \]
Here the expected value is over Wiener noise $W(t)$ in Eq.~(\ref{SDE}).
Note that the expected value $v(t,x)$ computed by our quantum algorithm can be expressed as
$v(t,x)=\EE_z u(t,x+z)$, where $z$ is the noise in the initial condition.
The Kolmogorov equation reads as
 \[
\frac{\partial}{\partial t} u(t,x) = \sum_{i=1}^N (-\lambda_i x_i + f_i(x)) \frac{\partial}{\partial x_i} u(t,x)
+\frac{q}2 \sum_{i=1}^N  \frac{\partial^2}{\partial x_i^2} u(t,x)
\]
with the initial condition $u(0,x)=u_0(x)$. 
By mapping each variable $x_i$ to a quantum harmonic oscillator,
we express the desired expected value $v(t,x)$ as an inner product of two (unnormalized) quantum states,
\[
v(t,x)=\la \psi_{out}(x)|\psi(t)\ra,
\]
where the ambient Hilbert space describes a system of $N$ quantum harmonic oscillators, $|\psi_{out}(x)\ra$ is the
coherent state embedding~\cite{engel2021linear} of the initial condition  $x$, 
the state $|\psi(0)\ra$
encodes the observable function $u_0$,
and $|\psi(t)\ra$ solves the second-quantized Kolmogorov equation  
\be
\label{ACmethods}
\frac{d}{dt} |\psi(t)\ra = (-A+C)|\psi(t)\ra
\ee
with
\[
A=\sum_{i=1}^N \lambda_i a_i^\dag a_i, 
\]
and
\[
C=
(q/2)^{1/2}\sum_{i,j,k=1}^N   (\lambda_i\lambda_j  \lambda_k)^{-1/2}   \lambda_i  c_{ijk} (a_j^\dag a_k^\dag a_i -  a_i^\dag a_k a_j).
\]
Here $a_i^\dag$ and $a_i$ are the creation and the annihilation operators for the $i$-th oscillator. 
Hamiltonians composed of terms $a_j^\dag a_k^\dag a_i$ (and their conjugate) are well studied in the quantum optics community — they
describes the spontaneous parametric down-conversion  where a single photon is converted into a pair of photons. 
Note $A^\dag=A$ and $C^\dag=-C$. 

The main technical challenge addressed in this work is proving that the Kolmogorov equation Eq.~(\ref{ACmethods})
has a well-defined solution and this solution can be well-approximated by projecting the time evolution onto a certain finite-dimensional subspace. This is far from obvious since $A$ and $C$ are unbounded operators. In particular,  the states $A|\psi(t)\ra$ or $C|\psi(t)\ra$ may have an infinite norm, in which case
Eq.~(\ref{ACmethods}) is meaningless. Furthermore, the operator $C$  does not preserve the particle number and cannot be considered as a small perturbation in the strong nonlinearity regime. 

The key to addressing these challenges is proving that the cubic term $C$ can be upper bounded by $A$ and $A^2$ such that for any states $\psi$
and $\phi$ one has
\[
|\la \phi|C|\psi\ra|\le \gamma
\sqrt{\la \phi|A^2|\phi\ra\la \psi|A|\psi\ra}
\]
for a certain coefficient $\gamma$ independent of $N$. Our proof gives $\gamma=3J\lambda_1^{-2}\sqrt{qs}$.
Furthermore, for any state $\psi$ one has
\[
|\la \psi|[A,C]|\psi\ra|\le \la \psi|A^2|\psi\ra + \kappa_1\la \psi|A|\psi\ra
\]
for $\kappa_1=6J^2\lambda_1^{-2} qs$.
We derive similar $N$-independent  bounds for commutators $[A^p,C]$ with any $p\ge 1$.
Loosely speaking, these bounds ensure that the damping effect of the term $-A$ in the Kolmogorov equation  is strong enough to suppress the creation of new particles due to $C$ and keep the norm of the states $A|\psi(t)\ra$ or $C|\psi(t)\ra$ under control.

Equipped with these bounds,
we were able to prove that the solution $\psi(t)$ of Eq.~(\ref{ACmethods}) is well-defined and can be efficiently approximated
on a quantum computer. To this end we define a low-dissipation subspace $\calH_k$ spanned by Fock basis vectors $|\mm\ra$ such that $\la \mm|A|\mm\ra\le k$, where $k>0$ is 
a cutoff parameter. The subspace $\calH_k$
is finite-dimensional and can be expressed using $O(k\lambda_1^{-1} \log{N})$ qubits.
The regularization ``turns off" the cubic term $C$ on the subspace orthogonal to $\calH_k$ while the action of $C$ on $\calH_k$ remains unchanged. 
We prove that the regularized Kolmogorov equation has the unique solution $\psi_k(t)$, and that the family of states $\psi_k(t)$ with the same initial condition $\psi_k(0)=\psi(0)$ has a limit: 
$\psi(t)=\lim_{k\to \infty} \psi_k(t)$. This limit $\psi(t)$ solves Eq.~(\ref{ACmethods}). Importantly, the regularization error
$\epsilon=\|\psi(t)-\psi_k(t)\|$ is upper bounded by a polynomial function of $1/k$, $J$, $\lambda_1^{-1}$. Hence to attain the given error $\epsilon$
one can choose $k=\operatorname{poly}(\epsilon^{-1},J,\lambda_1^{-1})$. For a formal statement see Theorem~\ref{thm:regul} in Section~\ref{sec:regul-sz}.

\subsection{Quantum algorithm}

Our algorithm consists of three steps:

\begin{itemize}[label={}]
    \item \textbf{Step 1.} Prepare an initial state $\psi(0)$ encoding the observable function $u_0(x)$. 
    To this end one expands $u_0(x)$ in the basis of $N$-variate Hermite polynomials. Coefficients in this expansion determine amplitudes of $\psi(0)$ in the Fock basis.
    For example, if $u_0(x)=x_i$ then $|\psi(0)\ra = \sqrt{q/2\lambda_i} a_i^\dag |\bf{0}\ra$, where $|\bf{0}\ra$ is the vacuum state. More precisely, we prepare a projected state $\Pi_k\psi(0)$, where $\Pi_k$ 
    is the projector onto the low-dissipation subspace
    $\calH_k$.

    \item \textbf{Step 2.} Simulate time evolution described by Kolmogorian $-A_k+C_k$, where $A_k$ and $C_k$ are the  projections of $A$ and $C$  onto $\calH_k$.
    To this end we construct quantum circuits block encoding $A_k$ and $C_k$. These circuits have gate and oracle query complexity  polynomial in $k$, $J$, and $\lambda_1^{-1}$. 
    Given the block encodings of $A_k$ and $C_k$, we employ the Linear Combination of Hamiltonian Simulations method of An, Liu, and Lin~\cite{an2023linear,an2026quantum} and its improved version  due to Low and Somma~\cite{low2025optimal} to simulate the time evolution under $-A_k+C_k$. This gives the time evolved state $\Pi_k|\psi_k(t)\ra= e^{t(-A_k+C_k)}\Pi_k|\psi(0)\ra$.

    \item \textbf{Step 3.} Estimate the desired expected value as $v(t,x)\approx\la \psi_{out}(x)|\Pi_k|\psi_k(t)\ra$, where $\psi_{out}(x)$ is a certain readout state encoding the initial condition $x$. The inner product is estimated using the Hadamard test.
    The readout state $\psi_{out}(x)$ coincides with the coherent state encoding of $x$ defined as
    \[
    |\psi_{out}(x)\ra = \bigotimes_{i=1}^N  |\varphi_i\ra, \qquad  
    |\varphi_i\ra =  \sum_{m\ge 0}
     \frac1{\sqrt{(m!)}} \left( x_i\sqrt{\frac{2 \lambda_i}{q}}\right)^{m}  |m\ra.
    \]
    The readout state has norm $\|\psi_{out}(x)\|=\exp{\left[q^{-1}\sum_{i=1}^N \lambda_i x_i^2\right]}$, which can be exponentially large.  This  amplifies errors incurred in the Hadamard test translating to the runtime scaling exponentially with the quantity $q^{-1}\sum_{i=1}^N \lambda_i x_i^2$. 
\end{itemize}

\subsection{Symmetrized damped Euler-Bardina (dEB) model}

As a concrete use case for our quantum algorithm we consider a symmetrization of damped Euler-Bardina (dEB) model~\cite{BardinaFerzigerReynolds1980SGS,LaytonLewandowski2006WellPosedTurbulenceModel} for an incompressible fluid confined to a two-dimensional box $\TT^2=[0,1]^2$ with periodic boundary conditions (BC). It reads as
\be
\label{eq:dEBmethods}
\frac{\partial}{\partial t} V +\overline{(\overline V \cdot \nabla)\overline V} +\lambda V  = -\nabla P.
\ee
Here $t\ge 0$ is time, $V=V(t,\xi)\in \RR^2$ is the fluid velocity at the spatial point $\xi\in \TT^2$, $P=P(t,\xi)$ is the pressure field enforcing incompressibility and $\lambda>0$ is the damping strength. The quantity $\overline V = \calF_{\alpha\beta}V$ is a filtered version of $V$ with filter $\calF_{\alpha\beta}=(I+\alpha S^\beta)^{-1}$ given by the Green's function of the operator 
$I+\alpha S^\beta$, where $S = \frac{1}{2\pi}(-\Delta)^{\frac{1}{2}}$ for the \textit{Stokes operator} $-\Delta$ \cite{Foias_Manley_Rosa_Temam_2001}, which coincides with Laplacian in the case of periodic BC. For $\alpha=\beta=\lambda=0$, Eq.~\eqref{eq:dEBmethods} reduces to the classical Euler equation. For $\lambda=0$ and $\alpha,\beta>0$, Eq.~\eqref{eq:dEBmethods} becomes the classical inviscid Euler-Bardina model, provided one bar is dropped in the nonlinear advection term, $(\overline V \cdot \nabla)\overline V$. We note that dEB is related to the $\alpha$-Navier-Stokes regularization models~\cite{CaoLunasinTiti2006Bardina}. Similarly to~\cite{Verstappen_2008}, we add the extra bar to preserve skew-symmetry and prove that the unique solution $V_\alpha$ of dEB converges to the unique solution $V$ of the classical dEB model at the rate
$\|V_\alpha- V\|_{L^2}=O(\alpha^{\frac{1}{\beta}})$. This and all the following claims/statements are proved in Section~\ref{sec:examples}.

We discretize Eq.~\eqref{eq:dEBmethods} in the basis of divergence-free sine waves
\[
e_p(\xi)=\sqrt{2}\,\frac{p^\perp}{\|p\|}\sin(2\pi p\cdot \xi), \qquad \xi\in \TT^2.
\]
Here $p=(p^1,p^2)\in \ZZ^2$ is a nonzero integer wave vector and $p^\perp=(-p^2,p^1)$. Note that $p$ and $-p$ can be identified since $e_p(\xi)=e_{-p}(\xi)$.
Choose a momentum cutoff $\Lambda$ and let $p_1,\ldots,p_N\in \ZZ^2$ be the list of all nonzero wave vectors $p$ satisfying $\|p\|\le \Lambda$
with one representative picked from each pair $p$ and $-p$. 
We approximate the velocity field by the Galerkin ansatz 
\[
V(t,\xi)=\sum_{i=1}^N X_i(t)e_{p_i}(\xi),
\]
where $X_i(t)$ are coordinates of $V(t,\xi)$ in the basis $e_{p_i}$.
Projecting Eq.~(\ref{eq:dEBmethods}) onto the span of $e_{p_1},\ldots,e_{p_N}$ yields an $N$-dimensional quadratic ODE
\be
\label{Euler2methods}
\dfrac{dX_i(t)}{dt} = -\lambda X_i(t)  + \sum_{j,k=1}^N c_{ijk} X_j(t) X_k(t) , 
\ee
with coefficients $c_{ijk}$ simply related to the wave vectors $p_i,p_j,p_k$ and the parameters $\alpha,\beta$,
see Section~\ref{sec:examples} for details.
 
The modified dEB model has two important properties:  (i)  the quadratic drift  is divergence-free in the sense of Eq.~\eqref{div_free}, 
(ii) the drift has $s=4$ channels\footnote{The drift coefficients $c_{ijk}$ are nonzero only if the corresponding wave vectors
$p_i,p_j,p_k$ obeys certain momentum conservation rules. For fixed $j$ and $k$ there are at most four wave vectors $p_i$
that obeys these rules, see Section~\ref{sec:examples} for details.},
 and
(iii) the drift has strength  $J=O(1)$ for any constant $\alpha>0$,  $\beta>1$ 
and $J=O(\log(N))$ for  $\beta=1$.
In addition, $J_1\equiv \sum_{i,j,k=1}^N |c_{ijk}|$ diverges to infinity as $N\to\infty$ polynomially in $N$ if $1<\beta<5/3$, and only logarithmically if $\beta=5/3$, and it stays bounded independent of $N$ if $\beta>5/3$. 
    
Finally, by adding Wiener noise we convert the ODE Eq.~\eqref{Euler2methods} to an SDE of the form Eq.~\eqref{SDE} that can be simulated by our quantum algorithm.

\subsection{Numerical experiment}

Here we sketch a numerical experiment demonstrating that  our quantum algorithm can reconstruct observables of turbulent flows at good precision, yet the exponential factor $\exp(\lambda \|x\|_2^2/q)$ (the readout norm) does not drastically increase the runtime reported in Theorem~\ref{thm:main}.
To this end we convert the ODE Eq.~\eqref{Euler2methods} describing the symmetrized damped Euler-Bardina model
to an SDE of the form Eq.~\eqref{SDE}
and simulate it classically to  estimate the noise-averaged observable function. 
The Galerkin ansatz defined above consists of $N=220$ basis functions and we consider the observable
\be
\label{eq:fig-observable}
u_{0,\xi^*}(X(t)) = \sum_{i=1}^N X_i(t) \big( \, (1,0) \cdot \left( w(p_i) e_{p_i}(\xi^*) \right) \big)
\ee
where $\xi^*$ is a given point in the computational domain and $w(p) = 1/(1 + \alpha \|p\|^\beta)$ represents the chosen filter $\calF_{\alpha\beta}$ in the basis $e_{p_i}$, ensuring that $\|u_{0,\xi^*}\|=O(1)$.  
In other words, $u_{0,\xi^*}(X(t))$ approximates the first component of the velocity field $V_1(t,\xi)$ at the spatial point $\xi=\xi^*$. 

\paragraph{Turbulent initial condition $x$.} To illustrate the growing simulation complexity due to the impact of a strong quadratic nonlinearity, we generate a turbulent initial condition (IC) $x$. Namely, we use the Euler-Maruyama (EM) method \cite{Maruyama1955ContinuousMP} to evolve Eq.~\eqref{SDE} with damping $\lambda_i=\lambda=10^{-4}$ and deterministic forcing  
\begin{equation}\label{eq:k-forcing_methods}
    F(t,\xi)
    = 
    F_0
    \begin{pmatrix}
     \sin \left(2 \pi k_f \, \xi_2\right), & 0 %\\
    % 0
\end{pmatrix}
^T
\end{equation}
with $F_0=10$ and $k_f=4$. We fix a single realization of Wiener noise with rate $q=5\times 10^{-4}$ and 
simulate the associated system from the zero state $X(0)=0$ 
until turbulent behavior is observed at time $t=t_{\text{turb}}$. The resulting state $x=X(t_{\text{turb}})$ serves as a \textit{deterministic} IC for the proposed quantum algorithm. The first two rows of \autoref{fig:dEB} show the realization of the velocity field $V(t,\xi)=(V_1(t,\xi),V_2(t,\xi))$ for $t\in[0,t_\text{turb}]$ and $\xi\in\mathbb{T}^2$. 

We stress that, in the noiseless case ($q=0$), dEB with forcing Eq.~\eqref{eq:k-forcing_methods} has analytical solution
\be
\label{eq:deb-analytical-sol_methods}
V(t,\xi) =
\frac{F_0}{\lambda} ( 1 - e^{-\lambda t})
\,
\begin{pmatrix}
     \sin \left(2 \pi k_f \, \xi_2\right),
    & 0 %\\
    % 0
\end{pmatrix}
^T.
\ee
The solution $V$ in Eq.~\eqref{eq:deb-analytical-sol_methods} satisfies $\overline{(\overline V\cdot \nabla)\overline V} = 0$, i.e., the nonlinearity vanishes and dEB simplifies to a linear model. However, for $q>0$, the nonlinear term becomes increasingly pronounced over time due to the Wiener noise until it eventually dominates the linear part of dEB. This justifies the inclusion of Wiener noise for generation of the turbulent IC $x$ above. The behavior of this system is illustrated by the first two rows of~\autoref{fig:dEB}. The first row shows, from left to right, the vector-field $V$ and its 1st and 2nd components at time $t=0.3$: $V$ only slightly deviates from the analytical solution Eq.~\eqref{eq:deb-analytical-sol_methods}, suggesting that the impact of nonlinearity is not strong enough; but in the second row, at $t_{\text{turb}}=2.0$, the turbulence fully develops. Clearly, as $t$ grows, so does the complexity of $V$: before $t=0.3$ one could accurately represent $V$ by just a few of basis functions (as $V$ does not deviate from Eq.~\eqref{eq:deb-analytical-sol}, and the latter requires just one basis function), but for $t>0.3$ one requires an increased number of basis functions to resolve $V$ with given precision.

\paragraph{Dynamics of the averaged observable $v(t,x)$.} 
To simulate dynamics of $v(t,x)$ 
we use EM-method as above but with the following key differences: (i) we reduce the Wiener noise rate to $q=2.5\times 10^{-4}$ and remove deterministic forcing, and (ii) we add noise to the turbulent IC, $X(t_{\text{turb}}) = x + z$ for $z\sim\calN(0,{q}/(2\lambda))$ with $q/(2\lambda)=1.25$, as per Eq.~\eqref{ic-noisy}. 
Rows 3, 4 and 5 of~\autoref{fig:dEB} show the ensemble averages of the resulting velocity field and expectation $v(t,x)$ over $5,000$ samples from $t=t_{\text{turb}}=2.0$ to $t=3.0$. The first two columns show the averaged vector field $V$ and 1st component $V_1$ at timepoints $t>t_{\text{turb}}$, and the last column shows $v(t,x)$ for a number of observable functions $u_{0,\xi^*}$, one per grid point $\xi^*\in l=\{(0.5,0),...,(0.5,1)\}$. Since $u_{0,\xi^*}$ in Eq.~\eqref{eq:fig-observable} approximates $V_1(t,\xi^*)$ at $\xi^*$, then $v(t,x)$ corresponding to $\xi^*\in l$ represents the average of the first component of the velocity $V$ along the vertical line,
$l$. The norm of the readout vector $\psi_{out}(x)$ in this case is $\exp(\lambda \|x\|_2^2/q)=7.3\times 10^{4}$. By slightly increasing the noise in the IC (from $1.25$ to $2.5$) we get a readout norm of just $2.7\times 10^{2}$, as reported in Section~\ref{sec:examples}.   For more details on the numerical experiment, we refer the reader to Section~\ref{sec:examples}.

\begin{figure}
\centering\includegraphics[width=0.81\textwidth]{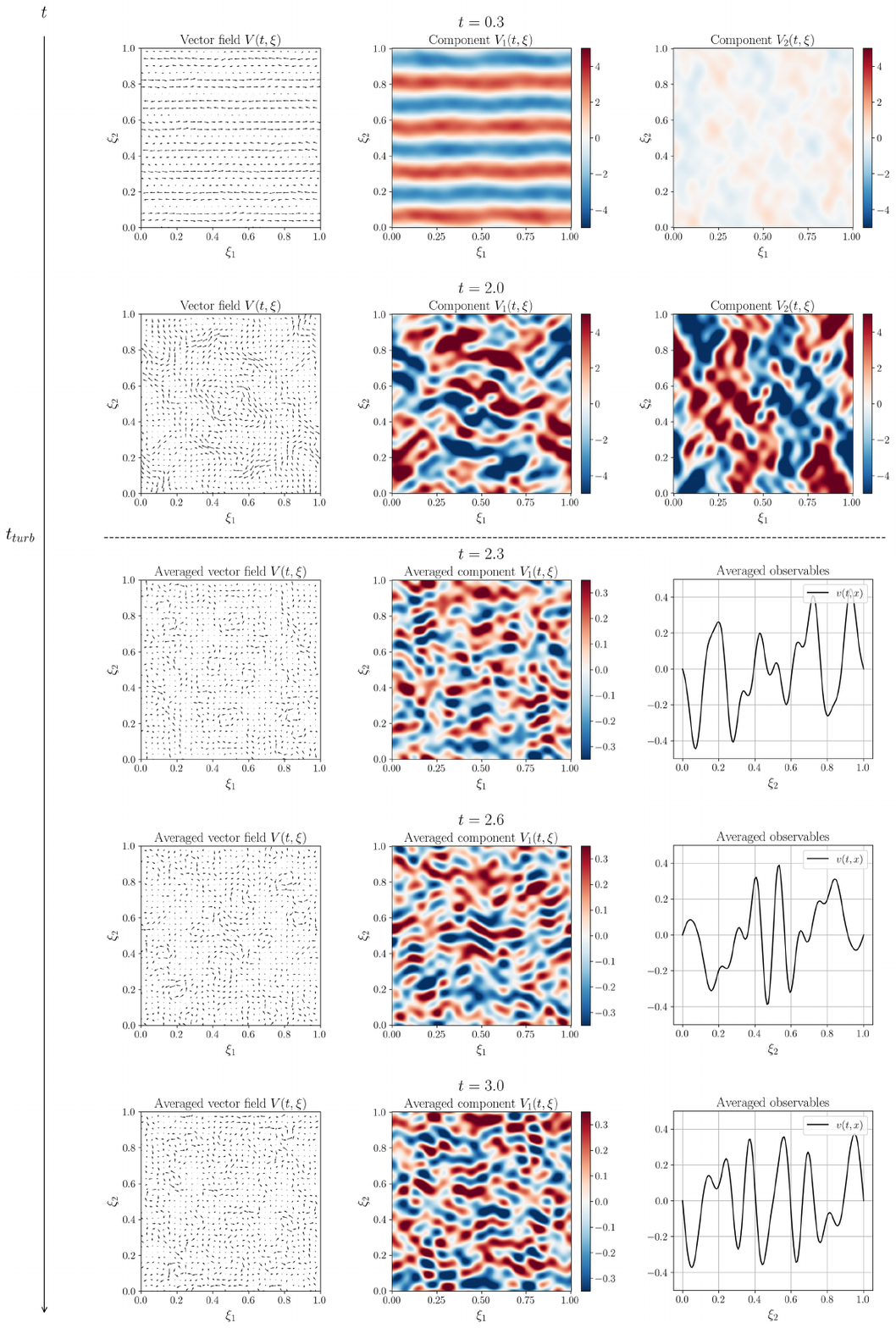}
        \caption{{\em Monte Carlo simulation of dEB.} Rows 1-2; a realization of the velocity field for the SDE Eq.~(\ref{SDE}) describing the discretized dEB model with an additional deterministic forcing Eq.~\eqref{eq:k-forcing_methods}.
        Rows 3-5; the same SDE without deterministic forcing. The first two columns show the averaged velocity field, and the third column shows the expectation $v(t,x)$.
        }
    \label{fig:dEB}
\end{figure}

\subsection{Limitations of classical algorithms}

In this section we address the question of whether  quadratic SDEs satisfying all conditions required for our quantum algorithm
can be solved at a comparable or lower cost on a conventional classical computer. We propose two natural strategies to ``dequantize" our quantum 
algorithm and expose their limitations.

\paragraph{BQP completeness.} To begin with, Ref.~\cite{bravyi2025quantum} proved that the considered problem is BQP-complete in the case when both drift and observable
are  {\em linear functions}, drift functions obey the divergence-free condition $\mathrm{div}(\mu(x)f(x))=0$, and the initial condition for the SDE
is a single  basis vector. 
As a consequence, a classical algorithm solving a linear SDE~(\ref{SDE}) with runtime matching the one of our quantum algorithm for all problem instances 
would be able to simulate the universal quantum computer in polynomial time — a possibility considered highly unlikely. 
However the BQP-completeness result of~\cite{bravyi2025quantum} does not directly extend to homogeneous quadratic drift  considered here.

\paragraph{Tensors with bounded or logarithmic 1-norm.} Next, consider a quadratic $N$-dimensional norm-preserving ODE
\be
\label{norm_preserving_ODE}
\frac{d}{dt} X_i(t) = \sum_{j,k=1}^N b_{ijk} X_j(t) X_k(t), \qquad b_{ijk}=-b_{kji}.
\ee
It can be obtained from the SDE Eq.~(\ref{SDE}) by removing Wiener noise and dissipation terms,
and choosing drift coefficients $c_{ijk}=(1/2)(b_{ijk} + b_{ikj})$,
see Lemma~\ref{lemma:div_free} in Section~\ref{sec:div_free}.
We show that Eq.~(\ref{norm_preserving_ODE})
can be solved classically in time scaling polynomially with $t,\log{N}$,
and drift $1$-norm 
\[
J_1 = \sum_{i,j,k=1}^N |b_{ijk}|.
\]
This algorithm requires an efficient subroutine for sampling triples $(i,j,k)$ from the probability distribution
$P_{ijk} = |b_{ijk}|/J_1$. To describe the algorithm, rewrite
the ODE  Eq.~(\ref{norm_preserving_ODE})  as the Schr\"oedinger equation with a {\em state-dependent} Hamiltonian
\[
\frac{d}{dt} |X(t)\ra = H(X(t)) |X(t)\ra, \qquad H(x) = \sum_{i,j,k=1}^N b_{ijk} x_j |i\ra\la k| = -H(x)^T.
\]
Here $|X(t)\ra = \sum_{i=1}^N X_i(t)|i\ra$ is a time dependent state vector.
We can now simulate the time evolution governed by $H(X(t))$ using qDRIFT algorithm~\cite{qDriftCampbell2019}.
The latter approximates the continuous time evolution under $H(X(t))$ by the first-order Trotter approximation
where each Trotter step evolves under a random term $X_j(t)( |i\ra\la k| -|k\ra\la i|)/2$ with 
$(i,j,k)$ sampled from $P_{ijk}$. The original qDRIFT error analysis~\cite{qDriftCampbell2019} applies almost verbatim
to a state-dependent Hamiltonian showing that $L=O(t^2 J_1^2 \epsilon^{-1})$ Trotter steps is enough to approximate $X(t)$ within error $\epsilon$.
 The key observation is that each Trotter step increases the sparsity of $X(t)$ only by a small constant.
If the initial state $X(0)$ has $s(0)=O(1)$ non-zeros, 
the final state $X(t)$ has at most $s(t)=s(0)+O(L)=O(L)$ non-zeros, which is polynomial in all relevant parameters.
Furthermore, all manipulations with intermediate sparse vectors $X(t)$ can be performed in time linear in $s(t)$. 
We formally describe this algorithm and analyze its runtime in Section~\ref{sec:qdrift-euler}.
We prove that, given an efficient subroutine for sampling triples $(i,j,k)$, the above qDRIFT approach enables efficient classical simulation of damped Euler-Bardina  
model~\eqref{Euler2} in 2D as well as of the closely related $\alpha$-NSE model~\cite{CaoLunasinTiti2006Bardina} for any constant $\beta\ge 5/3$, in which case $J_1=O(\alpha \beta)$ is
independent of $N$. 
In contrast, our quantum algorithm is efficient for any $\beta\ge 1$. Thus the regime $1\le \beta\le 5/3$ is a natural target for a quantum advantage.
Interestingly, the
classical algorithm based on qDRIFT algorithm  
does not directly apply to 
SDEs due to the presence of Wiener noise which maps sparse vectors to dense vectors.

\paragraph{Strongly dissipative SDEs.} Finally, we conjecture that our quantum algorithm can be dequantized
 in the {\em strong dissipation} regime,
when the $i$-th smallest dissipation rate scales as $\lambda_i \propto i^p$ for some positive constant $p>0$.
Indeed, let  $k$ be the regularization cutoff for the second-quantized Kolmogorov equation.
Recall that the subspace $\calH_k$ simulated by our quantum algorithm is spanned by Fock basis vectors $|\mm\ra$ satisfying 
$\sum_{i=1}^N \lambda_i m_i \le k$. 
Therefore,
any basis vector $|\mm\ra \in \calH_k$ must obey $m_i=0$ whenever $\lambda_i>k$.
Let us say that a variable $X_i$ is {\em active} if $\lambda_i\le k$
and let $N_{active}$ be the number of active variables.
The assumption $\lambda_i \propto i^p$
gives  $N_{active} =O(k^{1/p})$. 
Then the active space simulated by our quantum algorithm describes only $N_{active}$ oscillators, whereas all remaining oscillators are frozen in their ground state.
The key observation is that exactly the same second-quantized Kolmogorov equation projected onto $\calH_k$ can be obtained from a truncated SDE describing time evolution of active 
variables only (assuming that the observable function $u_0(x)$ depends only on the active variables). 
Therefore, it is natural to conjecture that the expected observable $v(t,x)$ for the original SDE can be approximated
by solving the truncated SDE with only $N_{active}$ variables. 
The latter can be solved by
classical Monte-Carlo methods, e.g. Euler-Maruyama-type algorithms.  The classical runtime would be polynomial 
in $t,N_{active}$, and the inverse error. 
A practically relevant example of ``strong" dissipation is given by the eigenvalues $\lambda_i$ of Stokes operator determining the dissipative part of 
the Navier-Stokes equation. 
For this reason we illustrate our algorithm on damped Euler-Bardina model which has constant dissipation rates, in which case 
the active space includes all $N$ degrees of freedom. This discussion is formalized in Section~\ref{sec:strongly}.
In this regard, it would be interesting to compare the Euler--Maruyama-type numerical algorithm, assuming the aforementioned conjecture holds, with practical implementations of quantum algorithms for the Navier--Stokes equations based on Carleman embedding, as proposed in~\cite{Jennings2025CarlemanNonlinear,Jennings2025FluidQuantum}. Such a comparison could help identify regimes in which, if any, the authors’ quantum approaches offer a genuine computational advantage over classical methods.

\section{Input data access model}

We assume that the  drift tensor $c$ is specified by the drift strength $J$, number of drift channels $s$, and 
a pair of oracles $\voracle{c}$ and $\poracle{c}$
that encode values and positions of nonzero entries of $c$ respectively. 
Informally,  querying $\voracle{c}$ on an index $i\in [N]$
creates a quantum state 
that encodes a two-dimensional slice of $c$ with a fixed first index $i$.
Querying $\poracle{c}$ on a pair of indices $j,k\in [N]$
creates a uniform superposition of all indices  $i\in [N]$ such that $c_{ijk}\ne 0$.
More precisely, let
\[
\calH_{ind}=\CC^{N+1}=\operatorname{span}\{|0\rangle,|1\rangle,\ldots,|N\rangle\}
\]
be an index register, where $|1\ra,\ldots,|N\ra$ label physical indices and $|0\rangle$ is a blank state. 
Let $\calH_{anc}$ be an ancillary work register of $poly(log N)$ qubits and 
$|0_{anc}\ra\in \calH_{anc}$ be the all-zero  state of the work register. 
The size of the ancilla register may depend on the oracle type. 

The oracles $\voracle{c}$ and $\poracle{c}$ are unitary operators acting on
$\calH_{\mathrm{ind}}^{\otimes 3}\otimes \calH_{anc}$.
We require that 
\[
\voracle{c}|i,0,0,0_{anc}\rangle
 =
 \frac{1}{J}\sum_{j,k=1}^N c_{ijk}|i,j,k,0_{anc}\rangle
 + |\mathsf{junk}_{\, i}\ra
\]
for all $i\in [N]$, 
where $|\mathsf{junk}_{\, i}\ra$ is an arbitrary state orthogonal to the subspace with the ancillary register in the state $|0_{anc}\ra$.

Define a tensor 
\[
\chi_{ijk}=\left\{ \ba{rcl}
1 &\mbox{if} & c_{ijk}\ne 0,\\
0 && \mbox{else}.\\
\ea\right.
\]
By assumption, $\sum_{i=1}^N \chi_{ijk}\le s$ for all $j,k\in [N]$, where $s=O(1)$ is the number of drift channels.
We require that
\[
\poracle{c}|0,j,k,0_{anc}\rangle
 =
 \frac{1}{\sqrt{s}}\sum_{i=1}^N \chi_{ijk}|i,j,k,0_{anc}\rangle
 + |\mathsf{junk}_{\, j,k}\rangle
\]
for all $j,k\in [N]$,
where $|\mathsf{junk}_{\, j,k}\ra$ is an arbitrary state orthogonal to the subspace with the ancillary register in the state $|0_{anc}\ra$.

The vector of dissipation rates $\lambda$ is specified by a value oracle $\voracle{\lambda}$ that computes
the dissipation rate $\lambda_i$ for a given index $i$.
This oracle is a unitary operator acting on $\calH_{ind}\otimes \calH_{anc}$
such that  $\voracle{\lambda}|i,0_{anc}\ra = |i,\lambda_i\ra$ for all $i\in [N]$, where $|\lambda_i\ra$ is the binary encoding of $\lambda_i$
with $poly(\log N)$ bits of precision.  For simplicity, we shall ignore the roundoff error in the binary representation of $\lambda_i$.

We shall consider two types of observables: linear observables
\[
u_0(x)=\sum_{i=1}^N b_i x_i
\]
and nonlinear monomial observables
\[
u_0(x)=\prod_{i=1}^N x_i^{d_i}.
\]
In the linear case the vector of coefficients $b\in \RR^N$ must be normalized such that 
$\|b\|\le J_u=O(1)$ and the normalization $J_u$ must be included in the problem input.
In the monomial case we assume that the total degree of $u_0(x)$ is a constant independent of $N$,
that is, $d=\sum_{i=1}^N d_i=O(1)$.
A monomial observable can be specified explicitly by the list of variables and their degrees 
that appear in the monomial. Linear observables are specified by a value oracle
$\voracle{u}$ analogous to the oracle used to describe the drift tensor.
The oracle $\voracle{u}$
 is a unitary operator acting on 
$\calH_{ind} \otimes \calH_{anc}$ such that 
\[
\voracle{u}|0,0_{anc}\ra = \frac1{J_u} \sum_{i=1}^N b_i |i,0_{anc}\ra + |\mathsf{junk}\ra,
\]
where $|\mathsf{junk}\ra$ is an arbitrary state orthogonal to the subspace with the ancillary register in the state $|0_{anc}\ra$.

We allow the quantum algorithm to access the above oracles, their inverses, and their controlled versions. We count a query to  the inverse or the controlled
version of an oracle as a single query. 

Our normalization of the observable $u_0(x)$ ensures that adding
weak noise in the initial condition of the SDE does not lead to a drastic change of the observable value at time $t=0$.
More precisely, consider the variance $\nu \equiv \mathrm{Var}_{z \sim \mu}(u_0(x+z))$, where $z\in \RR^N$ is a random vector sampled from the steady distribution $\mu$
that models noise in the initial condition.
In Section~\ref{sec:init} we prove that $\nu \le O(q\lambda_1^{-1})$ for the linear observable
and  $\nu \le O(q\lambda_1^{-1} \cdot (\|x\| + \sqrt{q/\lambda_1})^{2d-2})$
for the degree-$d$ monomial observable.
Thus the variance
of the observable is independent of $N$ and vanishes in the zero noise limit $q\to 0$.
In contrast, the variance of the full state vector $X(0)=x+z$ can scale linearly with $N$ even if $q$ is small
and $x$ has norm $O(1)$.

%\sergey{Remove this ?}
%Finally, the initial condition $x\in \RR^N$ is specified by a real number $J_x=O(1)$ such that $\|x\|\le J_x$
%and a value oracle $\voracle{x}$ acting on $\CC^N \otimes \calH_{anc}$
%such that
%\[
%\la i,0_{anc} |\voracle{x}|0,0_{anc}\ra = \frac{x_i}{J_x}.
%\]

\section{Divergence-free condition}
\label{sec:div_free}

In this section we derive several implications of the  divergence-free condition Eq.~(\ref{div_free}).
First, let us express this condition in terms of the coefficients $c_{ijk}$ and dissipation rates $\lambda_i$. 
\begin{lemma}
\label{lemma:div_free}
Consider quadratic drift functions $f_i(x)=\sum_{j,k=1}^N c_{ijk} x_j x_k$ with coefficients $c_{ijk}=c_{ikj}$.
The divergence-free condition $\mathrm{div}(\mu(x)f(x))=0$  for all $x\in \RR^N$
is equivalent  to coefficient-wise constraints 
\be
\label{div_free_part1}
\sum_{i=1}^N c_{iik} = 0 \qquad \mbox{for all $k$}
\ee
and 
\begin{equation}
\label{div_free_part2}
 \lambda_i c_{ijk} + \lambda_j c_{jki} + \lambda_k c_{kij} = 0 \qquad \mbox{for all $i,j,k$}.
\end{equation}
Furthermore, suppose all dissipation rates $\lambda_i$ are the same. 
Define coefficients 
\[
b_{ijk}=-b_{kji}=(2/3)(c_{ijk}-c_{kji}).
\]
Then
\be
\label{resummation}
f_i(x)= \sum_{j,k=1}^N b_{ijk} x_j x_k
\ee
for all $i\in [N]$ and all $x\in \RR^N$.
\end{lemma}
\begin{proof}
From $\mathrm{div}(\mu(x)f(x))=0$
one gets
\[
\frac{\mathrm{div}(\mu(x) f(x))}{\mu(x)} = \sum_{i=1}^N \left(  \frac{\partial f_i(x)}{\partial x_i}  -2 q^{-1} \lambda_i x_i f_i(x) \right)=0
\]
for all $x\in \RR^N$. Since the partial derivatives $\partial f_i(x)/\partial x_i$ are linear and
$x_i f_i(x)$ are homogeneous cubic  polynomials, the above implies
\[
\sum_{i=1}^N  \frac{\partial f_i(x)}{\partial x_i} =0 
\]
and
\[
\sum_{i=1}^N  \lambda_i x_i f_i(x)=0
\]
for all $x\in \RR^N$.
Substituting $f_i(x)=\sum_{j,k=1}^N c_{ijk} x_j x_k$ gives
\[
\sum_{i=1}^N \frac{\partial f_i(x)}{\partial x_i}  = \sum_{i,j,k=1}^N c_{ijk} \frac{\partial (x_j x_k)}{\partial x_i} =
2\sum_{i,k=1}^N c_{iik} x_k = 0.
\]
Here the last equality uses the symmetry $c_{ijk}=c_{ikj}$.
The sum $\sum_{i,k=1}^N c_{iik} x_k$ 
vanishes for all $x\in \RR^N$ iff $\sum_{i=1}^N c_{iik}=0$ for all $k$. 
We have proved Eq.~(\ref{div_free_part1}).

Next let us prove Eq.~(\ref{div_free_part2}). 
Let $d_{ijk} = \lambda_i c_{ijk} + \lambda_j c_{jki} + \lambda_k c_{kij}$. We have
\[
\sum_{i=1}^N \lambda_i x_i f_i(x) = \sum_{i,j,k=1}^N \lambda_i c_{ijk} x_i x_j x_k 
=
\frac13  \sum_{i,j,k=1}^N d_{ijk} x_i x_j x_k \equiv P(x).
\]
Clearly, $d_{ijk}=0$ for all $i,j,k$ implies $P(x)=0$ for all $x$.
Using the symmetry $c_{ijk}=c_{ikj}$ one can check that $d_{ijk}$ is invariant under any permutation of indices
$i,j,k$.  Thus
\[
\frac{\partial^3 P(x)}{\partial x_i \partial x_j \partial x_k}=2 d_{ijk}.
\]
Therefore $P(x)=0$ for all $x$ implies $d_{ijk}=0$ for all $i,j,k$.

Now suppose that all dissipation rates are equal. Then
Eq.~\eqref{div_free_part2} reduces to
\[
c_{ijk}+c_{jki}+c_{kij}=0
\qquad \mbox{for all $i,j,k$}.
\]
Since $x_jx_k=x_kx_j$, we get
\[
\sum_{j,k=1}^N b_{ijk}x_jx_k
=
\sum_{j,k=1}^N \frac{b_{ijk}+b_{ikj}}{2}x_jx_k.
\]
Substituting $b_{ijk}=(2/3)(c_{ijk}-c_{kji})$ and $b_{ikj}=(2/3)(c_{ikj}-c_{jki})$ gives
\[
\frac{b_{ijk}+b_{ikj}}{2}
=
\frac{1}{3}
\left[
(c_{ijk}-c_{kji})+(c_{ikj}-c_{jki})
\right].
\]
Using $c_{ijk}=c_{ikj}$ and $c_{kji}=c_{kij}$, this becomes
\[
\frac{b_{ijk}+b_{ikj}}{2}
=
\frac{1}{3}
\left(
2c_{ijk}-c_{kij}-c_{jki}
\right).
\]
By assumption,
$c_{ijk}+c_{jki}+c_{kij}=0$.
Thus
$-c_{kij}-c_{jki}=c_{ijk}$.
Therefore
\[
\frac{b_{ijk}+b_{ikj}}{2}
=
\frac{1}{3}(2c_{ijk}+c_{ijk})
=
c_{ijk}
\]
proving Eq.~(\ref{resummation}).
\end{proof}
{\em Comment:} Choosing $i=j=k$ in Eq.~(\ref{div_free_part2})
implies that $c_{iii}=0$ for all $i$. However, some coefficients $c_{ijk}$
with repeated indices may be nonzero.
For example, suppose $N=3$, $\lambda_1=\lambda_2=\lambda_3=1$, and
the drift functions are
\[
f_1(x)=x_1x_2,\qquad
f_2(x)=-x_1^2+x_3^2,\qquad
f_3(x)=-x_2x_3.
\]
They obey the divergence-free condition Eq.~(\ref{div_free}). However, the tensor $c$
has nonzero entries with repeated indices such as $c_{112}$ or $c_{233}$.

In the rest of this section we use the coefficient-wise version of the divergence-free condition to prove the following
upper bounds. 

\begin{prop}
\label{prop:coeff_upper_bound}
Let $J_{ijk}$ be the maximum of $|c_{ijk}|$ over all permutations of the three indices.
If the system is divergence-free then for all $i,j,k$ 
\be
\label{coeff_upper_bound}
|c_{ijk}|  \sqrt{\frac{\lambda_i}{\lambda_j\lambda_k}} \le J_{ijk} \sqrt{\frac{2}{\lambda_1}}
\ee
and
\be
\label{coeff_upper_comm}
|c_{ijk}| \frac1{\lambda_j \lambda_k}
|\lambda_j+\lambda_k-\lambda_i| \le \frac{2J_{ijk}}{ \lambda_1}.
\ee
\end{prop}
\begin{proof}
Let us prove Eq.~(\ref{coeff_upper_bound}).
First, we have an obvious bound
\[
\eta\equiv 
|c_{ijk}|  \sqrt{\frac{\lambda_i}{\lambda_j\lambda_k}} 
 \le J_{ijk} \sqrt{\frac{\lambda_i}{\lambda_j\lambda_k}}\equiv \eta_1.
\]
By Lemma~\ref{lemma:div_free}, the divergence-free condition implies
$\lambda_i c_{ijk} + \lambda_j c_{jki} + \lambda_k c_{kij}=0$.  Thus
\be
\label{lambda_i_cijk}
\lambda_i |c_{ijk}| \le \lambda_j |c_{jki}| + \lambda_k |c_{kij}| \le  (\lambda_j+\lambda_k) J_{ijk}.
\ee
This gives the  second upper bound
\[
\eta
= \frac1{\sqrt{\lambda_i\lambda_j\lambda_k}} |c_{ijk}| \lambda_i
\le   \frac1{\sqrt{\lambda_i\lambda_j\lambda_k}} (\lambda_j+\lambda_k) J_{ijk}
=
J_{ijk} \left(  \sqrt{\frac{\lambda_j}{\lambda_i\lambda_k}} +  \sqrt{\frac{\lambda_k}{\lambda_i\lambda_j}}\right)
  \equiv \eta_2.
\]
From $\eta\le \eta_1$ and $\eta\le \eta_2$ one gets
\[
\eta \le \sqrt{\eta_1 \eta_2} = J_{ijk}\sqrt{\frac1{\lambda_j} + \frac1{\lambda_k}} \le J_{ijk} \sqrt{\frac{2}{\lambda_1}}.
\]
This proves Eq.~(\ref{coeff_upper_bound}).

Next let us prove Eq.~(\ref{coeff_upper_comm}).
Consider two cases.

\noindent
{\em Case 1:} $\lambda_i\le \lambda_j + \lambda_k$. Then $|\lambda_j + \lambda_k - \lambda_i| \le \lambda_j + \lambda_k$ and
\[
|c_{ijk}| \frac1{\lambda_j \lambda_k}
|\lambda_j+\lambda_k-\lambda_i| \le
J_{ijk} \frac{\lambda_j+\lambda_k}{\lambda_j\lambda_k}
=J_{ijk}\left(\frac1{\lambda_j} + \frac1{\lambda_k}\right) \le  \frac{2J_{ijk}}{ \lambda_1}.
\]

\noindent
{\em Case 2:} $\lambda_i> \lambda_j + \lambda_k$. 
Then $|\lambda_j + \lambda_k - \lambda_i| \le \lambda_i$.
Using Eq.~(\ref{lambda_i_cijk}) one gets
\[
|c_{ijk}| \, |\lambda_j+\lambda_k-\lambda_i|  \le |c_{ijk}| \lambda_i \le J_{ijk}(\lambda_j+\lambda_k).
\]
The rest of the proof is the same as in Case~1.
\end{proof}

\section{Kolmogorov equation}
\label{sec:Kolmogorov}
Time evolution of the expected value $u(t,x)=\EE_W u_0(X(t))$ with $X(0)=x$
 is 
governed by the backward Kolmogorov  equation (bKE)~\cite{da2004kolmogorov}.
 Assuming that $u(t,x)$ is a sufficiently smooth function, the Kolmogorov equation reads as
\be
\label{Kolmogorov1}
\frac{\partial}{\partial t} u(t,x) = \sum_{i=1}^N (-\lambda_i x_i + f_i(x)) \frac{\partial}{\partial x_i} u(t,x)
+\frac{q}2 \sum_{i=1}^N  \frac{\partial^2}{\partial x_i^2} u(t,x)
\ee
for $t\ge 0$ with the initial condition $u(0,x)=u_0(x)$.  
The inverse mapping from the partial differential equation Eq.~(\ref{Kolmogorov1}) to an SDE is known as Feynman–Kac formula~\cite{oksendal2013stochastic}.

We note that formally the Gaussian measure $\mu(x)\sim \exp{\left[-q^{-1}\sum_{i=1}^N \lambda_i x_i^2\right]}$ is the invariant measure of Eq.~(\ref{Kolmogorov1}) in the sense that
\be
\label{steady_state}
\frac{d}{dt} \int_{\RR^N} dx\, \mu(x) u(t,x) =0.
\ee
Indeed, we have
\begin{align*}
\frac{d}{dt} \int_{\RR^N} dx\, \mu(x) u(t,x) & =
\int_{\RR^N} dx\, \mu(x) \left[  \sum_{i=1}^N (-\lambda_i x_i + f_i(x)) \frac{\partial}{\partial x_i} u(t,x)
+\frac{q}2 \sum_{i=1}^N  \frac{\partial^2}{\partial x_i^2}u(t,x) \right]\\
&= \int_{\RR^N} dx\, u(t,x) \left[  \sum_{i=1}^N \ \frac{\partial}{\partial x_i} \left[\mu(x)(\lambda_i x_i - f_i(x))\right]
+\frac{q}2 \sum_{i=1}^N  \frac{\partial^2}{\partial x_i^2} \mu(x)\right]=0.
\end{align*}
Here we used the integration by parts, the divergence-free condition $\mathrm{div}(\mu(x)f(x))=0$, and assumed that $u(t,x)$ and its 1st and 2nd derivatives are square-integrable w.r.t. $\mu$.
As a consequence, the mean value of $u(t,x)$ with respect to $\mu$ does not depend on time,
suggesting that $\mu$ is the invariant measure of the underlying SDE.
However, to make this statement rigorous one needs to show that Eq.~(\ref{Kolmogorov1}) has unique solution for appropriately chosen $u_0$. 
%One could also show that %\[
%\int_{\RR^N} dx\, \mu(x) u(t,x) =\int_{\RR^N} dx\, \mu(x) u_0(x)
%\]

\subsection[Mapping to quantum harmonic oscillators]
{Mapping to quantum harmonic oscillators\protect\footnote{This section is largely based on the preprint~\cite{bravyi2025quantum}.}}

We shall construct  a solution of Eq.~(\ref{Kolmogorov1}) in the form
\be
u(t,x)= \frac{\psi(t,x)}{
\sqrt{\mu(x)}},
\ee
where $\psi(t,x)$ is a function satisfying $\int_{\RR^N} dx\,  |\psi(t,x)|^2 <\infty$\footnote{This ansatz ensures that the solution $u(t,x)$ belongs to a 
Gauss-Sobolev space~\cite{da2004kolmogorov} associated with the measure $\mu(x)$, 
that is, $\int_{\RR^N} dx\, \mu(x) |u(t,x)|^2 <\infty$. }
A simple algebra shows that Eq.~(\ref{Kolmogorov1}) is equivalent to 
\be
\label{Kolmogorov1a}
\frac{\partial}{\partial t} \psi(t,x)
=-\frac12 \sum_{i=1}^N \left( q^{-1}\lambda_i^2  x_i^2 -q \frac{\partial^2}{\partial x_i^2} - \lambda_i\right)\psi(t,x)
+ \sum_{i=1}^N f_i(x) ( q^{-1} \lambda_i x_i + \frac{\partial}{\partial x_i}) \psi(t,x).
\ee
The term $( q^{-1}\lambda_i^2  x_i^2 -q \frac{\partial^2}{\partial x_i^2} - \lambda_i)/2$ is the Hamiltonian of
a quantum harmonic oscillator, up to an overall energy shift. 
The oscillator has mass $q^{-1}$, spring constant
$q^{-1}\lambda_i^2$, and frequency $\lambda_i$. 
It will be convenient to describe this Hamiltonian by the  creation ($a_i^\dag$) and annihilation ($a_i$) operators
\be
a_i = \frac1{\sqrt{2}}\left( \sqrt{\frac{\lambda_i}{q}} x_i + \sqrt{\frac{q}{\lambda_i}} \frac{\partial}{\partial x_i}\right) \quad \mbox{and} \quad 
a_i^\dag= \frac1{\sqrt{2}}\left( \sqrt{\frac{\lambda_i}{q}} x_i - \sqrt{\frac{q}{\lambda_i}} \frac{\partial}{\partial x_i}\right).
\ee
Then\footnote{Here we use the identity $(x_i\partial_{x_i}-\partial_{x_i}x_i)\psi = -\psi$.} $( q^{-1}\lambda_i^2  x_i^2 -q \frac{\partial^2}{\partial x_i^2} - \lambda_i)/2=\lambda_i a_i^\dag a_i$.
Using the identity
 $q^{-1} \lambda_i x_i + \frac{\partial}{\partial x_i}=a_i\sqrt{2\lambda_i/q}$ one can rewrite Eq.~(\ref{Kolmogorov1a}) as 
\be
\label{Kolmogorov1b}
\frac{\partial}{\partial t} \psi(t,x) = \left( -\sum_{i=1}^N \lambda_i a_i^\dag a_i  +   \sum_{i=1}^N \sqrt{\frac{2\lambda_i}{q}} f_i(x) a_i\right) \psi(t,x).
\ee
We shall work in the second quantization picture and 
express the wave function $\psi(t,x)$ of $N$ harmonic oscillators  in the Fock (occupation number) basis. 
Let $\ZZ_{\ge 0}$ be the set of nonnegative integers,  $m_i \in \ZZ_{\ge 0}$ be the occupation number of the $i$-th oscillator,
and $\mm = (m_1,\ldots,m_N)$ be a multi-index composed of $N$ occupation numbers. 
Let $|\mm\ra$ be the Fock basis vector
 corresponding to a multi-index $\mm \in \ZZ_{\ge 0}^N$.
These vectors form an orthonormal basis,  i.e. $\la \mm|\nn\ra=\delta_{\nn,\mm}$ for all multi-indices $\mm,\nn$, and 
\be
\label{ai_action}
a_i |\mm\ra = \left\{
\ba{rcl} 
\sqrt{m_i} |\mm-e^i\ra &\mbox{if} & m_i\ge 1,\\
0 &\mbox{if} & m_i=0.\\
\ea
\right.
\ee
Here $e^i\in \ZZ_{\ge 0}^N$ is a multi-index such that $(e^i)_i=1$ and $(e^i)_j=0$ for $j \ne i$.
Accordingly, $a_i^\dag |\mm\ra = \sqrt{m_i+1} |\mm+e^i\ra$ and
$a_i^\dag a_i |\mm\ra = m_i |\mm\ra$. Let $|{\bf 0}\ra=|00\ldots0\ra$ be the vacuum state
annihilated by all operators $a_i$. Recall that 
the Fock space $\calF$ is a Hilbert space formed by linear combinations $|\psi\ra = \sum_{\mm \in \ZZ_{\ge 0}^N}\;  \psi_\mm |\mm\ra$
with complex coefficients $\psi_\mm$ satisfying $\sum_{\mm\in \ZZ_{\ge 0}^N}\;  |\psi_\mm|^2 <\infty$. The inner product in $\calF$ is defined as
$\la \psi|\phi\ra = \sum_{\mm\in \ZZ_{\ge 0}^N}\; \psi_\mm^* \phi_\mm$.
The second-quantized version of Eq.~(\ref{Kolmogorov1b}) becomes
\be
\label{Kolmogorov1c}
\frac{d}{dt}|\psi(t)\ra  =  \left( -\sum_{i=1}^N \lambda_i a_i^\dag a_i  +  \sum_{i=1}^N   \sqrt{2\lambda_i/q}f_i(x) a_i\right)|\psi(t)\ra,
\ee
where $|\psi(t)\ra \in \calF$ 
and $f_i(x)$ is expressed in terms of creation/annihilation operators using the identity
\be
\label{x_vs_a}
x_i = \sqrt{q/(2\lambda_i)} (a_i+a_i^\dag)
\ee
The correspondence between the first and the second quantization pictures is established by $\psi(t,x)=\sum_{\mm\in \ZZ_{\ge 0}^N}\;  \la \mm|\psi(t)\ra \Phi_\mm(x)$,
where $\Phi_\mm(x)$ are the eigenfunctions for the considered system of harmonic oscillators. An explicit expression for $\Phi_\mm(x)$  
in terms of Hermite polynomials will be needed only later in Section~\ref{sec:readout}.
Recalling that  $f_i(x)=\sum_{j,k=1}^N c_{ijk} x_j x_k$ 
and expressing each variable $x_i$ in terms of creation/annihilation operators using Eq.~(\ref{x_vs_a})
one arrives at 
\be
\label{Kolmogorov1d}
\frac{d}{dt}|\psi(t)\ra  =  (-A+C)|\psi(t)\ra,
\ee
where
\be
\label{Aexplicit}
A = \sum_{i=1}^N \lambda_i a_i^\dag a_i
\ee
and
\be
\label{Cexplicit}
C = (q/2)^{1/2}\sum_{i,j,k=1}^N   (\lambda_i\lambda_j  \lambda_k)^{-1/2}  \lambda_i  c_{ijk} (a_j+a_j^\dag) (a_k+a_k^\dag) a_i.
\ee
Obviously, $A$ is hermitian and positive semi-definite. 
\begin{lemma}
\label{lemma:anti_hermitian}
Suppose the system is divergence-free. Then $C$ is anti-hermitian and one can write
\be
\label{Csimplified}
C=
(q/2)^{1/2}\sum_{i,j,k=1}^N   (\lambda_i\lambda_j  \lambda_k)^{-1/2}   \lambda_i  c_{ijk} (a_j^\dag a_k^\dag a_i -  a_i^\dag a_k a_j).
\ee
\end{lemma}
\begin{proof}
Define operators $Q_j = a_j + a_j^\dag$.
From Eq.~(\ref{Cexplicit}) one gets
\[
C = (q/2)^{1/2}\sum_{i,j,k=1}^N   (\lambda_i\lambda_j  \lambda_k)^{-1/2}  \lambda_i  c_{ijk} Q_j Q_k a_i
\]
and
\[
C^\dag = (q/2)^{1/2}\sum_{i,j,k=1}^N   (\lambda_i\lambda_j  \lambda_k)^{-1/2}  \lambda_i  c_{ijk} a_i^\dag Q_j Q_k.
\]
Here we noted that $Q_j$ and $Q_k$ are hermitian and $Q_jQ_k=Q_k Q_j$.
Canonical commutation rules imply that
$[a_i^\dag,Q_j] = -\delta_{i,j} I$.  
Thus
\[
a_i^\dag  Q_j Q_k =
[a_i^\dag,Q_j] Q_k + Q_j [a_i^\dag, Q_k] + Q_j Q_k a_i^\dag
=
 -\delta_{i,j}Q_k -\delta_{i,k}Q_j
+Q_j Q_k a_i^\dag.
\]
Substitute this into the above formula for $C^\dag$ and use the coefficient-wise
divergence free condition $\sum_{i=1}^N c_{iik}=0$ of Lemma~\ref{lemma:div_free}.
The symmetry of $c$ implies $\sum_{i=1}^N c_{iik}=0$.
Then the terms proportional to $\delta_{i,j}c_{ijk} Q_k$ and $\delta_{i,k}c_{ijk}Q_j$
vanish after summing over $i$.
This gives
\[
C^\dag = (q/2)^{1/2}\sum_{i,j,k=1}^N   (\lambda_i\lambda_j  \lambda_k)^{-1/2}   \lambda_i  c_{ijk} 
Q_j Q_k a_i^\dag.
\]
Using the identity $a_i^\dag = -a_i  + Q_i$ one gets 
\[
C^\dag = -C +T,
\]
where
\[
T =  (q/2)^{1/2}\sum_{i,j,k=1}^N   (\lambda_i\lambda_j  \lambda_k)^{-1/2}   \lambda_i  c_{ijk} Q_j Q_k Q_i.
\]
Since the operators $Q_i,Q_j,Q_k$ pairwise commute, one can replace $\lambda_i  c_{ijk}$ inside the sum
by the symmetrized expression $(1/3)(\lambda_i  c_{ijk} + \lambda_j c_{jki} + \lambda_k c_{kij})$ which vanishes
due to the divergence-free condition Eq.~(\ref{div_free_part2}).
This  implies $T=0$, that is, $C=-C^\dag$.

Let us prove the identity Eq.~(\ref{Csimplified}). Write
\[
C=\sum_{p\in \ZZ} C_p,
\]
where $C_p$ is the sum of all terms in Eq.~(\ref{Cexplicit}) that change the particle number by $p$.
From $C^\dag = -C$ one gets $C_{-p}=-C_p^\dag$ for all $p$.
Since any creation and annihilation operator changes the particle number by $+1$
and $-1$ respectively, Eq.~(\ref{Cexplicit}) gives
$C_p=0$ unless $p\in \{-3,-1,1\}$.
From $C_3=0$ one gets $C_{-3}=-C_3^\dag =0$.
The only terms in Eq.~(\ref{Cexplicit}) that change the particle number by $+1$ are those with two 
creation and one annihilation operator. Collecting all such terms gives
\[
C_1 =  (q/2)^{1/2}\sum_{i,j,k=1}^N   (\lambda_i\lambda_j  \lambda_k)^{-1/2}  \lambda_i  c_{ijk} a_j^\dag a_k^\dag a_i
\]
and $C_{-1}=-C_1^\dag$.
Thus $C = C_1 - C_1^\dag$ which 
proves Eq.~(\ref{Csimplified}).
\end{proof}

Our goal is to solve the Kolmogorov equation in the form Eq.~(\ref{Kolmogorov1d}). Since the vacuum state $|{\bf 0}\ra$ is annihilated by all
terms in the righthand side of Eq.~(\ref{Kolmogorov1d}), time evolution of the vacuum amplitude  is trivial, that is,
$\la {\bf 0}|\psi(t)\ra=\la {\bf 0}|\psi(0)\ra$ for all $t$.  Instead of keeping track of  this constant amplitude in all equations, we find it convenient
to work in the subspace of the Fock space $\calF$  orthogonal to the vacuum. Let this subspace be
\be
\calF_0 = \{ |\psi\ra \in \calF \, : \,  \la {\bf 0}|\psi\ra=0 \}.
\ee
Note that $\la \psi|A|\psi\ra \ge \lambda_1 \la \psi|\psi\ra$ for all $|\psi\ra \in \calF_0$.
We assume that the initial state $|\psi(0)\ra$ belongs to $\calF_0$. 
This can always be achieved by subtracting a multiple of $|{\bf 0}\ra$ from $|\psi(0)\ra$.

\subsection{Key technical lemmas}
\label{sec:technical}

In this  section we establish structural properties of the operators $A=\sum_{i=1}^N \lambda_i a_i^\dag a_i$ and $C$
that appear in the second-quantized version of the Kolmogorov equation. Our key technical lemmas are as follows (we postpone the proof until the end of this section).

\begin{lemma}
\label{lemma:AC1}
For any vectors $\phi,\psi\in\calF_0$ one has
\be
\label{C_overlap_upper}
|\la \phi|C|\psi\ra| \le \gamma  \sqrt{\la \psi|A|\psi\ra \la \phi|A^2|\phi\ra}.
\ee
where
\[
\gamma = \frac{3J \sqrt{qs}}{\lambda_1^2}.
\]
\end{lemma}

\begin{lemma}
\label{lemma:AC2}
For any integer \(p\ge 1\) and any vector \(\psi\in\calF_0\), one has
\be
\label{commutator_bound_general}
\left|\langle \psi|[A^p,C]|\psi\rangle\right|
\le
\gamma_p
\sqrt{
\langle \psi|A^p|\psi\rangle
\langle \psi|A^{p+1}|\psi\rangle
},
\ee
where
\[
\gamma_p
=\frac{2p\sqrt{6qs}\,J}{\lambda_1}
\left(1+\frac{2\lambda_N}{\lambda_1}\right)^{(p-1)/2}.
\]
In particular, $\gamma_1$ does not depend on the largest dissipation rate $\lambda_N$.
\end{lemma}

{\em Comments:} Recall that $\lambda_1$ is the smallest and $\lambda_N$ is the largest dissipation rate. 
We only consider vectors $\psi,\phi$ such that the righthand side of Eqs.~(\ref{C_overlap_upper},\ref{commutator_bound_general}) is finite.
Note that $[A^p,C]$ is a  hermitian operator since $A^p$ is hermitian while $C$ is anti-hermitian. Thus $\la \psi| [A^p,C]|\psi\ra$  is a real number.
The bound Eq.~(\ref{commutator_bound_general}) with $p\ge 2$  is rather weak as it depends on the largest dissipation rate $\lambda_N$, which often scales as a positive power of $N$.
Luckily, we need this bound with $p\ge 2$ only to prove that the Kolmogorov equation has a well-defined solution.
The scaling of $\gamma_p$ with $p\ge 2$  is irrelevant for the runtime of our quantum algorithm, which depends on  $\gamma$ and $\gamma_1$.
We shall use the following corollary of Lemma~\ref{lemma:AC2}.
\begin{corol}
\label{corol:commutator_general}
\label{corol:commutator_bound_order1}
\label{corol:commutator_bound_order_p}
For any integer \(p\ge 1\) and any vector \(\psi\in\calF_0\), one has
\[
\left| \langle \psi|[A^p,C]|\psi\rangle\right|
\le
\langle \psi|A^{p+1}|\psi\rangle
+
\kappa_p\langle \psi|A^p|\psi\rangle,
\]
where
\[
\kappa_p=\frac{\gamma_p^2}{4}
=
\frac{6p^2J^2qs}{\lambda_1^2}
\left(1+\frac{2\lambda_N}{\lambda_1}\right)^{p-1}.
\]
In particular, $\kappa_1$ does not depend on the largest dissipation rate.
\end{corol}
\begin{proof}
For any \(\beta>0\) one has
$\sqrt{\alpha_1\alpha_2}
\le
\frac12(\beta\alpha_2+\beta^{-1}\alpha_1)$.
Substituting 
$\alpha_1=\langle\psi|A^p|\psi\rangle$,
$\alpha_2=\langle\psi|A^{p+1}|\psi\rangle$,
 $\beta=2/\gamma_p$ and using Eq.~(\ref{commutator_bound_general})
gives
$|\langle \psi|[A^p,C]|\psi\rangle|
\le
\langle\psi|A^{p+1}|\psi\rangle
+ (\gamma_p^2/4)
\langle\psi|A^p|\psi\rangle$.
\end{proof}

Finally, we shall need the following very loose upper bound.
\begin{lemma}
\label{lemma:loose}
Let $p\ge 0$ be a real number. 
 There exist a real number $\omega_p<\infty$ such that  
\be
\label{ABCloose}
\| A^p C\psi\| \le \omega_p \| A^{p+3/2}\psi\|
\ee
for any vector $\psi \in \calF_0$. Here $\omega_p$ depends on $p$, $N$, $J$, $\lambda_1$, and $\lambda_N$. 
\end{lemma}

\subsection{Proof strategy}

The proof of Lemma~\ref{lemma:AC1} and~\ref{lemma:AC2} relies on the following
\(A\)-weighted bounds for cubic bosonic operators.
We shall need two versions of such bounds, for operators that change the particle number by $+1$ and $-1$.
\begin{lemma}
\label{lemma:master}
Consider operators
\be
\label{def_D}
D = \sum_{i,j,k=1}^N d_{ijk} (\lambda_i \lambda_j \lambda_k)^{1/2}  a_j^\dag a_k^\dag a_i
\ee
and
\be
\label{def_E}
E = \sum_{i,j,k=1}^N e_{ijk} a_i^\dag a_k a_j
\ee
where $d_{ijk}$ and $e_{ijk}$ are complex coefficients. Suppose for any fixed $i$ one has
\[
\left( \sum_{j,k=1}^N |d_{ijk}|^2 \right)^{1/2} \le J_d  \quad \mbox{and} \quad \left(\sum_{j,k=1}^N |e_{ijk}|^2 \right)^{1/2} \le J_e.
\]
Suppose for any fixed pair $j,k$ there are at most $s$ nonzero entries $d_{ijk}$ and at most $s$ nonzero entries $e_{ijk}$.
Then for any vectors $\phi,\psi\in\calF_0$ one has
\be
\label{master_upper_D}
|\la \phi|D|\psi\ra| \le J_d \sqrt{s} \sqrt{\la \psi|A|\psi\ra \la \phi|A^2|\phi\ra}
\ee
and
\be
\label{master_upper_E}
|\la \phi|E|\psi\ra| \le J_e \sqrt{s} \, \lambda_1^{-3/2}
 \sqrt{\la \psi|A|\psi\ra \la \phi|A^2|\phi\ra}.
\ee
\end{lemma}
We can  apply Lemma~\ref{lemma:master}  to obtain $A$-weighted bounds for  $C$
and $[A,C]$.
Indeed, by Lemma~\ref{lemma:anti_hermitian}, 
\[
C=C^\uparrow-C^\downarrow,
\]
where
\[
C^\uparrow=
(q/2)^{1/2}\sum_{i,j,k=1}^N c_{ijk}   \sqrt{\frac{\lambda_i}{\lambda_j\lambda_k}}
a_j^\dag a_k^\dag a_i
\]
and $C^\downarrow=(C^\uparrow)^\dag$.
Since $A$ is diagonal in the Fock basis, it follows that 
$C^\uparrow$ and $[A,C^\uparrow]$
are sums of bosonic cubic terms that change
the particle number by $+1$. We shall bound matrix elements of $C^\uparrow$ and $[A,C^\uparrow]$
using Eq.~(\ref{master_upper_D}) of Lemma~\ref{lemma:master}.
Likewise, $C^\downarrow$ and $[A,C^\downarrow]$ are sums of bosonic cubic terms that change
the particle number by $-1$. We shall  bound matrix elements of $C^\downarrow$ and $[A,C^\downarrow]$
using Eq.~(\ref{master_upper_E}) of Lemma~\ref{lemma:master}.
The commutator bound Eq.~(\ref{commutator_bound_general}) with $p\ge 2$ is proved using similar arguments. 
Finally, Lemma~\ref{lemma:loose} is proved separately by a direct, crude estimate.

In the rest of the section we prove all above lemmas, starting from Lemma~\ref{lemma:master}.

\subsection{Proof of Lemma~\ref{lemma:master}}

The definition of creation and annihilation operators implies that 
 $\la \nn|a_j^\dag a_k^\dag a_i|\mm\ra=0$ unless $\nn=\mm-e^i + e^j +e^k$.
 Let $\calE_{ijk}\subseteq \calJ \times \calJ$ be the set of all pairs $(\nn,\mm)$ such that 
$\nn=\mm -e^i + e^j + e^k$.
Thus  $\la \nn|a_j^\dag a_k^\dag a_i|\mm\ra=0$ unless $(\nn,\mm)\in \calE_{ijk}$.
For any $\mm$ there is at most one $\nn$ such that $(\nn,\mm)\in \calE_{ijk}$ and vice verse.
Nonzero matrix elements of $a_j^\dag a_k^\dag a_i$ obey
 \be
\label{aaa_bound}
|\la \nn |a_j^\dag a_k^\dag a_i |\mm\ra|\le 
\| a_k a_j |\nn\ra\| \cdot \| a_i |\mm\ra\| \le
(n_j n_k)^{1/2} m_i^{1/2}.
\ee
Let $\chi_{ijk}\in \{0,1\}$ be the indicator function of $d_{ijk}$ such that $d_{ijk}=d_{ijk}\chi_{ijk}$ for all $i,j,k$.
The triangle inequality and Eq.~(\ref{aaa_bound}) give
\[
|\la \phi |D|\psi\ra |\le  \sum_{i,j,k=1}^N\;  \sum_{(\nn,\mm)\in \calE_{ijk}}\; 
 x_{ijk\mm} y_{ijk\nn},
\]
where
\[
x_{ijk\mm} =  |d_{ijk}| (\lambda_i m_i)^{1/2} |\psi_\mm|  \quad \mbox{and} \quad y_{ijk\nn} =   \chi_{ijk} (\lambda_j n_j \lambda_k n_k)^{1/2} 
 |\phi_{\nn}|.
\]
Cauchy-Schwarz inequality gives
\be
\label{lemmaAC1_eq1}
|\la \phi |D|\psi\ra | \le 
\sqrt{  \sum_{i,j,k=1}^N  \; \sum_{(\nn,\mm)\in \calE_{ijk}} \; x_{ijk\mm}^2}
\cdot \sqrt{   \sum_{i,j,k=1}^N  \; \sum_{(\nn,\mm)\in \calE_{ijk}} \; y_{ijk\nn}^2}.
\ee
Let us upper bound the sum of $x_{ijk\mm}^2$.
Since for any $\mm$ there is at most one $\nn$ such that $(\nn,\mm)\in \calE_{ijk}$, one has
\[
 \sum_{i,j,k=1}^N  \; \sum_{(\nn,\mm)\in \calE_{ijk}} \; x_{ijk\mm}^2 
 =  \sum_{i,j,k=1}^N  \; \sum_{\mm \in \calJ}\;
 x_{ijk\mm}^2 =
   \sum_{i,j,k=1}^N  \; \sum_{\mm \in \calJ}\;
   |d_{ijk}|^2 \lambda_i m_i |\psi_\mm|^2.
 \]
For any fixed index $i$ one has $\sum_{j,k=1}^N |d_{ijk}|^2 \le J_d^2$.
Hence
\be
\label{lemmaAC1_eq2}
 \sum_{i,j,k=1}^N  \; \sum_{(\nn,\mm)\in \calE_{ijk}} \; x_{ijk\mm}^2
  \le 
J_d^2 \sum_{\mm \in \calJ} \; \sum_{i=1}^N  \;  m_i \lambda_i|\psi_\mm|^2
 =J_d^2 \la \psi| A |\psi \ra.
\ee

Let us upper bound the sum of $y_{ijk\mm}^2$.
Since for any $\nn$ there is at most one $\mm$ such that $(\nn,\mm)\in \calE_{ijk}$,
one has
\[
\sum_{i,j,k=1}^N  \; \sum_{(\nn,\mm)\in \calE_{ijk}} \; y_{ijk\nn}^2
\le  \sum_{i,j,k=1}^N  \;  \sum_{\nn \in \calJ} \; y_{ijk\nn}^2
= \sum_{i,j,k=1}^N  \;  \sum_{\nn \in \calJ} \;
 \chi_{ijk} \lambda_j n_j \lambda_k n_k 
 |\phi_{\nn}|^2.
\]
For any fixed pair of indices $j,k$ one has $\sum_{i=1}^N \chi_{ijk}\le s_d$.
Thus
\be
\label{lemmaAC1_eq3}
\sum_{i,j,k=1}^N  \; \sum_{(\nn,\mm)\in \calE_{ijk}} \; y_{ijk\nn}^2
 \le s_d \sum_{j,k=1}^N \; \sum_{\nn \in \calJ}\; |\phi_\nn|^2 \lambda_j n_j \lambda_k n_k = s_d  \la \phi|A^2|\phi\ra.
\ee
Substituting Eqs.~(\ref{lemmaAC1_eq2},\ref{lemmaAC1_eq3}) into Eq.~(\ref{lemmaAC1_eq1}) gives
\[
|\la \phi |D|\psi\ra | \le J_d \sqrt{s_d}  \sqrt{\la \phi|A^2|\phi\ra \la \psi|A|\psi\ra }.
\]
This proves Eq.~(\ref{master_upper_D}).

Next let us upper bound $|\la \phi |E|\psi\ra |$. 
For each integer $w\ge 0$ let $P_w$ be the projector onto the subspace of the Fock space $\calF_0$ with exactly $w$ particles,
\[
P_w = \sum_{\mm \in \calJ\, : \, |\mm|=w}\; |\mm\ra\la \mm|.
\]
Here $|\mm|=\sum_{i=1}^N m_i$. Note that $P_0=0$ since $\calF_0$ excludes the vacuum state $|{\bf 0}\ra$.
Thus $I=\sum_{w\ge 1} P_w$ is the identity decomposition on $\calF_0$. Define vectors
\[
|\phi_w\ra = P_w |\phi\ra \quad \mbox{and} \quad |\psi_w\ra = P_w |\psi\ra.
\]
Then $\sum_{w\ge 1} |\phi_w\ra=|\phi\ra$ and $\sum_{w\ge 1} |\psi_w\ra=|\psi\ra$.
Since $E$ decreases the particle number by one, 
we have $\la \phi_{w'} |E |\psi_w\ra=0$ unless $w'=w-1$. Thus
\[
\la \phi|E|\psi\ra =  \sum_{w\ge 2} \la \phi_{w-1} |E|\psi_w\ra.
\]
We have $\la \nn|a_i^\dag a_k a_j|\mm\ra=0$ unless 
$(\mm,\nn)\in \calE_{ijk}$.
Nonzero matrix elements of $a_i^\dag a_k a_j$  obey
 \be
\label{aaa_bound_dag}
|\la \nn | a_i^\dag a_k a_j  |\mm\ra|\le 
\| a_i |\nn\ra\| \cdot \| a_k a_j |\mm\ra\| \le
n_i^{1/2} (m_j (m_k-\delta_{j,k}))^{1/2}.
\ee
Let $\chi_{ijk}\in \{0,1\}$ be the indicator function of $e_{ijk}$ such that $e_{ijk}=e_{ijk}\chi_{ijk}$ for all $i,j,k$.
The triangle inequality gives
\[
|\la \phi |E |\psi\ra |\le 
\sum_{w\ge 2}\; 
 \sum_{i,j,k=1}^N  \;
 \sum_{(\mm,\nn)\in \calE_{ijk}}\;
\; x_{wijk\mm} y_{wijk\nn},
\]
where
\[
x_{ijk\mm} =\delta_{w,|\mm|}
\chi_{ijk} (m_j (m_k-\delta_{j,k}))^{1/2}  |\psi_\mm| 
\quad \mbox{and} \quad y_{ijk\nn} = \delta_{w-1,|\nn|}|e_{ijk}| n_i^{1/2}  |\phi_{\nn}|.
\]
Applying Cauchy-Schwarz inequality for each value of $w$ gives
\be
\label{lemmaAC1_eq8}
|\la \phi |E |\psi\ra | \le  \sum_{w\ge 2}
\sqrt{  \sum_{i,j,k=1}^N \;     \sum_{(\mm,\nn)\in \calE_{ijk}} \; x_{wijk\mm}^2}
\cdot \sqrt{ \sum_{i,j,k=1}^N \;     \sum_{(\mm,\nn)\in \calE_{ijk}} \;  y_{wijk\nn}^2}.
\ee
Let us upper bound the sum of $x_{wijk\mm}^2$.
For any $\mm$ there is at most one $\nn$ such that $(\mm,\nn)\in \calE_{ijk}$.
Thus
\[
 \sum_{i,j,k=1}^N \;  \sum_{(\mm,\nn)\in \calE_{ijk}} \; x_{wijk\mm}^2
 \le  \sum_{i,j,k=1}^N \;  \sum_{\mm \in \calJ}  \; x_{wijk\mm}^2
 = \sum_{i,j,k=1}^N \;  \sum_{\mm \in \calJ}  \; 
 \delta_{w,|\mm|}
\chi_{ijk} m_j (m_k-\delta_{j,k})   |\psi_\mm|^2.
\]
Since any one-dimensional slice of $e$ has at most $s_e$ nonzeros, 
for any fixed pair of indices $j$ and $k$ one has
$\sum_{i=1}^N \chi_{ijk} \le s_e$.
It follows that 
\[
 \sum_{i,j,k=1}^N \;  \sum_{(\mm,\nn)\in \calE_{ijk}} \; x_{wijk\mm}^2
 \le 
 s_e\sum_{j,k=1}^N \;  \sum_{\mm \in \calJ}  \; 
 \delta_{w,|\mm|}
 m_j (m_k -\delta_{j,k})   |\psi_\mm|^2.
\]
For any $\mm$ with $|\mm|=w$ one has
$\sum_{j,k=1}^N  m_j (m_k -\delta_{j,k})  = w^2 - w = w(w-1)$. Thus
\be
\label{Ebound_part1}
 \sum_{i,j,k=1}^N \;  \sum_{(\mm,\nn)\in \calE_{ijk}} \; x_{wijk\mm}^2
 \le s_e w(w-1) \la\psi|P_w|\psi\ra.
\ee
Next us first bound the sum of $y_{wijk\nn}^2$. For any $\nn$ there is at most one $\mm$ such that $(\mm,\nn)\in \calE_{ijk}$.
Thus
\[
 \sum_{i,j,k=1}^N \;  \sum_{(\mm,\nn)\in \calE_{ijk}} \; y_{wijk\nn}^2
 \le  \sum_{i,j,k=1}^N \;  \sum_{\nn \in \calJ}  \; y_{wijk\nn}^2
 = \sum_{i,j,k=1}^N \;  \sum_{\nn \in \calJ}  \; 
  \delta_{w-1,|\nn|}|e_{ijk}|^2 n_i  |\phi_{\nn}|^2.
\]
For any fixed index $i$ one has  $\sum_{j,k=1}^N |e_{ijk}|^2 \le J_e^2$.
Thus
\be
\label{Ebound_part2}
 \sum_{i,j,k=1}^N \;  \sum_{(\mm,\nn)\in \calE_{ijk}} \; y_{wijk\nn}^2
 \le 
  J_e^2 \sum_{i=1}^N \;  \sum_{\nn \in \calJ}  \; 
  \delta_{w-1,|\nn|} n_i  |\phi_{\nn}|^2 = J_e^2 (w-1) \la \phi|P_{w-1}|\phi\ra.
 \ee
Finally, substituting Eqs.~(\ref{Ebound_part1},\ref{Ebound_part2}) into Eq.~(\ref{lemmaAC1_eq8}) one gets
\be
|\la \phi |E |\psi\ra | \le J_e \sqrt{s_e}  \sum_{w\ge 2}  (w-1)\sqrt{w} \sqrt{\la \psi|P_w|\psi\ra \la \phi|P_{w-1}|\phi\ra}.
\ee
Applying Cauchy-Schwartz again gives
\be
\label{lemmaAC1_eq11}
|\la \phi |E |\psi\ra | \le J_e \sqrt{s_e}
\sqrt{ \left(\sum_{w\ge 2}  (w-1)^2
\la \phi|P_{w-1}|\phi\ra \right)
\left( \sum_{w\ge 2} w\la \psi|P_w|\psi\ra\right)}
\ee
Note that 
\[
\sum_{w\ge 2} w P_w \le \sum_{w\ge 1} w P_w = \sum_{\mm \in \calJ} \sum_{i=1}^N m_i |\mm\ra\la \mm| \le \lambda_1^{-1}  \sum_{\mm \in \calJ} \sum_{i=1}^N m_i \lambda_i  |\mm\ra\la \mm| = \lambda_1^{-1} A.
\]
It follows that 
\be
\label{lemmaAC1_eq12}
 \sum_{w\ge 2} w\la \psi|P_w|\psi\ra  \le \lambda_1^{-1} \la \psi |A|\psi\ra.
\ee
Likewise,
\be
\label{lemmaAC1_eq13}
\sum_{w\ge 2} (w-1)^2  \la \phi|P_{w-1}|\phi\ra =
\sum_{w\ge 1}w^2  \la \phi|P_{w}|\phi\ra 
\le  \lambda_1^{-2} \la \phi|A^2|\phi\ra.
\ee
Substituting Eqs.~(\ref{lemmaAC1_eq12},\ref{lemmaAC1_eq13}) into Eq.~(\ref{lemmaAC1_eq11}) one gets
\be
\label{lemmaAC1_eq14}
|\la \phi |E|\psi\ra | \le  J_e \sqrt{s_e} \lambda_1^{-3/2}  \sqrt{\la \phi|A^2|\phi\ra \la \psi|A|\psi\ra }.
\ee 

\subsection{Proof of Lemma~\ref{lemma:AC1}}
The triangle inequality gives
\be
\label{Cupdown}
|\la \phi|C|\psi\ra|\le |\la \phi|C^\uparrow|\psi\ra| +|\la \phi|C^\downarrow|\psi\ra|.
\ee
Let us upper bound $|\la \phi|C^\uparrow|\psi\ra|$.
Consider  a tensor 
\[
d_{ijk} =(q/2)^{1/2} \frac{c_{ijk}}{\lambda_j \lambda_k}.
\]
Then $C^\uparrow$ coincides with the operator $D$ defined in Lemma~\ref{lemma:master}.
The lemma gives
\[
|\la \phi |C^\uparrow|\psi\ra |\le J_d \sqrt{s}  \sqrt{\la \psi|A|\psi\ra \la \phi|A^2|\phi\ra}.
\]
Here we noted that $c$ and $d$ have the same number of drift channels. 
Clearly,  $|d_{ijk}|\le (q/2)^{1/2} \lambda_1^{-2} |c_{ijk}|$.
By assumption, any two-dimensional slice of $c$ has $2$-norm at most $J$. 
Thus the parameter $J_d$ of Lemma~\ref{lemma:master}  is bounded as 
$J_d \le (q/2)^{1/2} \lambda_1^{-2} J$ and we get
\be
\label{Cup}
|\la \phi |C^\uparrow|\psi\ra | \le \frac1{\sqrt{2}} (qs)^{1/2} J  \lambda_1^{-2}  \sqrt{\la \phi|A^2|\phi\ra \la \psi|A|\psi\ra }.
\ee
Next let us upper bound $|\la \phi |C^\downarrow|\psi\ra |$. 
Consider  a tensor
\be
\label{e_from_Cdown}
e_{ijk} = (q/2)^{1/2} c_{ijk}  \sqrt{\frac{\lambda_i}{\lambda_j\lambda_k}}.
\ee
Then the operator $E$ of Lemma~\ref{lemma:master} coincides with $C^\downarrow$ and we get
\be
\label{Cdown_intermediate}
|\la \phi |C^\downarrow|\psi\ra |\le J_e \sqrt{s}  \lambda_1^{-3/2}\sqrt{\la \psi|A|\psi\ra \la \phi|A^2|\phi\ra}.
\ee
Here we noted that $c$ and $e$ have the same support so that $s_e=s$.
Combining Eq.~(\ref{e_from_Cdown}) and the first part of Proposition~\ref{prop:coeff_upper_bound}
gives
\[
|e_{ijk}|\le q^{1/2} \lambda_1^{-1/2} J_{ijk}.
\]
For any index $i$ one has 
\[
\sum_{j,k=1}^N |e_{ijk}|^2 \le q \lambda_1^{-1} \sum_{j,k=1}^N J_{ijk}^2.
\]
Furthermore, the definition of $J_{ijk}$ implies
\[
|J_{ijk}|^2\le |c_{ijk}|^2 + |c_{jki}|^2 + |c_{kij}|^2.
\]
Indeed, since the tensor $c$ is symmetric under swapping the last two indices,
a permutation of indices that maximizes $|c_{ijk}|$ can be chosen as a cyclic shift.
We arrive at
\[
\sum_{j,k=1}^N |e_{ijk}|^2 \le q \lambda_1^{-1} \sum_{j,k=1}^N  |c_{ijk}|^2 + |c_{jki}|^2 + |c_{kij}|^2 \le 3J^2 q \lambda_1^{-1} 
\]
since any two-dimensional slice of $c$ has $2$-norm at most $J$.
Thus $J_e \le J \sqrt{3q/\lambda_1}$. Substituting this into Eq.~(\ref{Cdown_intermediate}) gives
\be
\label{Cdown}
|\la \phi |C^\downarrow|\psi\ra | \le  (3qs)^{1/2}  J \lambda_1^{-2}  \sqrt{\la \phi|A^2|\phi\ra \la \psi|A|\psi\ra }.
\ee 
Combining Eqs.~(\ref{Cupdown},\ref{Cdown},\ref{Cup})
and the inequality $\sqrt{3} + 1/\sqrt{2}\le 3$
 proves the lemma.

\subsection{Proof of Lemma~\ref{lemma:AC2}}

We shall first prove $p=1$ case. 
The triangle inequality gives
\be
\label{comm_eq1}
|\la \psi|[A,C]|\psi\ra| \le |\la \psi|[A,C^\uparrow]|\psi\ra| + |\la \psi|[A,C^\downarrow]|\psi\ra| = 2|\la \psi|[A,C^\uparrow]|\psi\ra|.
\ee
Using the canonical commutation rules $[a_i,a_j^\dag]= \delta_{i,j} I$
and $[a_i,a_j]=0$ 
one gets 
\[
[A,a_j^\dag]=\lambda_j a_j^\dag \quad \mbox{and} \quad  [A,a_i]=-\lambda_i a_i.
\]
The chain rule for commutators then gives
\[
[A,a_j^\dag a_k^\dag a_i] = (\lambda_j+\lambda_k-\lambda_i)a_j^\dag a_k^\dag a_i.
\]
Thus
\[
[A,C^\uparrow] = (q/2)^{1/2}\sum_{i,j,k=1}^N
c_{ijk}
   \sqrt{\frac{\lambda_i}{\lambda_j\lambda_k}}(\lambda_j+\lambda_k-\lambda_i)
a_j^\dag a_k^\dag a_i.
\]
Consider a tensor
\[
d_{ijk} =  (q/2)^{1/2} c_{ijk}  \frac1{\lambda_j \lambda_k}(\lambda_j+\lambda_k-\lambda_i).
\]
Then $[A,C^\uparrow]=\sum_{i,j,k=1}^N d_{ijk} (\lambda_i \lambda_j \lambda_k)^{1/2} a_j^\dag a_k^\dag a_i$ coincides with the operator $D$
defined in Lemma~\ref{lemma:master}.
 Since the tensor $d$ has at most $s$
 drift channels
(note that the factor $\lambda_j+\lambda_k-\lambda_i$ can be zero, which can potentially 
eliminate some drift channels),
Lemma~\ref{lemma:master} gives
\be
\label{comm_eq2}
|\la \psi| [A,C^\uparrow]|\psi\ra| \le J_d \sqrt{s}  \sqrt{\la \psi|A|\psi\ra\,\la \psi|A^2|\psi\ra}.
\ee
Combining the definition of $d_{ijk}$ and the second part
of Proposition~\ref{prop:coeff_upper_bound} gives
\[
|d_{ijk}|\le (2q)^{1/2} J_{ijk} \lambda_1^{-1}.
\]
Repeating the same arguments as in the proof of Lemma~\ref{lemma:AC1} gives
\[
\sum_{j,k=1}^N |d_{ijk}|^2 \le 2q\lambda_1^{-2} \sum_{j,k=1}^N J_{ijk}^2 \le 6q \lambda_1^{-2} J^2.
\]
Thus $J_d \le \sqrt{6q} \lambda_1^{-1} J$. Combining Eqs.~(\ref{comm_eq1},\ref{comm_eq2})
gives
\[
|\la \psi|[A,C]|\psi\ra| \le 2\sqrt{6qs} J \lambda_1^{-1}  \sqrt{\la \psi|A|\psi\ra\,\la \psi|A^2|\psi\ra}
\]
which is the claimed bound.

Consider now the case $p\ge 2$.
The triangle inequality gives
\be
\label{high_order_eq1}
|\la \psi|[A^p,C]|\psi\ra| \le
|\la \psi|[A^p,C^\uparrow]|\psi\ra| + |\la \psi|[A^p,C^\downarrow]|\psi\ra| =
2|\la \psi|[A^p,C^\uparrow]|\psi\ra|. 
\ee
Let $\lambda_\mm =\la \mm|A|\mm\ra = \sum_{i=1}^N \lambda_i m_i$.
For any triple of indices $i,j,k$  let $\calE_{ijk}\subseteq \calJ\times \calJ$ be the set of pairs
$(\nn,\mm)$ such that $\nn=\mm-e^i+e^j+e^k$, $\nn\in \calJ$, and $\mm\in \calJ$.
Then
$\la \nn|a_j^\dag a_k^\dag a_i|\mm\ra=0$ unless $(\nn,\mm)\in\calE_{ijk}$.
Furthermore, if $(\nn,\mm)\in\calE_{ijk}$ then
$\lambda_\nn-\lambda_\mm=\lambda_j+\lambda_k-\lambda_i$.
Hence
\[
|\lambda_\nn-\lambda_\mm|\le 2\lambda_N \qquad \mbox{for all $(\nn,\mm)\in \calE_{ijk}$}.
\]
Define a parameter
\[
\kappa=1+\frac{2\lambda_N}{\lambda_1}.
\]
Then
\[
\max\!\left(\frac{\lambda_\nn}{\lambda_\mm},\frac{\lambda_\mm}{\lambda_\nn}\right)\le \kappa.
\]
Using the factorization
\[
\lambda_\nn^p-\lambda_\mm^p =
(\lambda_\nn-\lambda_\mm)\sum_{t=0}^{p-1}\lambda_\nn^{\,p-1-t}\lambda_\mm^{\,t},
\]
one gets
\be
\label{high_order_eq3}
|\lambda_\nn^p-\lambda_\mm^p|
\le
p\,\kappa^{(p-1)/2}\,
|\lambda_\nn-\lambda_\mm|\,
(\lambda_\nn\lambda_\mm)^{(p-1)/2}.
\ee
Thus
\[
|\la \psi|[A^p,C^\uparrow]|\psi\ra|
\le
(q/2)^{1/2}\sum_{i,j,k=1}^N \sum_{(\nn,\mm)\in\calE_{ijk}}
|c_{ijk}|\sqrt{\frac{\lambda_i}{\lambda_j\lambda_k}}
\,|\lambda_\nn^p-\lambda_\mm^p|\,
|\la \nn|a_j^\dag a_k^\dag a_i|\mm\ra|\,
|\psi_\nn|\,|\psi_\mm|.
\]
The second part of Proposition~\ref{prop:coeff_upper_bound} gives 
\[
|c_{ijk}|\frac{|\lambda_j+\lambda_k-\lambda_i|}{\lambda_j\lambda_k}
\le \frac{2J_{ijk}}{\lambda_1}.
\]
Combining this and Eq.~(\ref{high_order_eq3}) gives
\[
|\la \psi|[A^p,C^\uparrow]|\psi\ra|
\le
p\sqrt{2q}\,\lambda_1^{-1}\kappa^{(p-1)/2}
\sum_{i,j,k=1}^N \sum_{(\nn,\mm)\in\calE_{ijk}} x_{ijk\mm}y_{ijk\nn},
\]
where
\[
x_{ijk\mm}= J_{ijk} (\lambda_i m_i)^{1/2}\,\lambda_\mm^{(p-1)/2}\,|\psi_\mm|
\]
and
\[
y_{ijk\nn}= \chi_{ijk}(\lambda_j n_j\lambda_k n_k)^{1/2}\,\lambda_\nn^{(p-1)/2}\,|\psi_\nn|.
\]
Here $\chi_{ijk}\in \{0,1\}$ is the indicator function of $c$ such that $c_{ijk}=c_{ijk}\chi_{ijk}$ for all indices.
Applying Cauchy--Schwarz gives
\be
\label{high_order_eq4}
|\la \psi|[A^p,C^\uparrow]|\psi\ra|
\le
p\sqrt{2q}\,\lambda_1^{-1}\kappa^{(p-1)/2}
\sqrt{\sum_{i,j,k=1}^N\sum_{(\nn,\mm)\in\calE_{ijk}} x_{ijk\mm}^2}
\sqrt{\sum_{i,j,k=1}^N\sum_{(\nn,\mm)\in\calE_{ijk}} y_{ijk\nn}^2}.
\ee

Let us upper bound the sum of $x_{ijk\mm}^2$. From 
\[
\sum_{j,k=1}^N J_{ijk}^2
\le
\sum_{j,k=1}^N  |c_{ijk}|^2+|c_{jki}|^2+|c_{kij}|^2 \le 3J^2,
\]
one gets
\be
\label{high_order_eq5}
\sum_{i,j,k=1}^N\sum_{(\nn,\mm)\in\calE_{ijk}} x_{ijk\mm}^2
\le
3J^2 \sum_{\mm\in\calJ} \lambda_\mm^{p-1}
\sum_{i=1}^N \lambda_i m_i |\psi_\mm|^2
=
3J^2 \la \psi|A^p|\psi\ra.
\ee

Next let us upper bound the sum of $y_{ijk\mm}^2$. 
Using the  assumption
$\sum_{i=1}^N \chi_{ijk}\le s$ 
for every fixed pair $j,k$ one gets
\be
\label{high_order_eq6}
\sum_{i,j,k=1}^N\sum_{(\nn,\mm)\in\calE_{ijk}} y_{ijk\nn}^2
\le
s\sum_{\nn\in\calJ}\lambda_\nn^{p-1}
\sum_{j,k=1}^N \lambda_j n_j \lambda_k n_k |\psi_\nn|^2
=
s\la \psi|A^{p+1}|\psi\ra.
\ee
Substituting Eqs.~(\ref{high_order_eq5},\ref{high_order_eq6}) into Eq.~(\ref{high_order_eq4}) gives
\[
|\la \psi|[A^p,C^\uparrow]|\psi\ra|
\le
p\sqrt{6qs}\,J\,\lambda_1^{-1}\kappa^{(p-1)/2}
\sqrt{\la \psi|A^p|\psi\ra \la \psi|A^{p+1}|\psi\ra}.
\]
Finally, combining this with Eq.~(\ref{high_order_eq1})  proves the lemma.

\subsection{Proof of Lemma~\ref{lemma:loose}}

Recall that  $C$ is a linear combination of finitely many operators $a_j^\dag a_k^\dag a_i - a_i^\dag a_k a_j$.  We have
\[
\|A^p a_j^\dag a_k^\dag a_i \psi\|^2  = \la \psi| X_{ijk}|\psi\ra, \qquad X_{ijk}\equiv a_i^\dag a_k a_j A^{2p} a_j^\dag a_k^\dag a_i.
\]
Note that $X_{ijk}$ is diagonal in the Fock basis $|\mm\ra$. From Eq.~(\ref{ai_action}) one gets
From Eq.~(\ref{ai_action}) one gets
\begin{align*}
\la \mm| X_{ijk}|\mm\ra & = 2m_i (m_j+1)(m_k+1)  \la \mm+e^j+e^k-e^i|A^{2p} |\mm +e^j+e^k-e^i\ra \\ 
&\le 2^{2p+1} m_i (m_j+1)(m_k+1)
(\la \mm |A^{2p}|\mm\ra + (2\lambda_N)^{2p}).
\end{align*}
Here the inequality follows from $(y+z)^{2p}\le (2y)^{2p} + (2z)^{2p}$ which holds for all $y,z\ge 0$.
Using a bound $m_i \le \lambda_1^{-1} \la \mm|A|\mm\ra$ and the analogous bound for $m_j,m_k$ one gets
\[
X_{ijk} \le 2^{2p+1}\lambda_1^{-3}  A(A + \lambda_1I) (A+\lambda_1I) ( A^{2p} +  (2\lambda_N)^{2p} I).
\]
For any integer $1\le \ell \le 2p+3$ one has   $A^{\ell} \le \lambda_1^{-2p-3+\ell} A^{2p+3}$.
Thus
$X_{ijk}\le \omega_p' A^{2p+3}$, where $\omega_p'<\infty$ depends only on $p$, $\lambda_1$, and $\lambda_N$. 
Hence 
$\|A^p a_j^\dag a_k^\dag a_i \psi\|^2 \le \omega_p' \la \psi|A^{2p+3}|\psi\ra$.
Similar arguments show that  $\|A^p a_i^\dag a_k a_j\psi\|^2$ and $\|A^p  a_j^\dag a_i\psi\|^2$
are upper bounded by $\omega_p' \la \psi|A^{2p+3}|\psi\ra$.
We can now bound $\| A^p C\psi\|$ using the triangle inequality arriving at the claimed bound.

\section[Regularization]
{Regularization\protect\footnote{This section is largely based on the preprint~\cite{bravyi2025quantum}, but also includes new material, in particular Theorem~\ref{thm:regul-kp}.}}
\label{sec:regul-sz}

In this section we prove that the Kolmogorov equation has a well-defined solution. Our proof is constructive in the sense that it provides 
a regularization procedure approximating the exact solution 
along with  rigorous bounds on the approximation error. The regularized Kolmogorov equation can be simulated on a quantum computer using Hamiltonian simulation methods. Define a set of multi-indices
\[
\calJ = \{\mm =(m_1,\ldots,m_N) \in \ZZ_{\ge 0}^N \, : \, m_1+\ldots+m_N\ge 1\}.
\]
By definition, any vector in the Fock space $\calF_0$ is a linear combination of basis states $|\mm\ra$ with $\mm \in \calJ$.
Given a real cutoff parameter $k>0$,  let
\[
\calJ_k = \{\mm \in \calJ \, : \, \sum_{i=1}^N \lambda_i m_i  \le k\}.
\]
In other words, $\mm \in \calJ_k$ iff $\la \mm|A|\mm\ra\le k$.
Note that  $\calJ_k$ is a finite set (possibly empty) for any $k>0$. Indeed, $\mm \in \calJ_k$ implies $|\mm|\equiv \sum_{i=1}^N m_i \le \lceil k/\lambda_1\rceil$.
It is well known that the number of partitions of an integer $L$  into a sum of $N$ non-negative integers is ${L+N-1\choose L}$.
Substituting $L=0,1,\ldots, \lceil k/\lambda_1 \rceil$ gives $|\calJ_k|\le O(N^{k/\lambda_1})$. Define a projector
\[
\Pi_k = \sum_{\mm \in \calJ_k} |\mm\ra\la \mm|
\]
and a subspace
\[
\calH_k = \mathrm{span}(|\mm\ra\, : \, \mm \in \calJ_k) \subseteq \calF_0.
\]
We shall refer to $\calH_k$ as a {\em low-dissipation subspace}.
Define the regularized Kolmogorov equation as 
\be
\label{KE_k}
\frac{d}{dt} |\psi_k(t)\ra = (-A + \Pi_k C\Pi_k ) |\psi_k(t)\ra
\ee
with the initial condition $|\psi_k(0)\ra = |\psi(0)\ra$, where $|\psi(0)\ra\in \calF_0$ is the initial condition for the original unregularized equation. 
We emphasize that all regularized Kolmogorov equations are defined on the same Hilbert space, namely the infinite-dimensional Fock space $\calF_0$.
However the operator $C$ is ``turned on" only on the low-dissipation subspace $\calH_k$, which is finite-dimensional.
The orthogonal complement of $\calH_k$ evolves only under the diagonal operator $-A$ 
so that $\la \mm |\psi_k(t)\ra = e^{-\lambda_\mm t } \la \mm|\psi(0)\ra$ for all $\mm \notin \calJ_k$, where $\lambda_\mm = \sum_{i=1}^N \lambda_i m_i$.
  Using these observations
one can easily check that Eq.~(\ref{KE_k}) has an infinitely differentiable solution $|\psi_k(t)\ra \in \calF_0$.
The main result of this section is the regularization theorem providing an upper bound on the error up to which the regularized solution $\psi_k(t)$ approximates the original unregularized 
solution $\psi(t)$.
We prove two versions of the theorem: a weak and a strong one. The former makes the minimal assumptions about the initial condition $\psi(0)$ and the SDE. The resulting bound on the regularization error scales as $1/k^{1/2}$. The latter proves an upper bound $1/k^p$ for any constant $p$ but it requires stronger assumptions about $\psi(0)$ and dissipation rates in the SDE. 

\subsection{Weak regularization theorem}

\begin{theorem}
\label{thm:regul}
Consider any vector $\psi_0\in \calF_0$ such that 
 $\la \psi_0|A^8|\psi_0\ra<\infty$. There exists a differentiable function $\psi\, : \, \RR_{\ge 0} \to \calF_0$ 
 with $\psi(0)=\psi_0$  such that for any evolution time $t\ge 0$ 
 one has $\la \psi(t)|A^3|\psi(t)\ra<\infty$ and $\psi(t)$ solves the Kolmogorov equation\footnote{More formally,
 $\| \psi(t+\delta) - \psi(t) - \delta (-A+C)\psi(t)\| = O(\delta^2)$ in the limit $\delta\to 0$, where the norm is the usual $2$-norm of the Fock
 space $\calF_0$.}
 \[
 \frac{d}{dt}\psi(t)=(-A+C)\psi(t).
 \]
For any regularization cutoff $k>0$ and for all $t\ge 0$
\be
\label{regul_bound1}
\| \psi_k(t) - \psi(t)\|^2 \le  \frac{12 \gamma}{ k^{1/2}} \left(\la \psi_0|A|\psi_0\ra  + \kappa_1 \la\psi_0|\psi_0\ra\right)
\ee
with
\[
\gamma = \frac{3J \sqrt{qs}}{\lambda_1^2} \quad \mbox{and} \quad \kappa_1=\frac{6J^2 qs}{\lambda_1^2}.
\]
Furthermore, 
\[
\| (I-\Pi_k)\psi(t)\|^2 \le \frac{\la \psi_0|A|\psi_0\ra + \kappa_1 \la \psi_0|\psi_0\ra }{ k}.
\]
% Furthermore, let $\lambda_N$ be the largest dissipation rate. Then 
%  %if $\frac{\lambda_N}{\lambda_1}$ is independent of $N$ the following improved estimate holds
% \begin{align}\label{eq:kp-reg-error}
%   \| \Pi_k (\psi_k(t) - \psi(t)) \|^2  \le \frac{\gamma^2}{k^p}\sum_{m=0}^{p+1} \calM_m(0) \kappa_1^{p+1-m} \frac{(p+1)!^2}{m!^2} \left(1+\frac{2\lambda_N}{\lambda_1}\right)^{\frac{p(p+1)-m(m-1)}{2}}
% %\left[\left(\frac{6qs\,J^2}{\lambda_1^2}\right)^{p+1-m} \left(\frac{(p+1)!}{m!}\right)^2 \left(1+\frac{2\lambda_N}{\lambda_1}\right)^{\frac{p(p+1)-m(m-1)}{2}}\right]\calM_m(0)
% \end{align}
\end{theorem}
{\em Comment:}
  Note that the bound $\la \psi(t)|A^3|\psi(t)\ra<\infty$ implies $\| C\psi(t)\|<\infty$ due to Lemma~\ref{lemma:loose} and thus
$(-A+C)\psi(t)\in \calF_0$, that is, $\psi(t)$ belongs to the domain of the operator $-A+C$.

In the rest of this section we prove the theorem.
For a fixed cutoff parameter $k>0$ define moment functions
\[
f_p(t) =\la \psi_k(t)|A^p|\psi_k(t)\ra,
\]
where $p$ is a non-negative integer. 
We shall need uniform (independent of $k$) upper bounds on $f_p(t)$ 
established
in the following lemma.
\begin{lemma}[\bf Moment bounds]
\label{lemma:moment_bounds}
For any integer $p\ge 0$ and evolution time $t\ge 0$ one has 
\be
\label{f_moment_p}
f_p(t) \le  \max{(f_p(0), (\kappa_p)^p f_0(0))},
\ee
where $\kappa_0\equiv 1$ and $\kappa_p$ with $p\ge 1$ are the coefficients defined in Corollary~\ref{corol:commutator_general}.
Furthermore, 
\be
\label{f_moment_p_integral}
\int_0^t d\tau\, f_p(t) \le f_{p-1}(0) + \kappa_{p-1} f_{p-2}(0) + \kappa_{p-2} \kappa_{p-1} f_{p-3}(0) + \ldots + (\kappa_1 \kappa_2 \cdots \kappa_{p-1}) f_0(0)
\ee
for $p\ge 1$.
\end{lemma}
\begin{proof}
From the equation of motion Eq.~(\ref{KE_k}) one gets
\[
\frac{d}{dt} f_p(t) = -2f_{p+1}(t) + \la \psi_k(t)| [A^p, \Pi_k C\Pi_k  ] |\psi_k(t)\ra = -2f_{p+1}(t) +\la \psi_k(t)|\Pi_k [A^p,C] \Pi_k  |\psi_k(t)\ra.
\]
Here we used the fact that $A$ is hermitian while $C$ is anti-hermitian. 
Also we used the identity $A\Pi_k=\Pi_k A$. 
Suppose $p=0$. Then the term with the commutator disappears, that is, $(d/dt)f_0(t) = -2 f_1(t)$.
Since $f_1(t)\ge 0$, this implies $f_0(t)\le f_0(0)$, proving Eq.~(\ref{f_moment_p}) with $p=0$.
Integrating $(d/dt)f_0(t) = -2 f_1(t)$ gives
\[
\int_0^t f_1(t) = \frac12( f_0(0)-f_0(t)) \le \frac12 f_0(0) \le f_0(0)
\]
proving Eq.~(\ref{f_moment_p_integral}) with $p=1$.
Suppose $p \ge 1$.
Applying the commutator bound of Corollary~\ref{corol:commutator_bound_order1} and Lemma~\ref{lemma:AC2} with $\psi = \Pi_k \psi_k(t)$ one gets
\[
|\la \psi_k(t)|\Pi_k [A^p,C] \Pi_k  |\psi_k(t)\ra| \le \la \psi_k(t)| A^{p+1} + \kappa_p A^p |\psi_k(t)\ra \le f_{p+1}(t) + \kappa_p f_p(t).
\]
Hence
\be
\label{f_p_derivative_eq1}
\frac{d}{dt} f_p(t) \le  -f_{p+1}(t)  + \kappa_p f_p(t).
\ee
Jensen's inequality gives
\[
f_{p+1}(t) \ge (f_0(t))^{-1/p} (f_p(t))^{(p+1)/p} \ge (f_0(0))^{-1/p} (f_p(t))^{(p+1)/p},
\]
where the  second inequality follows from $f_0(t)\le f_0(0)$.
Thus
\[
\frac{d}{dt} f_p(t) \le -  (f_0(0))^{-1/p} (f_p(t))^{(p+1)/p} + \kappa_p f_p(t).
\]
The derivative of $f_p(t)$ is  negative whenever $f_p(t)> (\kappa_p)^p f_0(0)$.
Hence, if  $f_p(0)\le (\kappa_p)^p f_0(0)$ then $f_p(t)\le (\kappa_p)^p f_0(0)$ for all $t\ge 0$. Otherwise, if $f_p(0)>(\kappa_p)^p f_0(0)$, then
$f_p(t)$ is monotone decreasing until the first time moment $t$ such that  $f_p(t)\le (\kappa_p)^p f_0(0)$ and this bound holds for all subsequent times.
This proves Eq.~(\ref{f_moment_p}).
From Eq.~(\ref{f_p_derivative_eq1}) one gets
\[
\int_0^t d\tau\, f_{p+1}(\tau) \le f_p(0) +  \kappa_p \int_0^t d\tau \, f_p(\tau).
\]
Applying this recursively proves Eq.~(\ref{f_moment_p_integral}).
\end{proof}
As a corollary of Lemma~\ref{lemma:moment_bounds} one infers that solutions of the regularized bKE  have most of their mass
on the low-dissipation subspace, as formally stated below. 
\begin{corol}[\bf Projection error]
\label{corol:projection_error}
For any cutoff parameters $k,\ell >0$ and evolution time $t\ge 0$ one has
\be
\label{projection_error}
\| (I-\Pi_k)\psi_\ell(t)\|^2 \le \frac{\la \psi(0)|A|\psi(0)\ra + \kappa_1 \la \psi(0)|\psi(0)\ra }{ k}.
\ee
%Here $\kappa_1$ is the coefficient defined in Section~\ref{sec:ABCstructural}.
\end{corol}
\begin{proof}
An operator inequality $I-\Pi_k \le k^{-1} A$ gives 
\be
\la \psi_\ell(t) | I- \Pi_k|\psi_\ell(t)\ra \le   k^{-1} \la \psi_\ell(t) |A|\psi_\ell(t)\ra \le  k^{-1} \cdot \max{(f_1(0), \kappa_1f_0(0))}.
\ee
Here the second inequality follows from Lemma~\ref{lemma:moment_bounds} with $p=1$.
This implies Eq.~(\ref{projection_error}).
\end{proof}
Fix a pair of cutoff parameters  $0<k<\ell$ and define  an ``error state"
\be
|e(t)\ra = |\psi_k(t)\ra - |\psi_\ell(t)\ra.
\ee
Our strategy is to prove that the norm of $e(t)$ goes to zero in the limit $k,\ell\to \infty$.
For technical reasons, we shall also need an upper bound on the norm of $A^{p/2}e(t)$ for integers $p=O(1)$.
To this end define moments
\be
h_p(t) = \la e(t)|A^p|e(t)\ra.
\ee
We shall be primarily concerned with $h_0(t)$ which controls the rate of convergence of solutions $\psi_k(t)$ in the limit $k\to \infty$.
Higher order moments $h_p(t)$ with $p\ge1$ are only needed to prove that 
$\psi_k(t)$ ``converges to the right thing" in the sense that the limiting point of the sequence $\psi_k(t)$ solves the Kolmogorov equation (the unregularized one).
 In particular, the runtime of our quantum algorithm only depends on the upper bound on $h_0(t)$.
\begin{lemma}
\label{lemma:regularization_p}
For  any  cutoff parameters $0< k<\ell$ and evolution time $t\ge 0$ one has 
\be
\label{h0_upper_bound}
h_0(t)\le 
\frac{12 \gamma}{k^{1/2}} (\la \psi(0)|A|\psi(0)\ra  + \kappa_1 \la\psi(0)|\psi(0)\ra).
\ee
For any integer $p\ge 1$ 
\be
\label{hp_upper_bound}
h_p(t) \le \frac{e^{2\kappa_p t}}{k^{1/2}} poly(s,J,\lambda_1^{-1},\lambda_N)\max_{0\le p'\le 2p+2} \la \psi(0)|A^{p'}|\psi(0)\ra.
\ee
Here  $\gamma$ and $\kappa_p$ are the coefficients defined in Lemma~\ref{lemma:AC1} and Corollary~\ref{corol:commutator_general}.
\end{lemma}
\begin{proof}
From the equation of motion Eq.~(\ref{KE_k}) 
and the analogous equation for $\psi_\ell(t)$
one gets
\[
\frac{d}{dt} |e(t)\ra =-A |e(t)\ra  + \Pi_k C\Pi_k |\psi_k(t)\ra - \Pi_\ell C\Pi_\ell |\psi_\ell(t)\ra
\]
and
\[
\frac{d}{dt} \la e(t)| =- \la e(t)|A  - \la \psi_k(t)| \Pi_k C\Pi_k + \la \psi_\ell(t)| \Pi_\ell C\Pi_\ell.
\]
Here we noted that $A$ is hermitian while $C$ is anti-hermitian. 
It follows that
\[
\frac12 \frac{d}{dt} h_p(t) = -h_{p+1}(t) +  \la e(t) | A^p  \Pi_k C\Pi_k |\psi_k(t)\ra - \la e(t) | A^p  \Pi_\ell C\Pi_\ell |\psi_\ell(t)\ra.
\]
Clearly, $k<\ell$ implies $\calJ_k\subseteq \calJ_\ell$. Thus one can write
\[
\Pi_\ell = \Pi_k + \Pi_k^\perp, \qquad \Pi_k^\perp \equiv  \sum_{\mm \in \calJ_\ell\setminus \calJ_k} \; |\mm\ra\la \mm|.
\]
Then
\begin{align*}
\frac12 \frac{d}{dt} h_p(t) & = -h_{p+1}(t)  +  \la e(t) | A^p  \Pi_k C\Pi_k |e(t)\ra
-   \la e(t) | A^p  \Pi_k^\perp  C\Pi_k |\psi_\ell(t)\ra \\
& -  \la e(t) | A^p  \Pi_k  C\Pi_k^\perp |\psi_\ell(t)\ra  -  \la e(t) | A^p  \Pi_k^\perp  C\Pi_k^\perp |\psi_\ell(t)\ra.
\end{align*}
Note that
\[
 \la e(t) | A^p  \Pi_k C\Pi_k |e(t)\ra =\frac12  \la e(t) | \Pi_k  [A^p,C] \Pi_k |e(t)\ra
 \]
for $p\ge 1$. Indeed,  $(1/2)  \Pi_k  [A^p,C] \Pi_k$ is the real (hermitian) part of the operator $A^p  \Pi_k C\Pi_k$.
Applying the commutator bound of Corollary~\ref{corol:commutator_bound_order1} with $\psi=\Pi_k e(t)$ one gets
\[
  \la e(t) | \Pi_k  [A^p,C] \Pi_k |e(t)\ra \le \la e(t)| A^{p+1} + \kappa_p A^p |e(t)\ra.
 \]
 We arrive at 
\begin{align}
\frac12 \frac{d}{dt} h_p(t) &  \le 
\kappa_p h_p(t)+
  |\la e(t) | A^p  \Pi_k^\perp  C\Pi_k |\psi_\ell(t)\ra|\nonumber \\ 
  & + 
|\la e(t) | A^p  \Pi_k  C\Pi_k^\perp |\psi_\ell(t)\ra| 
  +  |\la e(t) | A^p  \Pi_k^\perp  C\Pi_k^\perp |\psi_\ell(t)\ra|. \label{h_p_derivative}
\end{align}
Here the term $\kappa_p h_p(t)$ disappears for $p=0$ since $[A^0,C]=0$.
 We shall bound each term that depends on $C$ 
 using Lemma~\ref{lemma:AC1} and  the moment bounds of Lemma~\ref{lemma:moment_bounds}.
Applying Lemma~\ref{lemma:AC1} with $|\psi\ra = \Pi_k^\perp A^p |e(t)\ra$ and $|\phi\ra = \Pi_k |\psi_\ell(t)\ra$
gives 
\[
|\la e(t) | A^p  \Pi_k^\perp  C\Pi_k |\psi_\ell(t)\ra| \le \gamma \sqrt{\la e(t)| \Pi_k^\perp A^{2p+1} |e(t)\ra
\la \psi_\ell(t)|A^2 |\psi_{\ell}(t)\ra}.
\]
The restriction of $A$ onto the range of $\Pi_k^\perp$ has eigenvalues at least $k$ (and at most $\ell$).
Hence 
\[
\la e(t)| \Pi_k^\perp A^{2p+1} |e(t)\ra \le k^{-1}  \la e(t) | A^{2p+2}|e(t)\ra.
\]
Using the definition $e(t) = \psi_k(t)-\psi_\ell(t)$  and
 triangle inequality one gets
\[
 \sqrt{\la e(t)|  A^{2p+2} |e(t)\ra} \le
 \sqrt{\la \psi_k(t)|  A^{2p+2} |\psi_k(t)\ra} 
+\sqrt{\la \psi_\ell(t)|  A^{2p+2} |\psi_\ell(t)\ra}.
\]
We arrive at 
\[
|\la e(t) | A^p  \Pi_k^\perp  C\Pi_k |\psi_\ell(t)\ra| \le  2\gamma k^{-1/2} \cdot  \sqrt{f_{2p+2}(t) f_2(t)}.
\]
By a slight abuse of notations, here we ignore the dependence of the moment functions $f_{2p+2}(t)$ and $f_2(t)$ on the
cutoff  parameters $k$ and $\ell$
(this is justified since the upper bound of Lemma~\ref{lemma:moment_bounds} is uniform in $k$). 
Applying exactly the same argument to the last term in Eq.~(\ref{h_p_derivative})  gives
\[
 |\la e(t) | A^p  \Pi_k^\perp  C\Pi_k^\perp |\psi_\ell(t)\ra|\le 2\gamma k^{-1/2} \cdot  \sqrt{f_{2p+2}(t) f_2(t)}.
\]
It remains to bound the term $|\la e(t) | A^p  \Pi_k  C\Pi_k^\perp |\psi_\ell(t)\ra|$ 
 in Eq.~(\ref{h_p_derivative}).
Applying Lemma~\ref{lemma:AC1} with $|\phi\ra = \Pi_k  A^p |e(t)\ra$ and $|\psi\ra = \Pi_k^\perp |\psi_\ell(t)\ra$
gives 
\[
|\la e(t) | A^p  \Pi_k  C\Pi_k^\perp |\psi_\ell(t)\ra|\le \gamma \sqrt{ \la e(t)|A^{2p+2} |e(t)\ra \la \psi_\ell(t)| A \Pi_k^\perp |\psi_\ell(t)\ra}.
\]
The same arguments as above give $ \la \psi_\ell(t)| A \Pi_k^\perp |\psi_\ell(t)\ra\le k^{-1} f_2(t)$ and
$ \sqrt{\la e(t)|A^{2p+2} |e(t)\ra} \le 2\sqrt{f_{2p+2}(t)}$. Hence 
\[
|\la e(t) | A^p  \Pi_k  C\Pi_k^\perp |\psi_\ell(t)\ra|\le2\gamma k^{-1/2} \cdot  \sqrt{f_{2p+2}(t) f_2(t)}.
\]
Combining the above bounds and Eq.~(\ref{h_p_derivative}) gives
\be
\label{h_0_derivative1}
\frac{d}{dt} h_0(t) \le 12 \gamma k^{-1/2} f_2(t)
\ee
and
\be
\label{h_p_derivative1}
\frac{d}{dt} h_p(t)\le 2\kappa_p h_p(t) + 12   \gamma k^{-1/2} \sqrt{f_{2p+2}(t) f_2(t)}
\ee
for $p\ge 1$. Consider first the case $p=0$. Integrating Eq.~(\ref{h_0_derivative1}), recalling that $h_0(0)=0$,  and using the second part of Lemma~\ref{lemma:moment_bounds}
gives
\begin{align}
h_0(t) & \le  12 \gamma k^{-1/2} \int_0^t d\tau\,  f_2(\tau)
\le  12 \gamma k^{-1/2} (f_1(0) + \kappa_1 f_0(0)) \nonumber \\
& =  12 \gamma k^{-1/2} (\la \psi(0)|A|\psi(0)\ra  + \kappa_1 \la\psi(0)|\psi(0)\ra). 
\end{align}
Suppose $p\ge 1$. Lemma~\ref{lemma:moment_bounds} implies that
$ \sqrt{f_{2p+2}(t) f_2(t)}\le poly(J,s,\lambda_1^{-1},\lambda_N)\max_{0\le p'\le 2p+2} f_{p'}(0)$ for all $t\ge0$.
Hence Eq.~(\ref{h_p_derivative1}) and Gronwall inequality gives
\be
h_p(t) \le \frac{e^{2\kappa_p t}}{k^{1/2}} poly(s,J,\lambda_1^{-1},\lambda_N)\max_{0\le p'\le 2p+2} f_{p'}(0).
\ee
This is equivalent to the statement of the lemma.
\end{proof}

From Lemma~\ref{lemma:regularization_p} we derive the following. 
\begin{corol}
\label{corol:limit}
Let $p\ge 1$ be an integer such that $\la \psi(0)|A^{p'}|\psi(0)\ra<\infty$ for all $p'\in [0,2p+2]$.
Then for any $t\ge 0$ there exists a vector function  $\psi \, : \, [0,t]\to   \calF_0$ such that 
\be
\label{convergence_bound_with_A}
\lim_{k\to \infty} \,\max_{\tau \in [0,t]} \| A^{p/2} (\psi_k(\tau) - \psi(\tau))\|=0.
\ee
Furthermore, for all $t\ge 0$
\be
\label{regularization_error}
\| \psi_k(t) - \psi(t)\|^2 \le  \frac{12 \gamma}{ k^{1/2}} (\la \psi(0)|A|\psi(0)\ra  + \kappa_1 \la\psi(0)|\psi(0)\ra)
\ee
%Here $\gamma$ and $\kappa_1$ are the coefficients defined in Section~\ref{sec:ABCstructural}.
\end{corol}
\begin{proof}
Consider  a sequence of vector functions $\{A^{p/2} \psi_k(\tau) \}$ mapping $[0,t]$ to $\calF_0$.
For simplicity here we restrict ourselves to integer cutoff  parameters $k\ge 1$.
We claim that the sequence $\{A^{p/2} \psi_k(\tau) \}$  is uniformly  Cauchy.
Indeed,  for any $1\le k<\ell$ one has
\[
\| A^{p/2}\psi_k(\tau) - A^{p/2}\psi_\ell(\tau)\|^2 = \la e(\tau)|A^p |e(\tau)\ra=h_p(\tau),
\]
 where $e(\tau)=\psi_k(\tau)-\psi_\ell(\tau)$
is the error state and $h_p(\tau)$ decays as $1/\sqrt{k}$ in the limit $k \to \infty$ due to  Eq.~(\ref{hp_upper_bound}) of Lemma~\ref{lemma:regularization_p}.
Moreover, the upper bound on $h_p(\tau)$ is uniform with respect to $\tau$ (it depends only on $t$).
Thus $\{A^{p/2} \psi_k(\tau) \}$  is uniformly  Cauchy.
Since $\calF_0$ is the complete metric space, the sequence $\{A^{p/2} \psi_k(\tau)\}$
must have a limit 
\[
\theta_p(\tau)=\lim_{k\to \infty} A^{p/2} \psi_k(\tau) \in \calF_0.
\]
For each $\tau \in [0,t]$ define
\[
\psi(\tau) = A^{-p/2} \theta_p(\tau) \in \calF_0.
\]
This is well-defined  since $A^{-1/2}$ is a bounded operator on $\calF_0$ (note that eigenvalues of $A^{-1/2}$ are at most $\lambda_1^{-1/2}$).
Then $\theta_p(\tau)=A^{p/2} \psi(\tau)$ and 
Eq.~(\ref{convergence_bound_with_A}) is equivalent to
\[
\lim_{k\to \infty} \,\max_{\tau \in [0,t]} \| A^{p/2} \psi_k(\tau) - \theta_p(\tau)\|=0
\]
which follows from the fact that $\{A^{p/2} \psi_k(\tau) \}$  is uniformly  Cauchy and $\theta_p(\tau)$ is its limiting point. 
Since $A^{-1/2}$ is a bounded operator on $\calF_0$, one gets
\[
\psi(\tau) = A^{-p/2} \lim_{k\to \infty} A^{p/2} \psi_k(\tau) = \lim_{k\to \infty} \psi_k(\tau).
\]
For any $k\ge 1$ one has
\[
\| \psi(\tau)-\psi_k(\tau)\|^2 = \lim_{\ell \to \infty} \| \psi_\ell(\tau) - \psi_k(\tau)\|^2 = h_0(\tau).
\]
Thus the bound Eq.~(\ref{regularization_error}) follows from Eq.~(\ref{h0_upper_bound}) of Lemma~\ref{lemma:regularization_p}.
\end{proof}
\begin{lemma}
\label{lemma:KEsolution}
Suppose $\la \psi(0)|A^8|\psi(0)\ra<\infty$. Then
the vector function $\psi\, : \, \RR_{\ge 0} \to \calF_0$ constructed in Corollary~\ref{corol:limit} is differentiable and solves the Kolmogorov equation, that is,
\be
\label{KErestated}
\frac{d}{dt} \psi(t) = (-A+C)\psi(t)
\ee
for all $t\ge 0$.
\end{lemma}
\begin{proof}
We can apply Corollary~\ref{corol:limit} with $p=3$ since $\la \psi(0)|A^8|\psi(0)\ra<\infty$ (which implies $\la \psi(0)|A^{p'}|\psi(0)\ra<\infty$ for all $0\le p'\le 8=2p+2$).
From Eq.~(\ref{convergence_bound_with_A}) one gets
\[
\la \psi(t)|A^3|\psi(t)\ra =\lim_{k\to \infty} \la \psi_k(t)|A^3|\psi_k(t)\ra \le f_3(t),
\]
where $f_3(t)$ is the moment function upper bounded by Lemma~\ref{lemma:moment_bounds}.
Thus $\la \psi(t)|A^3|\psi(t)\ra<\infty$ which implies $\| (-A+C)\psi(t)\|<\infty$ due to Lemma~\ref{lemma:loose}.
Hence $(-A+C)\psi(t)\in \calF_0$.
We need to show that 
\be
\label{time_step_delta_eq1}
\| \psi(t+\delta) - \psi(t) - \delta (-A+C)\psi(t)\| = O(\delta^2)
\ee
in the limit $\delta\to 0$.
Consider any  $k\ge 1$. Since $\psi_k(t)$ is infinitely differentiable, one can write
\begin{align}
\label{time_step_delta_eq2}
\| \psi_k(t+\delta) - \psi_k(t) - \delta (-A+\Pi_k C\Pi_k)\psi_k(t)\| & =\left\| 
\int_t^{t+\delta} d\tau_1 \,  \left( \frac{d}{d\tau_1} \psi_k(\tau_1) - \frac{d}{dt}\psi_k(t)\right) \right\| \nonumber \\
&\le  \int_t^{t+\delta} d\tau_1 \, \int_t^{\tau_1} d\tau_2 \left\| \frac{d^2}{d\tau_2^2} \psi_k(\tau_2)\right\|
\end{align}
From the equation of motion Eq.~(\ref{KE_k}) one gets
\[
\frac{d^2}{d\tau_2^2} \psi_k(\tau_2) = \left(-A+\Pi_k C\Pi_k\right)^2 \psi_k(\tau_2).
\]
Let $\varphi \equiv \left(-A+\Pi_k C\Pi_k\right)\psi_k(\tau_2)$. Then
\[
\| \left( -A+\Pi_k C\Pi_k \right) \varphi \| \le \| A\varphi\| + \|C\Pi_k \varphi\| \le \|A \varphi\| + \omega_0 \| A^{3/2} \varphi\|
\le (\lambda_1^{-1/2} + \omega_0) \| A^{3/2} \varphi\|.
\]
Here the second inequality uses Lemma~\ref{lemma:loose}.
The definition of $\varphi$  gives
\[
 \| A^{3/2} \varphi\| \le \| A^{5/2} \psi_k(\tau_2)\| + \| A^{3/2} C \Pi_k \psi_k(\tau_2)\|.
\]
The last term can be upper bounded using Lemma~\ref{lemma:loose} which gives
\[
\| A^{3/2} C \Pi_k \psi_k(\tau_2)\|\le \omega_{3/2} \| A^3 \psi_k(\tau_2)\|.
\]
Combining the above bounds results in
\[
 \left\| \frac{d^2}{d\tau_2^2} \psi_k(\tau_2)\right\|^2 = 
\| \left(-A+\Pi_k C\Pi_k\right)^2 \psi_k(\tau_2)\|^2  \le \eta  \la \psi_k(\tau)|A^6|\psi_k(\tau)\ra
\]
for some real number $\eta<\infty$ independent of $k$ and $\delta$ (although $\eta$ depends on $N$ and other parameters of the original ODE). 
Lemma~\ref{lemma:moment_bounds} provides a uniform (independent of $k$ and $\delta$) upper bound on $\la \psi_k(\tau)|A^6|\psi_k(\tau)\ra$
since $\la \psi(0)|A^6|\psi(0)\ra<\infty$. Substituting this into Eq.~(\ref{time_step_delta_eq2}) and taking the double integeral one gets
\be
\label{time_step_delta_eq3}
\| \psi_k(t+\delta) - \psi_k(t) - \delta (-A+\Pi_k C\Pi_k)\psi_k(t)\|  \le O(\delta^2),
\ee
where the constant  hidden in $O(\delta^2)$ is independent of $k$. Suppose we can prove that 
\be
\label{time_step_delta_eq4}
\lim_{k\to \infty}  (-A+\Pi_k C \Pi_k)\psi_k(t) = (-A+C)\psi(t).
\ee
We have already established that $\lim_{k\to \infty} \psi_k(t)=\psi(t)$ and $\lim_{k\to \infty} \psi_k(t+\delta)=\psi(t+\delta)$.
Taking the limit $k\to \infty$ in Eq.~(\ref{time_step_delta_eq3}) and using Eq.~(\ref{time_step_delta_eq4}) proves
Eq.~(\ref{time_step_delta_eq1}), thereby proving the lemma. Hence it remains to prove Eq.~(\ref{time_step_delta_eq4}).
We divide the proof into two parts.
\begin{prop}
\label{prop:limit1}
\be
\label{time_step_delta_eq5}
 (-A+C)\psi(t) = \lim_{k\to \infty}  (-A+C) \psi_k(t).
\ee
\end{prop}
\begin{prop}
\label{prop:limit2}
\be
\label{time_step_delta_eq6}
\lim_{k\to \infty}  \| C\psi_k(t) -\Pi_k C\Pi_k \psi_k(t)\|=0.
\ee
\end{prop}
Replacing the righthand side of Eq.~(\ref{time_step_delta_eq4})
with $\lim_{k\to \infty}  (-A+C) \psi_k(t)$ using Proposition~\ref{prop:limit1},
cancelling the  terms $-A\psi_k(t)$ that appear in both sides, and using Proposition~\ref{prop:limit2} proves Eq.~(\ref{time_step_delta_eq4}).
\begin{proof}[\bf Proof of Proposition~\ref{prop:limit1}]
Applying the triangle inequality and Lemma~\ref{lemma:loose} one gets
\[
\|  (-A+C)\psi(t)  - (-A+C) \psi_k(t)\| \le \| A(\psi(t)-\psi_k(t))\| + \omega_0 \| A^{3/2} (\psi(t)-\psi_k(t)) \|,
\]
Both terms go to zero in the limit $k\to \infty$
due to Corollary~\ref{corol:limit}.
\end{proof}
\begin{proof}[\bf Proof of Proposition~\ref{prop:limit2}]
Let $\Pi_{>k}= I-\Pi_k$. We have
\[
C - \Pi_k  C\Pi_k  = C (\Pi_k+\Pi_{>k}) - \Pi_k  C \Pi_k = C\Pi_{>k}  + \Pi_{>k} C \Pi_{k}.
\]
Thus it suffices to prove that $\| C\Pi_{>k} \psi_k(t)\|$ and $\|  \Pi_{>k} C \Pi_{k}\psi_k(t)\|$ go to zero in the limit $k\to \infty$.
From Lemma~\ref{lemma:loose} one gets
\[
\| C\Pi_{>k} \psi_k(t)\| \le  \omega_0 \| A^{3/2} \Pi_{>k} \psi_k(t)\|.
\]
Since $\Pi_{>k}\le k^{-1} A \Pi_{>k}$, one gets
\[
\| C\Pi_{>k} \psi_k(t)\| \le   \frac{\omega_0}{k} \| A^{5/2}\psi_k(t)\| = \frac{\omega_0 \sqrt{f_5(t)}}{k},
\]
where $f_5(t)=\la \psi_k(t)|A^5|\psi_k(t)\ra$ has a uniform (independent of $k$)  upper bound of  Lemma~\ref{lemma:moment_bounds}.
Hence 
\[
\lim_{k\to \infty} \| C\Pi_{>k} \psi_k(t)\|=0.
\]
Similar  arguments give
\begin{align*}
\| \Pi_{>k} C \Pi_{k} \psi_k(t)\| &  \le k^{-1} \| A \Pi_{>k} C\Pi_k \psi_k(t)\| =  k^{-1} \|  \Pi_{>k}  A C\Pi_k \psi_k(t)\|  \\
& \le k^{-1} \| A C\Pi_k \psi_k(t)\| 
\le \omega_1 k^{-1}  \| A^{5/2}\Pi_k \psi_k(t)\| \le  \frac{\omega_1 \sqrt{f_5(t)}}{k}
\end{align*}
Here the third   inequality uses Lemma~\ref{lemma:loose} with $p=1$. Hence
\[
\lim_{k\to \infty} \| \Pi_{>k} C \Pi_{k} \psi_k(t)\|=0.
\]
This proves Eq.~(\ref{time_step_delta_eq6}).
\end{proof}

The first claim of Theorem~\ref{thm:regul} (existence of solution) is a direct consequence of Corollary~\ref{corol:limit}, combined with Lemma~\ref{lemma:KEsolution} for $p\ge3$. The regularization error Eq.~\eqref{regul_bound1} follows from Eq.~(\ref{regularization_error}).
Finally, the upper bound on 
$\|(I-\Pi_k)\psi(t)\|$
follows from Eq.~(\ref{projection_error}) by taking the limit $\ell\to \infty$.
\end{proof}
For any integer $p\ge 0$ define a moment
\[
\calM_p(t)=\la \psi(t)|A^p|\psi(t)\ra.
\]

\subsection{Strong regularization theorem}

\begin{theorem}
\label{thm:regul-kp}
Let $\lambda_N$ be the largest dissipation rate in SDE Eq.~\eqref{SDE}. Take any integer $p\ge3$ and vector $\psi_0\in \calF_0$ such that $\la \psi_0|A^{p'}|\psi_0\ra<\infty$ for all $p'\in[0,2p+6]$. Let $\psi\, : \, \RR_{\ge 0} \to \calF_0$ with $\psi(0)=\psi_0$ denote  solution of Kolmogorov equation \(
 \frac{d}{dt}\psi(t)=(-A+C)\psi(t).
 \)
For any regularization cutoff $k>0$, projectors $\Pi_k$ and $\Pi_k^\perp = I - \Pi_k$, and for all $t\ge 0$ and all $1\le p'\le p$ we have: 
\begin{align}
  &\| \Pi_k (\psi_k(t) - \psi(t)) \|^2  \le \frac{\gamma^2}{k^{p'}}\calW, \label{eq:kp-reg-error}\\
  &\| \Pi_k^\perp(\psi_k(t) - \psi(t)) \|^2  \le \frac{4\gamma}{k^{\frac{p'}{2}}}\left(\la \psi_0|A|\psi_0\ra  + \kappa_1 \la\psi_0|\psi_0\ra\right)^{\frac12} \calW^{\frac12}, \label{eq:ekperp}\\  
  &\calW=\sum_{m=0}^{p'+1}  \kappa_1^{p'+1-m}\frac{(p'+1)!^2}{m!^2} \la \psi_0|A^{m}|\psi_0\ra \left(1+\frac{2\lambda_N}{\lambda_1}\right)^{\frac{p'(p'+1)-m(m-1)}{2}}, 
\end{align}
Here $\gamma = \frac{3J \sqrt{qs}}{\lambda_1^2}$ and $\kappa_1=\frac{6J^2 qs}{\lambda_1^2}$. Specifically for $p'=1$ we get: 
\begin{align*}
    &\| \Pi_k (\psi_k(t) - \psi(t)) \|^2  \le \frac{\gamma^2}{k}\calW\\
    &\| \Pi_k^\perp (\psi_k(t) - \psi(t)) \|^2  \le \frac{4\gamma}{k^{\frac 12}}
    \left(\la \psi_0|A|\psi_0\ra  + \kappa_1 \la\psi_0|\psi_0\ra\right)^{\frac12}\calW^{\frac12}\\
    &\calW=\left(4 \calM_0(0) \kappa_1^2\bigl(1+\frac{2\lambda_N}{\lambda_1}\bigr)  + 4\calM_1(0)\kappa_1 \bigl(1+\frac{2\lambda_N}{\lambda_1}\bigr)+ \calM_2(0) \right)
\end{align*}

which improves upon Eq.~\eqref{regul_bound1} quadratically provided $\frac{\lambda_N}{\lambda_1}$ stays bounded as $N\to\infty$.
\end{theorem}
\begin{proof}
By theorem~\ref{thm:regul} there exists a differentiable function $\psi\, : \, \RR_{\ge 0} \to \calF_0$ with $\psi(0)=\psi_0$ such that $\psi$ solves the Kolmogorov equation \(
 \frac{d}{dt}\psi(t)=(-A+C)\psi(t)
 \). Corollary~\ref{corol:limit} implies that for any $\tilde p\in[0,p+2]$ one has $\calM_{\tilde p}(t)=\la \psi(t)|A^{\tilde p}|\psi(t)\ra<\infty$.  With this in mind let us prove Eq.~\eqref{eq:kp-reg-error}. Take $k>0$ and recall that \[
\Pi_k^\perp = I - \Pi_k=  \sum_{\mm \in \calJ\setminus \calJ_k} \; |\mm\ra\la \mm|.
\]
Take $\psi\in\calF_0$ such that $\calM_{\tilde p(t)}<\infty$ for $\tilde p\in[0,p+2]$. Then for any $t\ge 0$ one has:
\be
\label{projection_error1}
\| (I-\Pi_k)\psi(t)\|^2 \le \frac{\calM_{\tilde p}(t)}{k^{\tilde p}},\quad \tilde p\in[0,p+2] % \frac{\la \psi(0)|A|\psi(0)\ra + \kappa_1 \la \psi(0)|\psi(0)\ra }{ k}.
\ee
%Here $\kappa_1$ is the coefficient defined in Section~\ref{sec:ABCstructural}.
Indeed, noting that $\mm\in\calJ_k$ implies $\lambda^{\tilde p}_{\mm}=\la\mm|A^{\tilde p}|\mm\ra<k^{\tilde p}$ and hence $\lambda^{\tilde p}_{\bm m}>k^{\tilde p}$ if $\mm\notin\calJ_k$ we get:  \[
    \|(I-\Pi_k)\psi\|^2 = \sum_{\mm\notin\calJ_k} \frac{\lambda^{\tilde p}_{\bm m}}{\lambda^{\tilde p}_{\bm m}}\la \psi|\mm\ra^2
    \le \sum_{\mm\notin\calJ_k} \frac{\lambda^{\tilde p}_{\bm m}}{k^{\tilde p}}
    \la \psi|\mm\ra^2 \le \la\psi |A^{\tilde p}|\psi\ra k^{-\tilde p}=\calM_{\tilde p}(t)k^{-\tilde p},\quad \forall \tilde p\in[0,p+2]
  \]
Now fix $k>0$ and define an ``error state"
\be
|e(t)\ra = |\psi_k(t)\ra - |\psi(t)\ra.
\ee
From the equation of motion Eq.~(\ref{KE_k}) and the analogous equation for $\psi(t)$ one gets
\begin{align}
  \label{eq:error-eq}
\frac{d}{dt} |e(t)\ra =-A |e(t)\ra  + \Pi_k C\Pi_k |\psi_k(t)\ra - C|\psi(t)\ra, \quad |e(0)\ra = 0
\end{align}
Here we noted that $A$ is hermitian while $C$ is anti-hermitian. Multiplying Eq.~\eqref{eq:error-eq} by $ e_k(t)=  \Pi_k  e(t)$ we get:
\begin{align*}
  \frac12 \frac{d}{dt}  \la e_k(t)| e_k(t)\ra &= -\la e_k(t)| A| e_k(t)\ra  +  \la e_k(t) | \Pi_k C\Pi_k |\psi_k(t)\ra - \la e_k(t) | C |\psi\ra  \\
&= -\la e_k(t)| A| e_k(t)\ra  +  \la e_k(t) | \Pi_k C\Pi_k |\psi_k(t)\ra - \la e_k(t) | (\Pi_k+\Pi_k^\perp) C (\Pi_k+\Pi_k^\perp) |\psi\ra\\
                                              &=-\la e_k(t)| A| e_k(t)\ra  -  \la e_k(t) | \Pi_k C |\Pi^\perp_k |\psi\ra = -\la e_k(t)| A| e_k(t)\ra +\la \psi(t) | \Pi^\perp_k C \Pi_k |e_k\ra
\end{align*}
Applying Lemma~\ref{lemma:AC1} with $|\phi\ra = \Pi^\perp_k |\psi(t)\ra$ and $|\psi\ra = |e_k(t)\ra$ gives for any $1\le p'\le p$: 
\[
|\la e_k(t) | \Pi_k C\Pi^\perp_k |\psi(t)\ra| \le \gamma \sqrt{\la e_k(t)|  A |e_k(t)\ra
\la \psi(t)|\Pi^\perp_k A^2 \Pi_k^\perp |\psi (t)\ra}\overset{\text{Eq.~\eqref{projection_error1}}}{\le} \gamma \la e_k(t)| A |e_k(t)\ra^\frac12 \calM_{p'+2}^{\frac12}(t)k^{-p'/2}
\]
and the latter bound is valid since for any $\tilde p\in[0,p+2]$ one has $\calM_{\tilde p}(t)=\la \psi(t)|A^{\tilde p}|\psi(t)\ra<\infty$. With this bound in mind we write:
\begin{align*}
  \frac{d}{dt}  \la e_k(t)| e_k(t)\ra &\le -2\la e_k(t)| A | e_k(t)\ra  + 2\gamma \la e_k(t)| A| e_k(t)\ra^{\frac12}\calM_{p'+2}^{\frac12}(t)k^{-p'/2}\\
   &=-\la e_k(t)| A | e_k(t)\ra  - \bigl( \la e_k(t)| A| e_k(t)\ra^{\frac12} - \gamma \calM_{p'+2}^{\frac12}(t)k^{-p'/2}\bigr)^2 + \gamma^2 \calM_{p'+2}(t)k^{-p'}\\
&\le -\la e_k(t)| A| e_k(t)\ra +  \gamma^2 \calM_{p'+2}(t)k^{-p'}
\end{align*}
which, after integrating w.r.t. time and recalling that $e(0)=0$ becomes
 \[
  \la e_k(t)| e_k(t)\ra + \int_0^t \la e_k(s)|A| e_k(s)\ra ds  \le  \frac{\gamma^2}{k^{p'}} \int_0^t \calM_{p'+2}(s)ds, \quad 1\le p'\le p
\]
Let us recall Eq.~\eqref{f_moment_p_integral}: taking the limits w.r.t. $k\to\infty$ in Eq.~\eqref{f_moment_p_integral} we derive: for any $p'\in[1,p]$ 
\[
\int_0^t d\tau\, \calM_{p'+2}(\tau) \le \sum_{m=0}^{p'+1} \left( \prod_{r=m+1}^{p'+1} \kappa_r \right) \calM_{m}(0), \qquad \kappa_r=\frac{\gamma_r^2}{4},\quad \gamma_r = 2r\sqrt{6qs}\,J\,\lambda_1^{-1}
\left(1+\frac{2\lambda_N}{\lambda_1}\right)^{(r-1)/2}
%+ \kappa_{p+1} f_{p}(0) + \kappa_{p} \kappa_{p+1} f_{p-1}(0) + \cdots + (\kappa_1 \kappa_2 \cdots \kappa_{p+1}) f_0(0)
\]
Define \(
a:=1+\frac{2\lambda_N}{\lambda_1}
\)
so that
\(
\kappa_r=r^2\,\kappa_1 \,a^{r-1}.
\)
We have
\[
\prod_{j=m+1}^{p'+1}\kappa_j
= \kappa_1^{p'+1-m}\left(\prod_{j=m+1}^{p'+1}j^2\right) a^{\sum_{j=m+1}^{p'+1}(j-1)} =
\kappa_1^{p'+1-m}\left(\frac{(p'+1)!}{m!}\right)^2
a^{\frac{p'(p'+1)-m(m-1)}{2}}
\]
since \(
\sum_{j=m+1}^{p'+1}(j-1)=\sum_{r=m}^{p'} r
=\frac{p'(p'+1)-m(m-1)}{2}
\) and \(\prod_{j=m+1}^{p'+1}j^2=\left(\frac{(p'+1)!}{m!}\right)^2\). Therefore,
\begin{align*}
  \la e_k(t)| e_k(t)\ra + \int_0^t \la e_k(s)| A | e_k(s)\ra ds
  \le \frac{\gamma^2}{k^{p'}}\sum_{m=0}^{p'+1}
\left[
\left(\frac{6qs\,J^2}{\lambda_1^2}\right)^{p'+1-m}
\left(\frac{(p'+1)!}{m!}\right)^2
\left(1+\frac{2\lambda_N}{\lambda_1}\right)^{\frac{p'(p'+1)-m(m-1)}{2}}
\right]
\calM_m(0)
\end{align*}
This proves Eq.~\eqref{eq:kp-reg-error}. %The statement of  Theorem~\ref{thm:regul} is equivalent to 
%Corollary~\ref{corol:limit} with $p=3$, Lemma~\ref{lemma:KEsolution}, and Eq.~(\ref{regularization_error}).

Multiplying Eq.~\eqref{eq:error-eq} by $ e^\perp_k(t)=  \Pi_k^\perp  e(t)$ we get:
\begin{align*}
  \frac12 \frac{d}{dt}  \la e^\perp_k(t)| e^\perp_k(t)\ra &= -\la e^\perp_k(t)| A| e^\perp_k(t)\ra  - \la e^\perp_k(t) | C |\psi\ra  \\
&= -\la e^\perp_k(t)| A| e^\perp_k(t)\ra  - \la  \psi_k | \Pi_k^\perp C |\psi\ra + \la  \psi | \Pi_k^\perp C |\psi\ra
\end{align*}
Applying Lemma~\ref{lemma:AC1} with $|\phi\ra = \Pi^\perp_k |\psi(t)\ra$ and $|\psi\ra = |\psi(t)\ra$ gives for any $1\le p'\le p$: \[
|\la  \psi | \Pi_k^\perp C |\psi\ra| \le \gamma \sqrt{\la \psi(t)|  A |\psi(t)\ra
\la \psi(t)|\Pi^\perp_k A^2 \Pi_k^\perp |\psi (t)\ra}\overset{\text{Eq.~\eqref{projection_error1}}}{\le} \gamma \la \psi(t)| A |\psi(t)\ra^\frac12 \calM_{p'+2}^{\frac12}(t)k^{-p'/2}
\] similarly we get \[
|\la  \psi_k | \Pi_k^\perp C |\psi\ra| \le \gamma \sqrt{\la \psi(t)|  A |\psi(t)\ra
\la \psi_k(t)|\Pi^\perp_k A^2 \Pi_k^\perp |\psi_k (t)\ra}\overset{\text{Eq.~\eqref{projection_error1}}}{\le} \gamma \la \psi(t)| A |\psi(t)\ra^\frac12 \calM_{p'+2}^{\frac12}(t)k^{-p'/2}
\]
so that \begin{align*}
  \frac{d}{dt}  \la e^\perp_k(t)| e^\perp_k(t)\ra &= -2\la e^\perp_k(t)| A| e^\perp_k(t)\ra  +4  \gamma \la \psi(t)| A |\psi(t)\ra^\frac12 \calM_{p'+2}^{\frac12}(t)k^{-p'/2}  \\
\la e^\perp_k(t)| e^\perp_k(t)\ra &+ 2 \int_0^t \la e^\perp_k(s)| A| e^\perp_k(s)\ra ds \le  4 \gamma k^{-p'/2} \left(\int_0^t \la \psi(s)| A |\psi(s)\ra ds\right)^{\frac12}\left(\int_0^t \calM_{p'+2}(s)ds\right)^{\frac12}
\end{align*}
This implies~\eqref{eq:ekperp}.  
\end{proof}

\section[Readout state]
{Readout state\protect\footnote{This section is largely based on the preprint~\cite{bravyi2025quantum}.}}
\label{sec:readout}

Let $|\psi(t)\ra\in \calF_0$ be the solution of the second-quantized Kolmogorov equation  constructed in Section~\ref{sec:Kolmogorov}. 
By construction, the original Kolmogorov equation  Eq.~(\ref{Kolmogorov1}) has a solution
\be
\label{u_via_Phi}
u(t,x) 
= \frac1{\sqrt{\mu(x)}} \sum_{\mm \in \calJ} \la \mm|\psi(t)\ra \Phi_\mm(x),
\ee
where
$\Phi_\mm(x)$ are normalized eigenfunctions for a system of $N$ quantum harmonic oscillators such that  the $i$-th oscillator has mass,
$q^{-1}$, spring constant $q^{-1}\lambda_i^2$, and frequency $\lambda_i$. 
A well-known expression for these eigenfunctions is 
\[
 \Phi_\mm(x)=\sqrt{\mu(x)} \HH_\mm(x),
\]
where 
$\HH_\mm(x)$ are (rescaled) $N$-variate Hermite polynomials defined as follows.
Suppose first $x\in \RR$ is a single variable.
We use probabilist's Hermite polynomials defined as
\[
\herm_n(x) = (-1)^n e^{x^2/2} \frac{d^n}{dx^n} e^{-x^2/2},
\]
where $n\ge 0$ is the degree of $\herm_n(x)$. 
For example, $\herm_0(x)=1$, $\herm_1(x)=x$, $\herm_2(x)=x^2-1$, and
$\herm_3(x)=x^3-3x$. Define
\be
\label{Hermite_normalized}
\mathbb{H}_{\multi{m}}(x) = \prod_{i=1}^N  \frac1{\sqrt{(m_i)!}} \herm_{m_i} \left(x_i\sqrt{2\lambda_i/q}\right),
\ee
where $x\in \RR^N$. The rescaling of variables in $\HH_\mm(x)$ ensures that the eigenfunctions $\Phi_\mm(x)$ form an orthonormal family. Equivalently,
\be
\label{orthogonality_relations}
\int_{\RR^N} dx\, \mu(x) \HH_\nn(x)  \HH_\mm(x) = \delta_{\nn,\mm}.
\ee
 From Eq.~(\ref{u_via_Phi}) one gets
\be
\label{u_via_H}
u(t,x) =  \sum_{\mm \in \calJ} \la \mm|\psi(t)\ra \HH_\mm(x)
\ee

\subsection{Noiseless initial conditions}
\label{sec:readout_challenges}

Suppose first that the SDE Eq.~(\ref{SDE}) has noiseless initial conditions, $X(0)=x$.
 For the sake of argument, suppose a quantum computer can efficiently prepare the exact
solution $|\psi(t)\ra$ of the Kolmogorov equation (up to the normalization). Can one efficiently extract $u(t,x)$ from $|\psi(t)\ra$ ? 
From Eq.~(\ref{u_via_H}) one gets
\be
\label{readout_failed}
u(t,x) = \la \varphi_{out}(x)|\psi(t)\ra,
\ee
where $|\varphi_{out}(x)\ra$ is a ``state" with amplitudes $\la \mm|\varphi_{out}(x)\ra = \HH_\mm(x)$. One may hope to construct a quantum
circuit that prepares a normalized version of $|\varphi_{out}(x)\ra$ and estimates the inner product in Eq.~(\ref{readout_failed}) using the Hadamard test. 
Unfortunately, this approach fails since the ``state" $|\varphi_{out}(x)\ra$  has an infinite norm. 
For example, suppose $x$ is the all-zero vector, $x=0^N$.
By definition, $\HH_\mm(0^N)=\prod_{i=1}^N (1/\sqrt{m_i!}) \herm_{m_i}(0)$. 
Using the identity $\herm_n(0)=(-1)^{n/2} (n-1)!!$ for even $n$ and $\herm_n(0)=0$ for odd $n$, see Ref.~\cite{andrews1999special}, p.~282, one gets
\[
\frac1{n!} \herm_n(0)^2  = \frac{(n-1)!!}{n!!}  \sim 1/\sqrt{n}.
\]
Thus 
\[
\|\varphi_{out}(0^N)\|^2 =  \sum_{\mm\in \ZZ_{\ge 0}^N} \;  \HH_\mm(0^N)^2\sim\left( \sum_{n=0}^\infty 1/\sqrt{n}\right)^N=\infty.
\]
Using  Mehler's formula one can check that the sum $\sum_{\mm\in \ZZ_{\ge 0}^N} \HH_\mm(x)^2$ diverges for any $x\in \RR^N$.

\subsection{Noisy initial conditions}
\label{sec:readout_solution}

Recall that our quantum algorithm aims to approximate the expected value $v(t,x)=\EE_z u(t,x+z)$, where $z\in \RR^N$ is a random  vector
that models noise in the initial conditions. 
Each variable $z_i$ is drawn from $\calN(0,q/(2\lambda_i))$. Equivalently, 
$z$ is drawn from the steady-state distribution $\mu$.
We claim that 
for any $x\in \RR^N$ there there exists a state 
$|\psi_{out}(x)\ra\in \calF$ such that 
\be
\label{readout_eq1}
v(t,x) = \la \psi_{out}(x)|\psi(t)\ra
\ee
and
\be\begin{aligned}
\label{readout_state_norm}
    \|\psi_{out}(x)\|&=  \exp{\left[ q^{-1} 
    \sum_{i=1}^N \lambda_i x_i^2 \right]}.
\end{aligned}\ee
We shall refer to $|\psi_{out}(x)\ra$ as the {\em readout state}.
Indeed, from Eq.~(\ref{u_via_H}) one gets
\[
v(t,x)=
\int_{\RR^N}dz\,  \mu(z) u(t,x+z)
=\sum_{\mm\in \calJ} \la \mm|\psi(t)\ra \int_{\RR^N}dz\,  \mu(z)  \HH_\mm(x+z).
\]
The well-known  umbral identity for  single-variate Hermite polynomials reads as 
\[
\herm_n(x+z) = \sum_{k=0}^n {n \choose k} x^{n-k} \herm_k(z)
\]
for any $x,z\in \RR$. Using this identity and the orthogonality relations Eq.~(\ref{orthogonality_relations})
 gives
\[
\int_\RR dz \, e^{-z^2/2} \herm_n(x+z)  = \int_\RR dz\,  e^{-z^2/2} \herm_n(x+z)   \herm_0(z) = (2\pi)^{1/2} x^n
\]
for any $x\in \RR$.
As a consequence,
\[
\int_{\RR^N} dz \, \mu(z) \HH_\mm(x+z) = \prod_{i=1}^N \frac1{\sqrt{(m_i !)}} \left( x_i \sqrt{\frac{2 \lambda_i}{q}}\right)^{m_i}
\]
for any $x\in \RR^N$. We arrive at
\[
v(t,x) = \sum_{\mm\in \calJ} \psi_\mm(t)  \prod_{i=1}^N \frac1{\sqrt{(m_i !)}} \left( x_i \sqrt{\frac{2 \lambda_i}{q}}\right)^{m_i} = \la \psi_{out}(x)|\psi(t)\ra,
\]
where
\be
\label{psi_out}
|\psi_{out}(x)\ra = \sum_{\mm \in \ZZ_{\ge 0}^N}\; 
|\mm\ra \cdot  \prod_{i=1}^N \frac1{\sqrt{(m_i !)}} \left( x_i \sqrt{\frac{2 \lambda_i}{q}}\right)^{m_i} \in \calF.
\ee
We note that $\psi_{out}(x)$ coincides with the coherent state embedding of $x$
discussed in~\cite{engel2021linear}. 
Finally, the norm of $\psi_{out}$ can be 
computed by observing that $\psi_{out}$ is a tensor product of  coherent states associated with each variable $x_i$, the norm is multiplicative under the tensor product,
and $\|\psi_{out}(x)\|^2= e^{2q^{-1} \lambda_1 x_1^2}$
in the case $N=1$. This gives Eq.~(\ref{readout_state_norm}).

In Section~\ref{sec:algorithm} we approximate $v(t,x)$ by the inner product
 $\la \psi_{out}(x)|\Pi_k|\psi_k(t)\ra$, where
  $\psi_k(t)$ solves the  regularized Kolmogorov equation
  with a suitable cutoff $k$. 
 We shall see that normalized versions of the states $\Pi_k|\psi_{out}(x)\ra$ and $\Pi_k|\psi_k(t)\ra$ can be
efficiently approximated by quantum circuits. 

As a side remark, we note that the expected value $u(t,x)$ with a noise-free initial condition also admits a representation in terms of the readout state $\psi_{out}(x)$ defined in Eq.~(\ref{psi_out}).
Using the identity\footnote{This identity follows from
the well-known formula $\herm_n(x) = \EE_z (x+iz)^n$, where  $z\sim \calN(0,1)$, see for instance Ref.~\cite{andrews1999special}, p.~280.}
\[
\HH_\mm(x) = \int_{\RR^N} dz\, \mu(z)   \prod_{j=1}^N  \frac1{\sqrt{(m_j)!}} \left((x_j+ i z_j)\sqrt{2\lambda_j/q}\right)^{m_j}
\]
with $i=\sqrt{-1}$ one gets
\be
\label{u_coherent_state_rep}
u(t,x) =\int_{\RR^N} dz\, \mu(z) \la \psi_{out}(x+iz)|\psi(t)\ra.
\ee
Unfortunately, we were not able to convert Eq.~(\ref{u_coherent_state_rep}) into an efficient readout procedure capable of approximating $u(t,x)$. 

\section{Block-encoding of the projected Kolmogorian}
\label{sec:block_encoding}

Let $k$ be the regularization cutoff. Our quantum algorithm simulates  time evolution under the Kolmogorian 
projected onto 
the low-dissipation subspace 
\[
\calH_k = \mathrm{span}(|\mm\ra \, : \, \mm \in \calJ_k),
\]
where 
\[
\calJ_k =\{ \mm \in \ZZ_{\ge 0}^N \, : \,  \mm\ne 0^N \quad \mbox{and} \quad \sum_{i=1}^N \lambda_i m_i\le k\}.
\]
Let $\Pi_k = \sum_{\mm \in \calJ_k} |\mm\ra\la \mm|$ be the projector onto $\calH_k$
and $C_k=\Pi_k C\Pi_k$ be the projected Kolmogorian $C$.
In this section we prove the following. 
\begin{theorem}[\bf block encoding]
\label{thm:block_encoding}
There exists an encoding of $\calH_k$ by $Q\approx k\lambda_1^{-1}\log_2{N}$ qubits
and a quantum circuit $U$ of size $poly(Q)$ 
such that for all $\nn,\mm\in \calJ_k$ 
\be
\la \nn|C_k|\mm\ra = \alpha \la \nn,0_{anc}|U|\mm,0_{anc}\ra
\ee
where
\be
\alpha =  4(qs)^{1/2} J k^3  \lambda_1^{-7/2}.
\ee
Here $|0_{anc}\ra$ is the all-zero state of an ancillary register of $poly(Q)$ qubits. 
The circuit $U$ makes one query to the
oracles $\voracle{c}$, $\poracle{c}$  and $O(k\lambda_1^{-1})$ queries to the oracle $\voracle{\lambda}$.
\end{theorem}
We construct the qubit encoding of $\calH_k$ in Section~\ref{app:encoding}.
This section is mostly based on~\cite{engel2021linear,tanaka2023polynomial}.
Our block encoding of $C_k$ relies on certain transition operators for cubic bosonic terms.
These operators  are constructed in Section~\ref{app:transition}.
Finally, we combine these ingredients to prove Theorem~\ref{thm:block_encoding} in Section~\ref{app:block_encoding_proof}.

\subsection{Encoding of multi-indices by bit strings}
\label{app:encoding}

A  multi-index $\mm \in \calJ$ can be considered as a multiset\footnote{Recall that a multiset is
a set that may contain multiple copies of each element. The cardinality of a multiset is 
the sum of the multiplicities of all its elements.}
 $M \subseteq \{1,2,\ldots,N\}$
of cardinality $|M|=|\mm|=\sum_{i=1}^N m_i$ that contains $m_i$ copies of an integer $i$,  
\[
M = \{ \underbrace{1,1,\ldots,1}_{m_1},  \underbrace{2,2,\ldots,2}_{m_2}, \ldots,  \underbrace{N,N,\ldots,N}_{m_N}\}.
\]
Recall that $\calJ_k=\{\mm \in \calJ\, : \, \sum_{i=1}^N \lambda_i m_i \le k\}$, where $k$ is the cutoff parameter of
our regularization scheme. 
From $\lambda_1=\min_i \lambda_i>0$ one gets
\[
|\mm|=\sum_{i=1}^N m_i \le \lambda_1^{-1} \sum_{i=1}^N \lambda_i m_i \le \frac{k}{\lambda_1}.
\]
Thus
\be
\label{max_number_particles}
p \equiv  \max_{\mm \in \calJ_k}\; |\mm| \le \frac{k}{\lambda_1}.
\ee
We shall encode a multi-index $\mm \in\calJ_k$ by a 
$p$-tuple $(M_1,\ldots,M_p)$ such that 
$M_j \in \{1,2,\ldots,N\} \cup \{ \infty\}$ is the $j$-th smallest  element of  $M$ 
if $1\le j\le |\mm|$, and $M_j=\infty$ if $|\mm|<j\le p$. 
We shall use the convention $i< \infty$ for any integer $i$. 
As a concrete example, suppose $N=7$ and $p=4$. Consider a multi-index
$\mm = (2,0,0,0,0,0,1)$. The corresponding multiset $M=\{1,1,7\}$
has cardinality $|M|=3$. Thus 
$M_1=M_2=1$, $M_3=7$,  and $M_4=\infty$. The tuple $(M_1,\ldots,M_p)$ uniquely identifies $\mm$ since
$m_i$ is equal to the number of times an integer $i$ appears in $(M_1,\ldots,M_p)$.
Using the binary encoding of integers
one can identify $M_j$ with a bit string of length $\ell$, 
where  $\ell$ is the smallest integer such that $2^\ell \ge N+1$ (the symbol $\infty$
is encoded by the all-zero bit string).
Then $(M_1,\ldots,M_p)$ is a bit string of length $\ell p$
that uniquely identifies $\mm$.  The total number of qubits used to express the projected subspace
 $\calH_k$ is
\[
Q = \ell p \approx  k\lambda_1^{-1}  \log_2{(N)}.
\]
In the above example $N=7$ and thus $\ell=3$.
A multi-index $\mm = (2,0,0,0,0,0,1)$  is encoded by a qubit basis state
\[
 |100\ra \otimes |100\ra \otimes|111\ra \otimes  |000\ra
\]
where $p=4$ registers store the binary encodings of  $M_1=1$,  $M_2=1$, $M_3=7$, and $M_4=\infty$
 respectively (here the least significant bit is leftmost). In general, if $(M_1,\ldots,M_p)$ is an encoding of a multi-index
 $\mm\in \calJ_k$ then
\be
\label{encoded_lambda}
\sum_{i=1}^N \lambda_i m_i = \sum_{a=1}^p \lambda_{M_a}
\ee
if we use the convention $\lambda_\infty=0$.

\subsection{Transition operators for cubic bosonic terms}
\label{app:transition}

Consider  a Fock basis state $|\mm\ra$ with some multi-index  $\mm\in\calJ_k$
and a  cubic bosonic operator
\[
a_j^\dag a_\ell^\dag a_i.
\]
Here some of the indices $i,j,\ell$ may coincide. 
We have
\[
a_i|\mm\rangle=\sqrt{m_i}\,|\mm-e^i\rangle
\]
if $m_i\ge 1$ and $a_i|\mm\rangle=0$ if $m_i=0$. Suppose $m_i\ge 1$. Then 
\[
a_\ell^\dag|\mm-e^i\rangle
=
\sqrt{m_\ell+1-\delta_{i\ell}}\,
|\mm-e^i+e^\ell\rangle,
\]
and finally
\[
a_j^\dag|\mm-e^i+e^\ell\rangle
=
\sqrt{m_j+1-\delta_{ij}+\delta_{j\ell}}\,
|\mm-e^i+e^\ell+e^j\rangle.
\]
The final state $a_j^\dag a_\ell^\dag a_i|\mm\ra$ is proportional to  $|\tau_{ij\ell}(\mm)\ra$,
where $\tau_{ij\ell} \, : \, \ZZ^N \to \ZZ^N$ is the forward  transition map defined as
\[
\tau_{ij\ell}(\mm) = \mm - e^i + e^j + e^\ell.
\]
Thus  for all   $\nn,\mm \in \calJ$  one has 
\[
\langle \nn| \Pi_k a_j^\dag a_\ell^\dag a_i\Pi_k|\mm\rangle=
\rho_{ij\ell}(\mm)
\Delta_{ij\ell}(\nn,\mm),
\]
where
\[
\rho_{ij\ell}(\mm)
=\left\|a_j^\dag a_\ell^\dag a_i|\mm\rangle\right\|=\sqrt{
m_i\,
\bigl(m_\ell+1-\delta_{i\ell}\bigr)\,
\bigl(m_j+1-\delta_{ij}+\delta_{j\ell}\bigr)}
\]
and
\[
\Delta_{ij\ell}(\nn,\mm)=
\begin{cases}
1 & \mbox{if } m_i\ge 1, \ \nn=\tau_{ij\ell}(\mm), \
\mbox{ and }\nn,\mm\in\calJ_k,\\
0  & \mbox{otherwise.}
\end{cases}
\]
Define a transition operator $T\, : \, \calH_k \otimes (\CC^N)^{\otimes 3} \to \calH_k^{\otimes 2} \otimes  (\CC^N)^{\otimes 3} \otimes \CC^2$ such that 
\[
T |\mm,i,j,\ell\ra = \begin{cases}
|\tau_{ij\ell}(\mm),\mm,i,j,\ell,0\ra & \mbox{if} \quad m_i\ge 1 \quad \mbox{and} \quad \tau_{ij\ell}(\mm)\in \calJ_k,\\
|\mathsf{junk}(\mm), \mm,i,j,\ell,1\ra & \mbox{otherwise}.\\
\end{cases}
\]
Here $\mathsf{junk}(\mm)\in \calJ_k$ is an arbitrary basis state and the last single-qubit register stores a flag indicating whether the transition is valid.
Using the encoding of multi-indices by qubits from the previous section it is straightforward to implement
$T$ by a quantum circuit that uses $poly(k,\log{N})$ gates and ancillary qubits, and makes at most
$p$ queries to the oracle $\voracle{\lambda}$, where $p$ is the maximum number of particles defined in Eq.~(\ref{max_number_particles}).
The latter are needed to compute the  eigenvalue $\lambda_\nn = \sum_{i=1}^N \lambda_i n_i$
for $\nn=\tau_{ij\ell}(\mm)$ which determines whether $\nn \in \calJ_k$ (depending on whether $\lambda_\nn \le k$).

One can similarly implement a backward transition operator
$T'\,:\, \calH_k\otimes(\CC^N)^{\otimes 3}
\to
\calH_k^{\otimes 2}\otimes(\CC^N)^{\otimes 3}\otimes\CC^2$
defined as follows. Let
\[
\mm=\tau^{-1}_{ij\ell}(\nn)=\nn+e^i-e^j-e^\ell .
\]
Then
\[
T'|\nn,i,j,\ell\rangle
=
\begin{cases}
|\nn,\mm,i,j,\ell,0\rangle,
&
\mbox{if } \mm\in\calJ_k \mbox{ and } m_i\ge 1,\\
|\nn,\mathsf{junk}(\nn),i,j,\ell,1\rangle,
&
\mbox{otherwise}.
\end{cases}
\]
Here \(\mathsf{junk}(\nn)\in\calJ_k\) is an arbitrary basis state.

\subsection{Proof of Theorem~\ref{thm:block_encoding}}
\label{app:block_encoding_proof}

By definition,
\[
C_k=C_k^\uparrow-C_k^\downarrow,
\]
where
\[
C_k^\uparrow=
(q/2)^{1/2}
\sum_{i,j,\ell=1}^N
(\lambda_i\lambda_j\lambda_\ell)^{-1/2}
\lambda_i c_{ij\ell}\,
\Pi_k a_j^\dag a_\ell^\dag a_i\Pi_k
\]
and $C_k^\downarrow=(C_k^\uparrow)^\dag$.

It suffices to construct a  quantum circuit block encoding  $C_k^\uparrow$.
Then the conjugate circuit block encodes $C_k^\downarrow$.
The block encodings of $C_k^\uparrow$ and $C_k^\downarrow$ can be combined
to get a block encoding of $C_k=C_k^\uparrow- C_k^\downarrow$ using the standard techniques, 
see Lemma~52 in~\cite{gilyen2019quantum}.
Define
\[
\omega_{ij\ell}(\mm)=
(q/2)^{1/2}
(\lambda_i\lambda_j\lambda_\ell)^{-1/2}
\lambda_i\,
\rho_{ij\ell}(\mm).
\]
Then
\be
\label{Cplus_matrix_elem}
\langle \nn|C_k^\uparrow|\mm\rangle
=
\sum_{i,j,\ell=1}^N
c_{ij\ell}\,
\omega_{ij\ell}(\mm)\,
\Delta_{ij\ell}(\nn,\mm).
\ee

Our proof of Theorem~\ref{thm:block_encoding} is based on block encoding of Gram matrices, see Lemma~47 of ~\cite{gilyen2019quantum}.
Let $\calH_{anc}$ be an ancillary register of $O(1)$ qubits. 
We shall construct isometries 
\[
U_L,U_R \, : \, \calH_k \otimes \calH_{\mathrm{anc}}^{\otimes 2} \to \calH_k^{\otimes 2} \otimes (\CC^N)^{\otimes 3}\otimes \calH_{\mathrm{anc}}^{\otimes 2}
\]
such that  for any Fock basis vectors $|\nn\ra,|\mm\ra\in \calH_k$ one has
\be
\label{ULUR}
\la \nn|C_k^\uparrow|\mm\ra =\alpha_L \alpha_R \la \nn,0_{\mathrm{anc}},0_{\mathrm{anc}}| U_L^\dag U_R |\mm,0_{\mathrm{anc}},0_{\mathrm{anc}}\ra,
\ee
for certain coefficients   $\alpha_L,\alpha_R>0$.
The isometry $U_R$ must obey 
\be
\label{U_R}
U_R|\mm,0_{\mathrm{anc}},0_{\mathrm{anc}}\ra = \frac1{\alpha_R} \sum_{\nn\in \calJ_k} \sum_{i,j,\ell} c_{ij\ell} \omega_{ij\ell}(\mm) \Delta_{ij\ell}(\nn,\mm)
|\nn,\mm,i,j,\ell,0_{\mathrm{anc}},0_{\mathrm{anc}}\ra+ |\mathsf{junk}_R(\mm)\ra
\ee
for all $\mm \in\calJ_k$. Here
$|\mathsf{junk}_R(\mm)\ra$ are some irrelevant junk states. We require that $|\mathsf{junk}_R(\mm)\ra$ belongs to the subspace
where  the first ancilla register is orthogonal to $0_{\mathrm{anc}}$ and 
 the second ancilla register remains in the state $0_{\mathrm{anc}}$.
The isometry $U_L$ must obey
\be
\label{U_L}
U_L|\nn,0_{\mathrm{anc}},0_{\mathrm{anc}}\ra = \frac1{\alpha_L} \sum_{\mm \in \calJ_k} \sum_{i,j,\ell} \chi_{ij\ell}\Delta_{ij\ell}(\nn,\mm)
|\nn,\mm,i,j,\ell,0_{\mathrm{anc}},0_{\mathrm{anc}}\ra  + |\mathsf{junk}_L(\nn)\ra
\ee
for all $\nn\in \calJ_k$. Here  $|\mathsf{junk}_L(\nn)\ra$ are some irrelevant junk states. 
We require that  $|\mathsf{junk}_L(\nn)\ra$ belongs to the subspace where 
the first ancilla register remains in the state $0_{\mathrm{anc}}$  and the second ancilla register is orthogonal to $0_{\mathrm{anc}}$.
Recall that $\chi_{ij\ell}\in \{0,1\}$ is the indicator function of $c_{ij\ell}$ such that $\chi_{ij\ell}=1$ if $c_{ij\ell}\ne 0$ and $\chi_{ij\ell}=0$ if $c_{ij\ell}=0$.
The assumptions about the junk states ensure that 
\be
\label{left_right_junk}
\la \mathsf{junk}_L(\nn)|\mathsf{junk}_R(\mm)\ra = 0
\ee
for all $\nn,\mm$ and the junk states are orthogonal to any state in which both ancillary registers are in the state $0_{\mathrm{anc}}$.
It follows that the junk states do not contribute to the matrix element in Eq.~(\ref{ULUR}). Thus we get
\[
\la \nn,0_{\mathrm{anc}},0_{\mathrm{anc}}| U_L^\dag U_R |\mm,0_{\mathrm{anc}},0_{\mathrm{anc}}\ra
=\frac1{\alpha_R\alpha_L} \sum_{i,j,\ell} c_{ij\ell} \omega_{ij\ell}(\mm) \Delta_{ij\ell}(\nn,\mm) = \frac1{\alpha_R\alpha_L}  \la \nn|C_k^\uparrow|\mm\ra.
\]
Here the second equality follows from  Eq.~(\ref{Cplus_matrix_elem}). This proves Eq.~(\ref{ULUR}).
Combining block encodings of $C_k^\uparrow$ and its conjugate using Lemma~52 of~\cite{gilyen2019quantum}
gives the desired block encoding of $C_k$ with $\alpha=2\alpha_R\alpha_L$.

We  construct the desired  isometries $U_R$ and $U_L$ in the following two lemmas.
\begin{lemma}[\bf Isometry \(U_R\)]
Let
\[
\alpha_R = O\!\left(q^{1/2}J k^2\lambda_1^{-5/2}\right).
\]
Then the isometry $U_R$ satisfying Eq.~(\ref{U_R})
can be implemented by  a circuit of size $poly(p,\log{N})$ that makes one query to the oracle $\voracle{c}$
and $O(p)$ queries to the oracle $\voracle{\lambda}$.
Here $p$ is the maximum number of particles defined in Eq.~(\ref{max_number_particles}).
\end{lemma}

\begin{proof}
Fix an input basis state $|\mm\ra\in \calH_k$. We shall write $i\in \mm$ if $m_i\ne 0$.
Let $\mathrm{supp}(\mm) \subseteq [N]$ be the support of $\mm$ defined as the union
of all indices $i\in \mm$. 
Let $|\mm|_0=|\mathrm{supp}(\mm)|$.
The desired isometry $U_R$ is constructed in several steps.

\noindent
{\em  Step 1: Prepare a superposition over the support of $\mm$.}
Given the tuple $(M_1,\ldots,M_p)$ encoding $\mm$, 
the support of $\mm$  is simply the union of $M_1,\ldots,M_p$, ignoring the $\infty$ symbol.
One can compute the size of the support $|\mm|_0$ and prepare the uniform superposition of indices $i\in \mm$ by a circuit
of size $poly(p,\log{N})$. This gives a state 
\[
|\psi_1(\mm)\ra = |\mm|_0^{-1/2} \sum_{i\in \mm } |\mm\ra |i\ra.
\]
Note that $|\mm|_0\ne 0$ since the input Hilbert space $\calH_k$ excludes the vacuum state.

\noindent
{\em  Step 2: Query the oracle $\voracle{c}$.}
Apply $\voracle{c}$ to the register storing $|i\ra$.  This gives a state
\[
|\psi_2(\mm)\ra=
|\mm|_0^{-1/2} J^{-1} 
\sum_{i\in \mm }
\sum_{j,\ell=1}^N
c_{ij\ell} |\mm,i,j,\ell,0\ra
+
|\mathsf{junk}(\mm)\ra,
\]
where $|\mathsf{junk}(\mm)\ra$ are  junk states that belong to the subspace where the ancillary qubit is $|1\ra$.

\noindent
{\em Step 3:  Apply the transition operator. }
Applying the transition operator $T$ constructed above to the registers storing $|\mm,i,j,\ell\ra$ gives a state
\[
|\psi_3(\mm)\ra=|\mm|_0^{-1/2} J^{-1}
\sum_{i\in \mm}
\sum_{\nn \in \calJ_k} 
\sum_{j,\ell=1}^N
c_{ij\ell} \Delta_{ij\ell}(\nn,\mm)  |\nn,\mm,i,j,\ell,0,0\ra
+
|\mathsf{junk}(\mm)\ra.
\]
Here $|\mathsf{junk}(\mm)\ra$ are junk states that belong to the subspace where the ancillary two-qubit register  is orthogonal to  $|0,0\ra$.
By definition, $\Delta_{ij\ell}(\nn,\mm)=0$ if $m_i=0$. Hence we can extend the sum over $i\in \mm$ to the sum over all indices $i$ arriving at
\[
|\psi_3(\mm)\ra=|\mm|_0^{-1/2} J^{-1}
\sum_{\nn \in \calJ_k} 
\sum_{i,j,\ell=1}^N
c_{ij\ell} \Delta_{ij\ell}(\nn,\mm)  |\nn,\mm,i,j,\ell,0,0\ra
+
|\mathsf{junk}(\mm)\ra.
\]

\noindent
{\em Step 4:  Coherently load the weight $\omega_\tau(\mm)$.}
On the non-junk branch,  use the encoding of $\mm$
to compute the occupation numbers
$m_i,m_j,m_\ell$  together with the support size $|\mm|_0$,
and then compute the number
\[
\beta_{ij\ell}(\mm) =\alpha_R^{-1} J  |\mm|_0^{1/2} \omega_{ij\ell}(\mm).
\]
We shall see that for the chosen value of $\alpha_R$ one has $|\beta_{ij\ell}(\mm) |\le 1$ for all $i,j,\ell$ and $\mm\in \calJ_k$.
Apply a controlled $Y$-rotation to a fresh ancilla qubit  prepared in $|0\ra$. It implements a map
\[
|0\ra \mapsto \beta_{ij\ell}(\mm)   |0\ra + \sqrt{1-( \beta_{ij\ell}(\mm))^2} |1\ra.
\]
Including the state $|1\ra$ into the junk branch results in a state
\[
|\psi_4(\mm)\ra=
\frac1{\alpha_R}
\sum_{\nn \in \calJ_k} 
\sum_{i,j,\ell=1}^N
c_{ij\ell} \omega_{ij\ell}(\mm) \Delta_{ij\ell}(\nn,\mm)  |\nn,\mm,i,j,\ell,0,0,0\ra
+
|\mathsf{junk}(\mm)\ra.
\]
Here $|\mathsf{junk}(\mm)\ra$ are junk states that belong to the subspace where the ancillary three-qubit register  is orthogonal to  $|0,0,0\ra$.
The state $|\psi_4(\mm)\ra$ is the desired output of $U_R$  for a suitable bookkeeping on the ancilla registers
that satisfies Eq.~(\ref{left_right_junk}).

It remains to check that  $|\beta_{ij\ell}(\mm) |\le 1$ for all $i,j,\ell$ and $\mm\in \calJ_k$.
Using the definition of \(\omega_{ij\ell}(\mm)\), one gets
\[
\beta_{ij\ell}(\mm)
=
\alpha_R^{-1}J|\mm|_0^{1/2}
(q/2)^{1/2}
\sqrt{\frac{\lambda_i}{\lambda_j\lambda_\ell}}\,
\rho_{ij\ell}(\mm).
\]
Since \(\mm\in\calJ_k\), one has
\[
\lambda_i m_i\le k,
\qquad
|\mm|_0\le p,
\]
where \(p\) is the maximum number of particles. Also, by the definition of
\(\rho_{ij\ell}(\mm)\),
\[
\rho_{ij\ell}(\mm)
\le
\sqrt{m_i(p+1)(p+2)}.
\]
Therefore
\[
\omega_{ij\ell}(\mm)
\le
(q/2)^{1/2}
\sqrt{\frac{\lambda_i m_i}{\lambda_j\lambda_\ell}}\,
\sqrt{(p+1)(p+2)}
\le
(q/2)^{1/2}\sqrt{k}\,(p+2)\lambda_1^{-1}.
\]
Thus
\[
|\beta_{ij\ell}(\mm)|
\le
\alpha_R^{-1}J\sqrt{p}\,
(q/2)^{1/2}\sqrt{k}\,(p+2)\lambda_1^{-1}.
\]
If \(k\ge\lambda_1\), then \(p=\lceil k/\lambda_1\rceil\le 2k/\lambda_1\), and hence
\[
|\beta_{ij\ell}(\mm)|
\le
O\!\left(
\alpha_R^{-1}Jq^{1/2}k^2\lambda_1^{-5/2}
\right).
\]
Choosing  large enough  constant in the definition $\alpha_R=O\!\left(q^{1/2}Jk^2\lambda_1^{-5/2}\right)$
ensures \(|\beta_{ij\ell}(\mm)|\le 1\).
\end{proof}

\begin{lemma}[\bf Isometry $U_L$]
\label{lemma:UL}
Let
\[
\alpha_L =s^{1/2} k  \lambda_1^{-1}.
\]
Then the isometry $U_L$ satisfying Eq.~(\ref{U_L})
can be implemented by  a circuit of size $poly(p,\log{N})$ that makes one query to the oracle $\poracle{c}$
and $O(p)$ queries to the oracle $\voracle{\lambda}$. Here $p$ is the maximum number of particles defined in Eq.~(\ref{max_number_particles}).
\end{lemma}

\begin{proof}
Fix an input basis state $|\nn\ra\in \calH_k$. The desired isometry $U_L$ is constructed in several steps.

\noindent
{\em  Step 1: Prepare  a superposition over the support of $\nn$.}
Given the tuple $(N_1,\ldots,N_p)$ encoding $\nn$ (using the encoding of Section~\ref{app:encoding}),
the support of $\nn$  is simply the union of $N_1,\ldots,N_p$, ignoring the $\infty$ symbol.
One can compute the size of the support $|\nn|_0$ and prepare two copies of  the uniform superposition of indices in $\nn$. 
This gives a state 
\[
|\psi_1(\nn)\ra = |\nn|_0^{-1} \sum_{j,\ell \in \nn} \;  |\nn,j,\ell\ra
\]
that can be prepared by a circuit of size $poly(p,\log{N})$. 
Note that $|\nn|_0\ne 0$ since the input Hilbert space $\calH_k$ excludes the vacuum state.

\noindent
{\em  Step 2: Query the oracle $\poracle{c}$.}
Apply $\poracle{c}$ to the register storing $|j,\ell\ra$.  This gives a state
\[
|\psi_2(\nn)\ra=|\nn|_0^{-1}
 s^{-1/2} \sum_{i=1}^N \sum_{j,\ell \in \nn}  \chi_{ij\ell}  |\nn,i,j,\ell,0\ra
+
|\mathsf{junk}(\nn)\ra,
\]
where $|\mathsf{junk}(\nn)\ra$ are  junk states that belong to the subspace where the ancillary qubit is $|1\ra$.

\noindent
{\em Step 3:  Apply the backward transition operator.}
Applying the backward transition operator $T'$ constructed above to the registers storing $|\nn,i,j,\ell\ra$
gives a state
\[
|\psi_3(\nn)\ra=|\nn|_0^{-1}
 s^{-1/2} \sum_{\mm \in \calJ_k} 
 \sum_{i=1}^N \sum_{j,\ell \in \nn}  
 \chi_{ij\ell}  \Delta_{ij\ell}(\nn,\mm) |\nn,\mm,i,j,\ell,0,0\ra
+
|\mathsf{junk}(\nn)\ra.
\]
Here $|\mathsf{junk}(\nn)\ra$ are junk states that belong to the subspace where the ancillary two-qubit register  is orthogonal to  $|0,0\ra$.
If $j\notin \nn$ or $\ell \notin \nn$ then  $\Delta_{ij\ell}(\nn,\mm)=0$
since $\nn=\tau_{ij\ell}(\mm)$ implies $j\in \nn$ and $\ell\in \nn$.
Hence we can extend the sum over indices $j,\ell \in \nn$
to all pairs of indices arriving at
\[
|\psi_3(\nn)\ra=|\nn|_0^{-1}
 s^{-1/2} \sum_{\mm \in \calJ_k} 
 \sum_{i,j,\ell=1}^N 
 \chi_{ij\ell}  \Delta_{ij\ell}(\nn,\mm) |\nn,\mm,i,j,\ell,0,0\ra
+
|\mathsf{junk}(\nn)\ra.
\]

\noindent
{\em Step 4:  Cancel the $\nn$-dependent normalization factor.}
On the non-junk branch,  use the encoding of $\nn$
to compute the support size $|\nn|_0$
and then compute the number
\[
\beta(\nn) =\alpha_L^{-1} s^{1/2} |\nn|_0
\]
We shall see that for the chosen value of $\alpha_L$ one has $|\beta(\nn) |\le 1$ for all  $\nn\in \calJ_k$.
Apply a controlled $Y$-rotation to a fresh ancilla qubit  prepared in $|0\ra$. It implements a map
\[
|0\ra \mapsto \beta(\nn)   |0\ra + \sqrt{1-( \beta(\nn))^2} |1\ra.
\]
Including the state $|1\ra$ into the junk branch results in a state
\[
|\psi_4(\nn)\ra=
\frac1{\alpha_L}
 \sum_{\mm \in \calJ_k} \sum_{i,j,\ell=1}^N \chi_{ij\ell}  \Delta_{ij\ell}(\nn,\mm) |\nn,\mm,i,j,\ell,0,0,0\ra
+
|\mathsf{junk}(\nn)\ra.
\]
Here $|\mathsf{junk}(\nn)\ra$ are junk states that belong to the subspace where the ancillary three-qubit register  is orthogonal to  $|0,0,0\ra$.
The state $|\psi_4(\nn)\ra$ is the desired output of $U_L$   for a suitable bookkeeping on the ancilla registers that satisfies Eq.~(\ref{left_right_junk}).

It remains to check that  $|\beta(\nn) |\le 1$ for all $\nn\in \calJ_k$.
Since $|\nn|_0\le p\le k\lambda_1^{-1}$, one has
\[
\beta(\nn)  =\alpha_L^{-1} s^{1/2} |\nn|_0  \le \alpha_L^{-1} s^{1/2}  p \le \alpha_L^{-1} s^{1/2} k \lambda_1^{-1} = 1
\]
for the chosen $\alpha_L$. 
\end{proof}

The above shows that the circuit $U_L^\dag U_R$ block encodes $C_k^\uparrow$ with a normalization  $\alpha_R \alpha_L$.
As a consequence,  the circuit $U_R^\dag U_L$ block encodes $C_k^\downarrow=(C_k^\uparrow)^\dag$ with a normalization  $\alpha_R \alpha_L$.
Combining block encodings of $C_k^\uparrow$ and $C_k^\downarrow$ using  Lemma~52 of~\cite{gilyen2019quantum}
gives a block encoding of $C_k=C_k^\uparrow- C_k^\downarrow$ with a normalization 
\[
\alpha=2\alpha_R \alpha_L =  4(qs)^{1/2} J k^3  \lambda_1^{-7/2}
\]
and the same query complexity (provides that a controlled oracle unitary counts as a single query).

\section{Quantum algorithm}
\label{sec:algorithm}

Let $k$ be a fixed regularization cutoff. 
Our  algorithm simulates the regularized Kolmogorov equation projected onto the low-dissipation subspace $\calH_k$.
Let $\Pi_k$ be the projector onto $\calH_k$.
Define a quantity
\be
\label{v1_approximation}
v_1(t,x) =  \la  \psi_{out}(x)|\Pi_k|\psi_k(t)\ra=\la \psi_{out}(x)|\Pi_ke^{t(-A + \Pi_kC\Pi_k)} \Pi_k|\psi(0)\ra.
\ee
In this section 
we construct quantum circuits $U_{init}$, $U_{sim}$, and $U_{out}$ 
such that 
$U_{init}$ prepares a normalized state $\Pi_k |\psi(0)\ra/\alpha_{init}$,
$U_{sim}$ block encodes a normalized time evolution operator $e^{t(-A+\Pi_k C\Pi_k)}/\alpha_{sim}$ restricted to $\calH_k$,
and $U_{out}$ prepares a normalized state $\Pi_k |\psi_{out}(x)\ra/\alpha_{out}$.
Here $\alpha_{init},\alpha_{sim}$, $\alpha_{out}>0$ are normalization factors defined below.
These circuits act on the Hilbert space $\calH_k$
expressed using $O(k\lambda_1^{-1} \log{N})$ qubits, as described in Section~\ref{sec:block_encoding},
and use  $poly(k,\lambda_1^{-1},\log{N})$ ancillary qubits.
Then 
\[
v_1(t,x) \approx \alpha \la {\bf 0}| U_{out}^{-1} U_{sim} U_{init}|{\bf 0}\ra, \qquad \alpha= \alpha_{init} \alpha_{sim} \alpha_{out}.
\]
Here $|{\bf 0}\ra$ is the all-zero qubit state.
The matrix element $\la {\bf 0}| U_{out}^{-1} U_{sim} U_{init}|{\bf 0}\ra$ can be estimated with a small additive error
using the Hadamard test. 
Finally, the desired expected value of the observable can be expressed as 
\begin{align}
v(t,x) & = \la \psi_{out}(x)|\psi(t)\ra =   \la \psi_{out}(x)|\Pi_k|\psi(t)\ra +  \la \psi_{out}(x)|\Pi_k^\perp |\psi(t)\ra  \nonumber \\
&= v_1(t,x) +   \la \psi_{out}(x)|\Pi_k|\psi(t)-\psi_k(t)\ra +   \la \psi_{out}(x)|\Pi_k^\perp |\psi(t)\ra. \label{v_approximation_eq1}
\end{align}
We shall choose $k$ large enough that the error introduced by the last two terms is negligible.
This error can be upper bounded using Theorem~\ref{thm:regul}.
For convenience of the reader, the following table collects all parameters affecting the runtime of our algorithm.
\begin{center}
\begin{tabular}{|c|c|}
\hline
$N$ & number of variables \\
\hline
$q$ & noise rate \\
\hline
$s$ & number of drift channels  \\
\hline
$J$ & drift strength \\
\hline
$k$ & regularization cutoff \\
\hline$\lambda_1$ & smallest dissipation rate\\
\hline
$\epsilon$ & approximation error \\
\hline
\end{tabular}
\end{center}

\subsection{Initialization}
\label{sec:init}

Below we assume that $\Pi_k|\psi(0)\ra\ne 0$ since otherwise $v_1(t,x)=0$ and there is nothing to do.
Recall that amplitudes of the initial state $\psi(0)$ in the Fock basis are determined by the expansion of the centered observable function $u_0(x)-\overline{u}_0$ 
in the Hermite basis: 
\be
\label{u0_expansion}
u_0(x) -\overline{u}_0= \sum_{\mm \in \calJ} \la \mm|\psi(0)\ra \HH_{\mm}(x),
\ee
where 
\[
\overline{u}_0=\int_{\RR^N} dx\, \mu(x) u_0(x)
\]
is the mean value of the observable.
It will be convenient to define parameters
\[
\sigma_i=\sqrt{\frac{q}{2\lambda_i}}.
\]
Then  the Hermite basis is 
\be
\label{hermit_poly_restated}
\mathbb{H}_{\multi{m}}(x) = \prod_{i=1}^N  \frac1{\sqrt{(m_i)!}} \herm_{m_i} \left(x_i/\sigma_i\right)
\ee
with $\mm \in \calJ=\ZZ_{\ge 0}^N \setminus {\{0^N\}}$.

First, let us explain how to compute the mean value $\overline{u}_0$.
By definition, the steady distribution $\mu(x)$ is a product of  $N$ normal distributions 
with the zero mean and standard deviations $\sigma_1,\ldots,\sigma_N$.
If $u_0(x)=\sum_{i=1}^N b_i x_i$ is a linear function then obviously  $\overline{u}_0=0$.
If $u_0(x)=\prod_{i=1}^N x_i^{d_i}$ is a monomial then
Wick's theorem gives
\be
\label{Wick}
\overline{u}_0=
\int_{\RR^N} dx\, \mu(x) \prod_{i=1}^N x_i^{d_i} =  \prod_{i=1}^N \frac{(d_i)! \sigma_i^{d_i}}{2^{d_i/2} (d_i/2) !} 
\ee
if $d_1,\ldots,d_N$ are even and $\overline{u}_0=0$  otherwise. 

Next let us explain how to prepare the initial state $\Pi_k|\psi(0)\ra$.
Consider first the linear observable.
From Eq.~(\ref{hermit_poly_restated})  one gets
$x_i=\sigma_i \HH_{e^i}(x)$ and thus
\[
u_0(x)=\sum_{i=1}^N b_i x_i =\sum_{i=1}^N b_i\sigma_i\HH_{e^i}(x).
\]
Substituting this into Eq.~(\ref{u0_expansion}) gives
\be
\label{init_linear_observable}
|\psi(0)\ra = \sum_{i=1}^N b_i\sigma_i |e^i\ra \quad \mbox{and} \quad
\Pi_k |\psi(0)\ra=\sum_{i\, : \, \lambda_i\le k} b_i\sigma_i |e^i\ra.
\ee
Below we write $|0_{anc}\ra$ for the all-zero state of the relevant ancilla register.
The size of the ancilla register may change after each step.
We write $|\mathsf{junk}\ra$ for any state in which the ancilla register is orthogonal to $|0_{anc}\ra$.
Such states will not contribute to the final outcome.
Let $\CC^N$ be the data register initialized in the state $|0\ra$.
Applying the oracle $\voracle{u}$ to the initial state $|0,0_{anc}\ra\in \CC^N \otimes \calH_{anc}$ gives a state 
\[
|\phi_1\ra=
\voracle{u}|0,0_{anc}\ra
=\frac1{J_u}\sum_{i=1}^N b_i |i,0_{anc}\ra+|\mathsf{junk}\ra,
\]
Tensor in ancillary qubits used by the oracle $\voracle{\lambda}$ to store
the binary representation of $\lambda_i$. 
Applying the oracle $\voracle{\lambda}$ gives a state 
\[
|\phi_2\ra =\frac1{J_u} \sum_{i=1}^N b_i |i,\lambda_i,0_{anc}\ra+|\mathsf{junk}\ra.
\]
Introduce a fresh ancilla qubit $|0\ra$ and 
perform the controlled rotation
\[
|\lambda_i,0\ra
\mapsto
\frac{\sigma_i}{\sigma_1}\,|\lambda_i,0\ra + \sqrt{1-\frac{\sigma_i^2}{\sigma_1^2}}\,|\lambda_i,1\ra.
\]
Note that  $\sigma_1=\max_i \sigma_i$ since $\lambda_1=\min_i \lambda_i$.
This gives a state
\[
|\phi_3\ra =\frac1{\sigma_{1} J_u}  \sum_{i=1}^N  b_i\sigma_i |i,\lambda_i,0_{anc}\ra+|\mathsf{junk}\ra.
\]
Introduce a fresh ancilla qubit $|0\ra$ and apply the controlled bit flip
\[
|\lambda_i,0\ra \mapsto \left\{
\ba{rcl}
|\lambda_i,0\ra &\mbox{if}& \lambda_i\le k,\\
|\lambda_i,1\ra &\mbox{if}& \lambda_i>k.\\
\ea
\right.
\]
This gives a state
\[
|\phi_4\ra =\frac1{\sigma_{1} J_u}  \sum_{i\, : \, \lambda_i\le k}  b_i\sigma_i |i,\lambda_i,0_{anc}\ra+|\mathsf{junk}\ra.
\]
Finally, use $\voracle{\lambda}^{-1}$ to uncompute $\lambda_i$ and apply a map $i\mapsto e^i$ that embeds
the data register $\CC^N$ into the Fock space $\calH_k$. This gives a state
\[
|\phi_5\ra = \frac1{\sigma_{1} J_u}  \Pi_k |\psi(0)\ra \otimes |0_{anc}\ra + |\mathsf{junk}\ra \equiv \frac1{\alpha_{init}} U_{init} |\bf{0}\ra,
\]
where
\[
\alpha_{init} = \sigma_1 J_u = J_u \sqrt{\frac{q}{2\lambda_1}}.
\]
%Postselecting all-zero state of the ancilla register prepares the normalized version of $\Pi_k|\psi(0)\ra$.
%The postselection success probability is 
%\[
%P_{init}
%=\frac1{\sigma_1^2 J_u^2} \|\Pi_k \psi(0)\|^2 =\Omega(\lambda_1 q^{-1} \|\Pi_k \psi(0)\|^2).
%\]
%Here we used the assumption $J_u=O(1)$.
%The value of $P_{init}$ can be estimated with an additive error $\delta$  by repeating the circuit $O(\delta^{-2})$ times and computing the fraction of shots
%that resulted in the successful postselection. Since   $\sigma_1=\sqrt{q/(2\lambda_1)}$ is known exactly,
%the norm  $\|\Pi_k \psi(0)\|^2$ can be estimated with an additive error $\delta$ by repeating the circuit $O(\sigma_1^{4} \delta^{-2})=O(q^2\lambda_1^{-2}\delta^{-2})$
%times.

To apply the regularization theorem we need upper bounds on
$\|\psi(0)\|^2$ and $\la \psi(0)|A|\psi(0)\ra$. From the above one easily gets
\[
\|\psi(0)\|^2
=
\sum_{i=1}^N b_i^2\frac{q}{2\lambda_i}
\le
\frac{q}{2\lambda_1}J_u^2  = O(q\lambda_1^{-1})
\]
and 
\[
\la\psi (0)|A|\psi(0)\ra
=
\sum_{i=1}^N b_i^2\frac{q}{2}
\le
\frac q2 J_u^2  = O(q).
\]

Next, consider a  monomial observable $u_0(x)=\prod_{i=1}^N x_i^{d_i}$. We shall use a well-known identity
\[
y^{n} = n! \sum_{\substack{0\le \ell \le n \\ \ell = n {\pmod 2}\\}}\;  \frac1{2^{(n-\ell)/2} ((n-\ell)/2)! \ell !} \herm_{\ell}(y).
\]
Substituting $n=d_i$ and  $y=x_i/\sigma_i$ in the above identity, taking the product over $i=1,\ldots,N$
and using Eq.~(\ref{hermit_poly_restated}) 
 one gets
\[
u_0(x)-\overline{u}_0 = \sum_{\mm \in \calS} \alpha(\mm) \HH_\mm(x),  \qquad \alpha(\mm)=\prod_{i=1}^N \frac{\sigma_i^{d_i}(d_i)!}{2^{(d_i-m_i)/2} ((d_i-m_i)/2)! \sqrt{m_i !}},
\]
where
\[
\calS=\{ \mm \in \calJ\, : \, m_i\le d_i \quad \mbox{and} \quad m_i=d_i {\pmod 2} \quad \mbox{for all $i=1,\ldots,N$}\}.
\]
As a consequence,
\be
\label{monomial_psi0}
|\psi(0)\ra = \sum_{\mm \in \calS} \alpha(\mm) |\mm\ra \quad \mbox{and} \quad \Pi_k|\psi(0)\ra =  \sum_{\mm \in \calS \cap \calJ_k} \alpha(\mm) |\mm\ra.
\ee
Note that $|\calS|=O(1)$ since we assumed $d=\sum_{i=1}^N d_i=O(1)$. Thus $\Pi_k|\psi(0)\ra$ is a superposition of $O(1)$ Fock basis vectors $|\mm\ra$
with easy-to-compute amplitudes $\alpha(\mm)$. 
The norm of $\Pi_k |\psi(0)\ra$ can be computed directly from Eq.~(\ref{monomial_psi0}) by summing $|\alpha(\mm)|^2$ over $\mm \in \calS\cap \calJ_k$.

It is known~\cite{mao2024toward} that any $Q$-qubit  superposition of $\chi$ basis states
can be prepared by a quantum circuit of size linear in $Q + Q \chi/\log{(Q)}$. Recall that the Hilbert space $\calH_k$ is expressed
using $Q=O(k\lambda_1^{-1}\log{N})$ qubits.
Since $\Pi_k|\psi(0)\ra \in \calH_k$ is a superposition of  $O(1)$ basis states,
the normalized version of $\Pi_k|\psi(0)\ra$ can be prepared by a quantum circuit $U_{init}$  of size $O(k\lambda_1^{-1}\log{N})$, such that 
\[
U_{init} |{\bf 0}\ra = \frac1{\alpha_{init}} \Pi_k |\psi(0)\ra \otimes |0_{anc}\ra + |\mathsf{junk}\ra
\]
for any given normalization  $\alpha_{init}\ge \|\Pi_k \psi(0)\|$.
The above shows that 
\[
\la \psi(0)|\psi(0)\ra \le |\calS|  \max_{\mm \in \calS} |\alpha(\mm)|^2 \quad \mbox{and} \quad \la \psi(0)|A|\psi(0)\ra \le |\calS| \max_{\mm \in \calS} \lambda_\mm |\alpha(\mm)|^2,
\]
with $|\calS|=O(1)$ and $\lambda_\mm = \sum_{i=1}^N \lambda_i m_i$. 
Since $\sigma_i\le \sigma_1=\sqrt{q/(2\lambda_1)}$, one gets
\[
|\alpha(\mm)|^2 \le O(\sigma_1^{2d})=O(q^d \lambda_1^{-d}).
\]
As a consequence, $\la \psi(0)|\psi(0)\ra \le O(q^d \lambda_1^{-d})$ and we can choose
\[
\alpha_{init} = O(q^d \lambda_1^{-d}).
\]
To upper bound $\lambda_\mm |\alpha(\mm)|^2$ note that $|\alpha(\mm)|^2$ contains a positive integer power of $\sigma_i^2 = q/(2\lambda_i)$ for each $m_i\ge 1$.
Hence $\lambda_i  m_i |\alpha(\mm)|^2 \le O(q \sigma_1^{2d-2}m_i)$ for all $i$ and we get
\[
\lambda_\mm  |\alpha(\mm)|^2 = \sum_{i=1}^N \lambda_i m_i |\alpha(\mm)|^2 \le O(q\sigma_1^{2d-2}) \sum_{i=1}^N m_i = |\mm| \cdot poly(q,\lambda_1^{-1}).
\]
Furthermore, $|\mm|=O(1)$ for all $\mm \in \calS$. Thus 
\[
\la \psi(0)|A|\psi(0)\ra \le O(1) \max_{\mm \in \calS} \lambda_\mm |\alpha(\mm)|^2= poly(q,\lambda_1^{-1}).
\]
%%%%%%%%%%%%%%%%%%%%%%%%%%%%%

Finally, let us prove that the chosen observables are robust to noise in the initial condition in the sense that their variance
vanishes at least linearly  with the noise rate $q$ and is independent of the dimension $N$. 
Let 
\be
\nu(x) = \mathrm{Var}_{z\sim \mu} u_0(x+z),
\ee
where $z\in \RR^N$ is a random vector sampled from the steady distribution $\mu$. 
For the linear observable one has $u_0(x+z)=u_0(x)+u_0(z)$ and 
\[
\nu(x)=  \mathrm{Var}_{z\sim \mu} u_0(z) = \EE_{z\sim \mu} u_0(z)^2 = \sum_{i=1}^N \sigma_i^2 b_i^2 \le \sigma_1^2 J_u^2 = O(q\lambda_1^{-1}).
\]
We claim that for the monomial degree-$d$ observable 
\[
\nu(x) 
\le
O\left( q \lambda_1^{-1} (\|x\|+\sqrt{q/\lambda_1})^{2d-2}\right),
\]
where the constant factor hidden in the big-$O$ notation may depend on $d$.
Indeed, Gaussian Poincare inequality asserts that for any smooth function $F\, : \, \RR^N\to \RR$ one has 
\[
\mathrm{Var}_{z\sim \mu}F(z)
\le
\sum_{i=1}^N \sigma_i^2\,
\EE_{z\sim\mu}
\left(\partial_i F(z)\right)^2.
\]
Choose $F(z)=u_0(x+z)$. Then
$\partial_i F(z)=0$ unless $d_i\ge 1$.
Below we only consider indices $i$ with $d_i\ge 1$. Then
\[
\partial_i F(z)
=
d_i (x_i+z_i)^{d_i-1}
\prod_{j\ne i}(x_j+z_j)^{d_j}.
\]
For any integer $r\in [0,d]$ Gaussian moments formula gives
\[
\EE_{z\sim N(0,\sigma_j^2)}  (x_j+z_j)^{2r}
\le
c_r\left(|x_j|+\sigma_j\right)^{2r}
\le
c_r\left(|x_j|+\sigma_1\right)^{2r} \le c_r \left(\|x\|+\sigma_1\right)^{2r} 
\]
where $c_r$ is a constant that depends only on $r$.
By definition, the function $F(z)$ depends only on $O(1)$ variables of $z$.
Thus the sum over $i$ in the Gaussian Poincare inequality contains $O(1)$ terms.
Each coefficient $\sigma_i^2$ is upper bounded by $\sigma_1^2=q/(2\lambda_1)$.
This gives the claimed upper bound. 

\subsection{Time evolution}
\label{sec:time_evolution}

Time evolution generated by the projected Kolmogorian  can be simulated using
the Linear Combination of Hamiltonian Simulations (LCHS) method of An, Liu, and Lin~\cite{an2023linear,an2026quantum}.
Here we use an improved version of LCHS due to Low and Somma~\cite{low2025optimal}.
Given Hermitian operators $L$ and $H$ such that $L$ is positive semi-definite, LCHS 
 approximates  a non-unitary time evolution operator $e^{-(L+iH)t}$  with $t \ge 0$ by a linear combination of unitary time evolutions, 
 \be
 \label{LCHS_eq1}
e^{-t(L+iH)}\approx \frac1{\sqrt{2\pi}}\int_{-\infty}^\infty dk\, \hat{f}(k) e^{-t(ikL + iH)} 
\ee
for a suitable kernel  function $\hat{f}\, : \, \RR \to \CC$, where the approximation error  is controlled by Theorem~1 of~\cite{low2025optimal}.
 The integral is then approximated by a finite sum 
 and block encoded by a quantum circuit
using the 
 linear combination of unitaries method~\cite{berry2015simulating,low2019hamiltonian}.
More precisely,  suppose we have access to oracles implementing
block encodings of $L$ and $H$ with some normalization constants $\alpha_L$ and $\alpha_H$.
Let us denote these oracles $\mathrm{BE}(L/\alpha_L)$ and $\mathrm{BE}(H/\alpha_H)$.
Suppose the oracles act on at most $m$ qubits. By combining Theorems~2,3,4 and Lemma~16 of~\cite{low2025optimal} one obtains
a quantum circuit $U_{sim}$ that block encodes an $\epsilon$-approximation of $e^{-t(L+iH)}$ with a normalization constant $\alpha_{sim}=O(1)$ such that 
\[
\| e^{-t(L+iH)} - \alpha_{sim}  \la 0_{anc}|U_{sim}|0_{anc}\ra \| \le \epsilon_{sim},
\]
where $|0_{anc}\ra$ is the all-zero state of all ancillary qubits.
The circuit $U_{sim}$ makes at most
\be
\label{LCHS_query_complexity}
n_{query} = O(\alpha_H t)  + O(\alpha_L t  \log{(1/\epsilon_{sim})}) +  O(\log{(1/\epsilon_{sim})})
\ee
queries to the oracles $\mathrm{BE}(L/\alpha_L)$ and $\mathrm{BE}(H/\alpha_H)$ and uses
\[
n_{gate}=O(m n_{query}) + \tilde{O}(n_{query})
\]
two-qubit gates. For simplicity, we shall ignore the second term in the gate count (it has logarithmic scaling with $t$, $\|L\|$, and doubly logarithmic scaling with $1/\epsilon_{sim}$).
Then
\be
\label{LCHS_gate_complexity}
n_{gate} = O(m n_{query}).
\ee

Let us specialize this general result  to simulate time evolution $e^{t(-A_k+C_k)}$ generated by the projected Kolmogorian acting on the low-dissipation subspace $\calH_k$.
Choose  $L=A_k$ and $H=iC_k$. Then
$L$ and $H$ are Hermitian, $L$ is positive definite. LCHS gives a quantum circuit $U_{sim}$ such that 
\be
\label{LCHS_eq2}
\| e^{t(-A_k+C_k)} - \alpha_{sim}  \la 0_{anc}|U_{sim}|0_{anc}\ra \| \le \epsilon_{sim}, \qquad \alpha_{sim}=O(1).
\ee
Our Theorem~\ref{thm:block_encoding} gives an implementation of the oracle $\mathrm{BE}(H/\alpha_H)$ with
$\alpha_H=4(qs)^{1/2} J k^3  \lambda_1^{-7/2}$. This requires  $m=poly(k\lambda_1^{-1},\log{N})$ qubits. 
By definition, $A_k$ is a diagonal operator with diagonal entries $\lambda_\mm = \sum_{i=1}^N \lambda_i m_i$ for $\mm \in \calJ_k$.
In particular, $\lambda_\mm \in [\lambda_1,k]$.
Thus the oracle $\mathrm{BE}(L/\alpha_L)$ with a normalization constant  $\alpha_L=k$ can be implemented by querying the value oracle
 $\voracle{\lambda}$ at most $O(k\lambda_1^{-1})$ times (since any multi-index $\mm \in \calJ_k$ has at most $k\lambda_1^{-1}$ nonzero indices $m_i\ne 0$).
  From Eq.~(\ref{LCHS_query_complexity}), one infers
 that the  query complexity of the circuit $U_{sim}$ 
in terms of the oracles $\voracle{c}$, $\voracle{\lambda}$, and $\poracle{c}$ is at most
\be
\label{LCHS_query_complexity_KE}
n_{query} = O(t(qs)^{1/2} J k^3  \lambda_1^{-7/2}) + O(kt  \log{(1/\epsilon_{sim})})  + O(\log{(1/\epsilon_{sim})}).
\ee
From Eq.~(\ref{LCHS_gate_complexity}) one infers that 
the gate count of $U_{sim}$ scales as
\be
\label{LCHS_gate_complexity_KE}
n_{gate} = poly(k\lambda_1^{-1}, \log{N}) n_{query}.
\ee

\subsection{Readout state preparation}
\label{sec:readout_state_prep}

Let $x\in \RR^N$ be the target initial condition for the SDE. Let 
\[
\mathrm{supp}(x)=\{i\in [N] \,:\, x_i\ne 0\}
\]
be the support of $x$. Recall that we assumed 
\[
r_x\equiv |\mathrm{supp}(x)|=O(1).
\]
The vector $x$ is specified by a list of pairs $(i,x_i)$ with $i \in \mathrm{supp}(x)$.
By definition of the readout state, 
\[
\Pi_k |\psi_{out}(x)\ra = \sum_{\mm \in \calS}\;  
|\mm\ra \cdot  \prod_{i=1}^N \frac{a_i^{m_i}}{\sqrt{(m_i)!}}, \qquad a_i =x_i \sqrt{\frac{2 \lambda_i}{q}},
\]
\[
\calS=
\left\{
\mm\in \calJ_k\, :\, 
\mathrm{supp}(\mm)\subseteq \mathrm{supp}(x)
\right\}.
\]
Let $p$ be the smallest integer such that $p\ge k\lambda_1^{-1}$.
Then  $\max_i m_i \le p$ for any $\mm \in \calJ_k$ and 
\[
|\calS|\le (p+1)^{r_x} = poly(k,\lambda_1^{-1}).
\]
Indeed, $\mm \in \calS$ implies $m_i=0$ for all $i\notin \mathrm{supp}(x)$
and $m_i \in \{0,1,\ldots,p\}$ for $i \in \mathrm{supp}(x)$.
Then the norm $\|\Pi_k \psi_{out}(x)\|$ can be computed classically by the brute force method
in time $poly(k,\lambda_1^{-1},\log{N})$ with  at most $r_x=O(1)$ queries to   the oracle $\voracle{\lambda}$.
The latter is needed to learn dissipation rates $\lambda_i$ for all $i \in \mathrm{supp}(x)$.

Recall that any $Q$-qubit  superposition of $\chi$ basis states
can be prepared by a quantum circuit of size linear in $Q + Q \chi/\log{(Q)}$, see~\cite{mao2024toward}.
Since $\calH_k$ is expressed using $Q=O(k\lambda_1^{-1}\log{N})$ qubits
and $\chi=|\calS|=poly(k,\lambda_1^{-1})$, this gives a circuit $U_{out}$ such that 
\[
U_{out}|{\bf 0}\ra
=
\frac1{\alpha_{out}}\Pi_k|\psi_{out}(x)\ra\otimes |0_{anc}\ra
\]
with 
\[
\alpha_{out}=\|\Pi_k \psi_{out}(x)\| \le \| \psi_{out}(x)\| =  \exp{\left[ q^{-1}\sum_{i=1}^N \lambda_i x_i^2\right]}.
\]
The circuit $U_{out}$ uses  $poly(k,\lambda_1^{-1},\log{N})$ gates and makes $O(1)$ queries to $\voracle{\lambda}$
at the classical preprocessing step.

\subsection{Putting everything together}
\label{sec:runtime}

It remains to choose the regularization cutoff $k$, simulation error tolerance $\epsilon_{sim}$,
and the error tolerance of the Hadamard test such that the resulting approximation of the desired 
expected value $v(t,x)$ has error at most $\epsilon$.

First, let us choose $k$ large enough that the last two terms in Eq.~(\ref{v_approximation_eq1}) have magnitude at most $\epsilon/2$, that is,
\[
|\la \psi_{out}(x)|\Pi_k|\psi(t)-\psi_k(t)\ra|\le \frac{\epsilon}4 \quad \mbox{and}  \quad  | \la \psi_{out}(x)|\Pi_k^\perp |\psi(t)\ra|\le \frac{\epsilon}4.
\]
This is the case whenever
\[
\|\psi(t)-\psi_k(t)\|^2 \le \frac{\epsilon^2}{16\|\psi_{out}(x)\|^2}
\]
and
\[
\| \Pi_k^\perp \psi(t)\|^2 \le \frac{\epsilon^2}{16 \| \psi_{out}\|^2}.
\]
By Theorem~\ref{thm:regul}, this is the case whenever 
\[
\frac{12 \gamma}{ k^{1/2}} \left(\la \psi(0)|A|\psi(0)\ra  + \kappa_1 \la\psi(0)|\psi(0)\ra\right)\le  \frac{\epsilon^2}{16\| \psi_{out}(x)\|^2}
\]
and  
\[
\frac1k  \left(\la \psi(0)|A|\psi(0)\ra  + \kappa_1 \la\psi(0)|\psi(0)\ra\right)\le  \frac{\epsilon^2}{16 \| \psi_{out}\|^2}.
\]
We have already established that the quantities $\la \psi(0)|\psi(0)\ra$ and $\la \psi(0)|A|\psi(0)\ra$ are upper bounded by
$poly(q,\lambda_1^{-1})$, see Section~\ref{sec:init}. 
Recalling that $\gamma$ and $\kappa_1$ are upper bounded by $poly(J,q,\lambda_1^{-1})$,
it suffices to choose
\be
\label{our_choice_k}
k = \| \psi_{out}(x)\|^4 \cdot poly(q,J,\lambda_1^{-1},\epsilon^{-1}).
\ee
This ensures that 
\be
|v(t,x)-v_1(t,x)|\le \frac{\epsilon}2.
\ee
Thus it suffices to approximate $v_1(t,x)$ within error $\epsilon/2$.
Define our approximation of $v_1(t,x)$ as 
\[
v_2(t,x) =\alpha \la {\bf 0}| U_{out}^{-1} U_{sim} U_{init}|{\bf 0}\ra, \qquad \alpha= \alpha_{init} \alpha_{sim} \alpha_{out}.
\]
Choose the error tolerance $\epsilon_{sim}$ of Hamiltonian simulation such that 
\be
\label{error_v1}
|v_1(t,x)-v_2(t,x)|\le \frac{\epsilon}{4}.
\ee
This is the case whenever 
$\epsilon_{sim} \le \frac{\epsilon}{\alpha}$.
Using Eqs.~(\ref{LCHS_query_complexity_KE},\ref{LCHS_gate_complexity_KE},\ref{our_choice_k}) one concludes that 
the gate and query complexity of $U_{sim}$ is upper bounded as
\[
n_{gate}+n_{query} \le t \cdot poly(\log{N},q,J,\lambda_1^{-1}, \epsilon^{-1}, \|\psi_{out}\|) \cdot \log{(\alpha/\epsilon)}.
\]
Collecting the values of $\alpha_{init},\alpha_{sim},\alpha_{out}$ defined above one gets
\[
\alpha \le poly(q,\lambda_1^{-1}) \cdot \|\psi_{out}(x)\|.
\]
Hence
\[
n_{gate}+n_{query} \le t \cdot poly(\log{N},q,J,\lambda_1^{-1}, \epsilon^{-1}, \|\psi_{out}\|) .
\]
Finally, approximating $v_2(t,x)$ within error $\epsilon/4$ using the Hadamard test requires $O(\alpha^2 \epsilon^{-2})$ circuit repetitions.
This introduces another factor $poly(q,\lambda_1^{-1},\epsilon^{-1},\|\psi_{out}(x)\|)$.
Using the exponential scaling of $\|\psi_{out}(x)\|$ established in Eq.~(\ref{readout_state_norm}) gives the simulation cost claimed in Theorem~1.

%\begin{rem}
%  Suppose  $\lambda_1=\dots=\lambda_N$ and $\la \psi(0)|A^{p'}|\psi(0)\ra<\infty$ for all $p'\in[0,2p+6]$, $p\ge3$. Then the scaling Eq.~(\ref{our_choice_k}) can be improved using Eq.~(\ref{eq:kp-reg-error}) of Theorem~\ref{thm:regul-kp}. Indeed, for any $1\le p'\le p$ the inequality $\epsilon_{reg}\le \epsilon/6$ is satisfied for  
%\begin{align}
%k^{p'/2}\ge  \frac{6\gamma}{\epsilon}\left(\sum_{m=0}^{p'+1}  \kappa_1^{p'+1-m}\frac{(p'+1)!^2}{m!^2} \la \psi_0|A^{m}|\psi_0\ra 3^{\frac{p'(p'+1)-m(m-1)}{2}}\right)^{\frac12}\|\psi_{out}(x)\| 
%\end{align}
%with $\gamma = \frac{3J \sqrt{qs}}{\lambda_1^2}$ and $ \kappa_1=\frac{6J^2 qs}{\lambda_1^2}$. Above we showed~\eqref{Apsi0_bound}: in a similar way, employing~\eqref{psi0_norm_bound} and~\eqref{eq:lambdam-bound} we show that $\la \psi_0|A^{m}|\psi_0\ra\le\operatorname{poly}(\lambda_1^{-m},q^m)$. Hence 
%\[
%(6\gamma)^{\frac2{p'}}\left(\kappa_1^{p'+1-m}\frac{(p'+1)!^2}{m!^2} \la \psi_0|A^{m}|\psi_0\ra 3^{\frac{p'(p'+1)-m(m-1)}{2}}\right) \le \operatorname{poly}(\lambda_1^{-m},q^m,J,q,s,\lambda_1)
%\] so that 
%\begin{align}
%    k&\ge \frac{1}{ \epsilon^{\frac2{p'}}}
%    \left(\sum_{m=0}^{p'+1} \operatorname{poly}(\lambda_1^{-m},q^m,J,q,s,\lambda_1) \right)^{\frac1{p'}}
%  \exp\left[\frac{2\lambda_1\| x\|^2}{q p'}\right],\quad 1\le p'\le p, \quad p\ge 3
%\end{align}
%\end{rem}

\begin{rem}
\label{remark:mild_exponential}
If we additionally assume that the distribution of dissipation rates is nearly flat such that $\lambda_N=O(\lambda_1)$ then
the cost of our quantum algorithm can be reduced.
In particular,
for any integer $\ell \ge 1$ the required qubit count $Q$ scales as
 \[
 Q\le  \log{(N) }
    \operatorname{poly}_\ell(q,J,\lambda_1^{-1},\epsilon^{-1})
    \|\psi_{out}\|^{1/\ell}
\]
The gate and query count per quantum circuit in a single Hadamard test scale as
\[
n_{gate}+n_{query}
  \le
    t\,\operatorname{poly}_\ell(
        q,J,\lambda_1^{-1},\epsilon^{-1},\log N) \|\psi_{out}\|^{1/\ell}.
\]
Here 
$\operatorname{poly}_\ell$ are finite-degree polynomials that depend on $\ell$.
However, the number of samples (circuit repetitions) in the Hadamard test scales as 
$\epsilon^{-2} \operatorname{poly}(q,\lambda_1^{-1})\|\psi_{out}\|^2$ regardless of $\ell$.
\end{rem}
Below we sketch the proof for the special case when the dissipation rates are perfectly flat, 
$0<\lambda_1=\ldots=\lambda_N\le O(1)$.
Then $A=\lambda_1\sum_{i=1}^N a_i^\dag a_i$ is proportional to the particle number operator.
One can easily check that the initial state $\psi(0)$ in the Kolmogorov equation for the chosen linear and monomial observables $u_0$ obeys
\[
\la \psi(0)|A^m|\psi(0)\ra \le \operatorname{poly}_m(q,\lambda_1^{-1})
\]
for any integer $m\ge 0$. Here $\operatorname{poly}_m$ is some finite-degree polynomial that depends on $m$.
For example, if $u_0$ is a linear observable, Eq.~(\ref{init_linear_observable})
gives 
\[
\la \psi(0)|A^m|\psi(0)\ra=q\lambda_1^{m-1} J_v^2/2= O(q\lambda_1^{m-1})
\]
for all $m\ge 0$.
Fix an integer $r$ with $1\le r\le p$ and apply
Theorem~\ref{thm:regul-kp} with $p'=r$. Since $\lambda_N=\lambda_1$, the quantity
$\mathcal{W}$ in that theorem becomes
\[
\mathcal{W}_r =
\sum_{m=0}^{r+1}
\kappa_1^{r+1-m}
\frac{(r+1)!^2}{m!^2}
\langle \psi(0)|A^m|\psi(0)\rangle
3^{\frac{r(r+1)-m(m-1)}{2}} \le \operatorname{poly}_r(q,J,\lambda_1^{-1}).
\]
The projected regularization error bound gives
\[
    \|\Pi_k(\psi_k(t)-\psi(t))\|
    \le
    \frac{\gamma \mathcal{W}_r^{1/2}}{k^{r/2}} .
\]
Consequently,
\be
\label{improved_term1}
\left|
\langle \psi_{out}(x)|\Pi_k|\psi_k(t)-\psi(t)\rangle
\right|
\le
\|\psi_{out}(x)\|
\frac{\gamma \mathcal{W}_r^{1/2}}{k^{r/2}} .
\ee
The moment
bounds of Lemma~\ref{lemma:moment_bounds} give
\be
\label{improved_term2}
    \|\Pi_k^\perp\psi(t)\|^2
    \le
    \frac1{k^r} \la \psi(t) |A^r|\psi(t)\ra \le 
    \frac{\mathcal{B}_r}{k^r},
\ee
where 
$\mathcal{B}_r$ is a bound scaling linearly with the initial moments 
$\la \psi(0)|A^\ell|\psi(0)\ra$ with $\ell \le r$ 
and polynomially with the 
parameters $q,J,\lambda_1^{-1}$.
In particular,
\[
\calB_r \le  \operatorname{poly}_r(q,J,\lambda_1^{-1}).
\]

%%%%%%
Using Eqs.~\eqref{improved_term1} and~\eqref{improved_term2} in place of the
bounds based on the weak regularization theorem, it suffices to choose
\[
k
=
\operatorname{poly}_{r}(q,J,\lambda_1^{-1},\epsilon^{-1})
\|\psi_{out}(x)\|^{2/r}.
\]
The number of qubits then scales as
\[
Q=O(k\lambda_1^{-1}\log N).
\]
As shown above,
the gate/query complexity of a single Hadamard test circuit is linear in $t$
and 
 polynomial in
 $k,\log{(N)}, q, J,\lambda_1^{-1},\epsilon^{-1}$
  for some fixed finite-degree polynomial.
Therefore,
for any fixed $\ell$, we may choose $r=O(\ell)$ such that 
\[
n_{gate}+n_{query}\le  t\cdot
\operatorname{poly}_{\ell}
(q,J,\lambda_1^{-1},\epsilon^{-1},\log N)
\|\psi_{out}(x)\|^{1/\ell},
\]
as claimed. The number of samples scales as
$\epsilon^{-2}\alpha^2=\epsilon^{-2}\operatorname{poly}(q,\lambda_1^{-1})\|\psi_{out}\|^2$.

\section{Damped Euler-Bardina model} 
\label{sec:examples}
In this section we describe a symmetrization of the damped Euler-Bardina (dEB) model~\cite{BardinaFerzigerReynolds1980SGS,LaytonLewandowski2006WellPosedTurbulenceModel} for an incompressible fluid confined to a two-dimensional box $\TT^2=[0,1]^2$ with periodic boundary conditions. In section~\ref{sec:contDEB} we begin with the continuous model and prove that the symmetrization error is controlled by the choice of regularization parameters. In section~\ref{subsec:discrete-deb} we discretize the dEB model in space using divergence-free sine basis functions and demonstrate that the resulting quadratic ODE satisfies all the conditions of our quantum algorithm: divergence-free condition, and a small number of drift channels. 

\subsection{Continuous dEB model}\label{sec:contDEB}

\emph{Preliminaries.} We introduce the following functional spaces following the standard theory of fluid dynamics (see, e.g., \cite[p.~45-48]{Foias_Manley_Rosa_Temam_2001}). Unless otherwise stated, we let lowercase variables denote scalar quantities and uppercase variables denote vector quantities. We denote the components of an $N$-dimensional $V$ by $V_1,...,V_N$ and demote the spatial domain $\Omega = \mathbb{T}^2$.

\begin{itemize}
    \item $\mathbb{R}^N$ -- Euclidean space of $N$-dimensional vectors $U,V$ with inner product $\ip{U}{V}_{\mathbb{R}^N} = \sum_{i=1}^N U_i V_i$ and norm $\norm{U}_{\mathbb{R}^N}^2 = \ip{U}{U}_{\mathbb{R}^N}$.
    \item $L^2(\Omega)$ -- space of $\Omega$-periodic functions $u,v:\Omega\subset \mathbb{R}^2\to\mathbb{R}$ with inner product $\ip{u}{v}_{L^2}=\int_\Omega u (\xi_1,\xi_2) v (\xi_1,\xi_2) \, d\xi_1 d\xi_2$ and norm $\norm{u}_{L^2}^2 = \ip{u}{u}_{L^2}$.
    \item $L^2(\Omega)^2$ -- space of $\Omega$-periodic vector functions $U,V:\Omega\subset \mathbb{R}^2\to\mathbb{R}^2$ with inner product $\ip{U}{V}_{L^2}=\ip{U_1}{V_1}_{L^2}+\ip{U_2}{V_2}_{L^2}$ and norm $\norm{U}_{L^2}^2 =\ip{U}{U}_{L^2}$.
    \item $H^1(\Omega) = \{ u \in L^2(\Omega) \, : \, \norm{\D u}_{\mathbb{R}^2} \in L^2(\Omega)\}$.
    \item $H^1(\Omega)^2 = \{ U \in L^2(\Omega)^2 \, : \, U_{1,2}\in H^1(\Omega)\}$ with norm $\norm{U}_{H^1}^2 = \norm{U}_{L^2}^2 + \norm{\D U}_{L^2}^2$.
    \item $H^2(\Omega)^2 = \{ U \in L^2(\Omega)^2 \, : \, U_{1,2}\in H^2(\Omega)\}$ with norm $\norm{U}_{H^2}^2 = \norm{U}_{L^2}^2 + \norm{\D U}_{L^2}^2 + \norm{\D^2 U}_{L^2}^2$.
    \item $\mathring{\mathcal{V}}$ -- space of divergence-free functions $U\in H^1 (\Omega)^2$ with zero-mean components; $\int_\Omega U_i  \, d\xi_1 d\xi_2=0$ for $i=1,2$.
    \item $\mathcal{H}$ -- closure of $\mathring{\mathcal{V}}$ in $L^2(\Omega)^2$ with inner product and norm denoted by $\ip{\cdot}{\cdot}$ and $\abs{\cdot}$, respectively.
    \item $\mathcal{V}$ -- closure of $\mathring{\mathcal{V}}$ in $H^1(\Omega)^2$ with inner product and norm denoted by $\ipp{\cdot}{\cdot}$ and $\norm{\cdot}$, respectively.
\end{itemize} 

We consider the following symmetrization of the damped Euler-Bardina (dEB) model for an incompressible fluid
\be
\label{eq:dEB}
\frac{\partial}{\partial t} V +\overline{(\overline V \cdot \nabla)\overline V} +\lambda V  = -\nabla P, \qquad\quad \D \cdot V = 0
\ee
on the 2D domain $\Omega=\TT^2=[0,1]^2$ with periodic boundary conditions (BC) and initial condition $V(0)\in H^2(\Omega)^2$. Here $t\ge 0$ is time, $V=V(t,\xi)\in \RR^2$ is the fluid velocity at $\xi\in \Omega$, $P=P(t,\xi)\in\mathbb{R}$ is the scalar pressure field enforcing incompressibility, and $\lambda>0$ is a damping strength. The quantity $\overline V = \calF_{\alpha\beta}V$ is a filtered version of $V$ with filter $\calF_{\alpha\beta}=(I+\alpha S^\beta)^{-1}$ given by the Green's function of the operator $I+\alpha S^\beta$ for $S = \frac{1}{2\pi}(-\Delta)^{\frac{1}{2}}$. Here $-\Delta$ is the Stokes operator which coincides with the Laplacian in the case of periodic BC~\cite[p.~52]{Foias_Manley_Rosa_Temam_2001}. For $\alpha=\beta=\lambda=0$, the dEB model Eq.~\eqref{eq:dEB} reduces to the classical Euler equation. For $\lambda=0$ and $\alpha,\beta>0$, the dEB model becomes the classical inviscid Euler-Bardina model, provided one bar is dropped in the nonlinear advection term, ${(\overline V \cdot \nabla)\overline V}$. We note that dEB is related to the $\alpha$-Navier-Stokes regularization models~\cite{CaoLunasinTiti2006Bardina}. 

Before establishing the convergence properties of the dEB model, we prove the following useful Lemma (adapted from \cite[Lemma 10]{Cao23122009}).

\begin{lemma}\label{lemma:alpha-norm-bound}
    Let $\varphi \in \calH$ and $\delta \in \calV$. Then, for $\beta\geq1$,
    \begin{equation}\label{eq:alpha-norm-bound}
        \left| \langle (I - \calF_{\alpha\beta}) \varphi , \delta \rangle \right| \leq \alpha^{\frac{1}{\beta}} \abs{\varphi} \, \norm{\delta}.
    \end{equation}
\end{lemma}
\begin{proof}
Since $\calF_{\alpha \beta}$ is invertible we can write
\begin{equation}\label{eq:lemma10-1}
\begin{aligned}
    I - \calF_{\alpha\beta}
    &= ( \calF_{\alpha\beta}^{-1} - I) \calF_{\alpha\beta} \\
    &= \alpha S^\beta ( I + \alpha S^\beta)^{-1}. 
\end{aligned}
\end{equation}
Since $-\Delta$ is self-adjoint and non-negative \cite[p.52]{Foias_Manley_Rosa_Temam_2001}, so too is the operator $S$. Denote the eigenpairs of $S$ by $S e_j = \sigma_j e_j$ for $e_j\in \mathcal{V}$ and $0<\sigma_1\leq\sigma_2\leq...\,$. Using Eq.~(\ref{eq:lemma10-1}), self-adjointess of $S$ and the Cauchy-Schwarz inequality, we have
\begin{equation*}
\begin{aligned}
    \left| \langle (I - \calF_{\alpha\beta}) \varphi , \delta \rangle \right| &= \left| \ip{ \alpha S^\beta ( I + \alpha S^\beta)^{-1} \varphi} {\delta} \right| \\
    &= \left| \ip{ \alpha S^{\beta-1} ( I + \alpha S^\beta)^{-1} \varphi}{S \delta} \right| \\
    &\leq \abs{\alpha S^{\beta-1} ( I + \alpha S^\beta)^{-1} \varphi} \, \norm{\delta}
\end{aligned}
\end{equation*}
where the last line follows from $\abs{(-\Delta)^{\frac{1}{2}}V} = \norm{V}$ for all $V\in\calV$ \cite[p.51]{Foias_Manley_Rosa_Temam_2001}. Denote the operator norm of an operator $O:\calH\to \calH$ by $\norm{O}_{\text{op}}$. Then
\begin{equation}\label{eq:discrete-sup-bound}
\begin{aligned}
    \left| \langle (I - \calF_{\alpha\beta}) \varphi , \delta \rangle \right| &\leq \norm{\alpha S^{\beta-1} ( I + \alpha S^\beta)^{-1}}_{\text{op}} \, \abs{\varphi} \, \norm{\delta} \\
    &= \abs{\varphi} \, \norm{\delta}  \sup_{k=1,2,...} \frac{\alpha \sigma_k^{\beta- 1}}{1 + \alpha \sigma_k^\beta}. 
\end{aligned}
\end{equation}

The change of variables $z_k = \alpha \sigma_k^\beta$ yields
\begin{equation*}
\begin{aligned}
    \sup_{k=1,2,...} \frac{\alpha \sigma_k^{\beta- 1}}{1 + \alpha \sigma_k^\beta} &= \sup_{k=1,2,...} \frac{\alpha^{\frac{1}{\beta}} z_k^{1 - \frac{1}{\beta}}}{1 + z_k} \\
    &= \alpha^{\frac{1}{\beta}} \sup_{0\leq y < \infty} \frac{y^p}{1+y}
\end{aligned}
\end{equation*}
where $p=1 - 1/\beta$. We treat the cases $\beta>1$ and $\beta=1$ separately. For $\beta>1$, it follows from direct differentiation of $h(y) = y^p / (1+y)$ that
\begin{equation}\label{eq:h-y-star}
    y^* \coloneq \argsup_{0\leq y < \infty} h(y) = \frac{p}{1-p} \qquad \Rightarrow \qquad h(y^*) = \frac{(\beta - 1)^{1 - \frac{1}{\beta}}}{\beta} < 1.
\end{equation}
Using Eq.~\eqref{eq:h-y-star} to upper bound the expression in Eq.~\eqref{eq:discrete-sup-bound}, we obtain the result in Eq.~\eqref{eq:alpha-norm-bound}. Finally, for $\beta=1$ we have $h(y) = 1/ (1+y)$ so
\begin{equation*}
    \sup_{0\leq y < \infty} h(y) = \sup_{0\leq y < \infty} \frac{1}{1+y}=1
\end{equation*}
and we obtain Eq.~\eqref{eq:alpha-norm-bound} similarly to the case $\beta>1$.
\end{proof}

The following proposition collects the key properties of the continuous dEB model. We make use of several standard inequalities from the analysis of PDEs in fluid dynamics, e.g., Ladyzhenskaya's inequality and the Brezis-Gallouet inequality. For more details, we refer the reader to \cite[p.99-101]{Foias_Manley_Rosa_Temam_2001}. 

\begin{prop}\label{prop:deb-convergence}
 Let $\beta\geq1$. The unique solution $V_\alpha$ of Eq.~\eqref{eq:dEB} converges to the unique solution $V$ of the classical damped-Euler model, 
 \be
\label{eq:dEBc}
\frac{\partial}{\partial t} V + ( V \cdot \nabla) V +\lambda V  = -\nabla P, \qquad\quad V(0) \in H^2(\Omega)^2,
\ee
with convergence rate $\|V_\alpha- V\|_{L^2}=O(\alpha^{\frac{1}{\beta}})$. Furthermore, for fixed $\alpha>0$, the operator $V\mapsto \overline{(\overline V \cdot \nabla)\overline V}$ is a bounded nonlinear operator inducing the skew-symmetric bilinear form
\begin{equation}\label{eq:reg-nonlinearity-skewsym}
    \int_{\Omega} \overline{(\overline V \cdot \nabla)\overline V} \cdot U d\xi = - \int_{\Omega} \overline{(\overline V \cdot \nabla)\overline U} \cdot V d\xi
\end{equation}
defined for divergence-free vector functions $U,V$. 
\end{prop}
\begin{rem}
    The operator $V\mapsto \overline{(\overline V \cdot \nabla)\overline V}$ is a bounded modification of the unbounded nonlinear advection operator $V\mapsto ( V \cdot \nabla) V$. 
\end{rem}

\begin{proof}
Here we prove Proposition \ref{prop:deb-convergence}. For divergence-free $U,V,W$, define the trilinear form
\begin{align}\label{eq:def-trilinear-form}
    b(U,V,W) = \int_{\Omega} {\big(( U \cdot \nabla) V}\big) \cdot W d\xi.
\end{align}
Integrating by parts and applying the divergence-free condition gives
\begin{equation}\label{eq:b-skew-symmetric}
    b(U,V,W) = - b(U,W,V).
\end{equation}
Using self-adjointness of $\calF_{\alpha\beta}$, the right-hand side of Eq.~\eqref{eq:reg-nonlinearity-skewsym} can be written
\begin{equation*}
    \int_{\Omega} \overline{(\overline V \cdot \nabla)\overline V}  \cdot U dx =\int_{\Omega} {(\overline V \cdot \nabla)\overline V}  \cdot \overline{U} dx
    = b(\overline V,\overline V,\overline U)
\end{equation*}
and so the skew-symmetry claim follows from Eq.~\eqref{eq:b-skew-symmetric}. 

We proceed to bound the nonlinear operator $\Phi (V) \coloneq \overline{(\overline V \cdot \nabla)\overline V}$. Using the dual characterization of the $\calH$-norm \cite[p.90-92]{Foias_Manley_Rosa_Temam_2001}, self-adjointness of $\calF_{\alpha\beta}$, and Eq.~(A.46e) in \cite[p.100]{Foias_Manley_Rosa_Temam_2001}, we get
\begin{equation}\label{eq:phi_bound_1}
\begin{aligned}
    \norm{\Phi (V)}_{L^2} &= \sup_{\substack{{W}\in \calH \\ \abs{W}=1}} \ip{\Phi (V)}{{W}} \\
    &= \sup_{\substack{\abs{W}\in \calH \\ {W}=1}} \ip{{(\overline V \cdot \nabla)\overline V}}{\overline{W}}\\
    &\leq c_1  \sup_{\substack{W\in \calH \\ \abs{W}=1}} \abs{\overline{V}}^{\frac{1}{2}} \, \norm{\overline{V}}^{\frac{1}{2}} \, \abs{\overline{V}}^{\frac{1}{2}} \, \norm{\overline{V}}^{\frac{1}{2}} \norm{\overline{W}}.
\end{aligned}
\end{equation}
It follows from the definition of the operator norm $\norm{\cdot}_{\text{op}}$ that, for $\beta\geq 1$, $\alpha>0$,
\begin{equation}\label{eq:D-S-op-norm}
    \norm{\D \calF_{\alpha\beta}}_{\text{op}} = \sup_{k=1,2,...} \abs{k} \left( 1 + \alpha \abs{k}^{\beta} \right)^{-1} < \infty. 
\end{equation}
Thus, if $W\in \calH$ then $\overline{W}\in \calV$ since 
\begin{equation}\label{eq:K-delta}
    \norm{\overline{W}} = \abs{\nabla \overline{W}} \leq \norm{\D \calF_{\alpha\beta}}_{\text{op}} \abs{W} < \infty.
\end{equation}
Take finite $K>0$ satisfying $\norm{\D \calF_{\alpha\beta}}_{\text{op}} \leq K$. We further note that $\norm{\calF_{\alpha\beta}}_{\text{op}}\leq 1$ so $\abs{\overline{V}}\leq \norm{\calF_{\alpha\beta}}_{\text{op}}\abs{{V}}$ for all $V\in\calH$. Putting everything together with Eq.~\eqref{eq:phi_bound_1} yields
\begin{equation*}
    \norm{\Phi (V)}_{L^2} \leq c_1 K^2 \abs{{V}}^2.
\end{equation*}

We move on to the convergence result. Let $V$ solve Eq.~\eqref{eq:dEBc} and let $V_\alpha$ solve Eq.~\eqref{eq:dEB}, with the same initial condition,
\begin{equation*}
    V(t=0) = V_\alpha(t=0) = V_0 \in H^2(\Omega)^2.
\end{equation*}
Define the error $e:= V - V_\alpha$. From Eqs.~(\ref{eq:dEBc},~\ref{eq:dEB}), $e$ satisfies
\begin{align}\label{eq:error-eq-deb}
\partial_t e + \lambda e + (V\cdot\D)V - \overline{(\overline{V}_\alpha \cdot \D)\overline{V}_\alpha}
= - \D p + \D p_\alpha.
\end{align}
Define the bilinear form $B(U,V) := (U\cdot\D)V$ which, by bilinearity, satisfies
\begin{align}\label{eq:bilinearity}
B(V-V_\alpha,V) + B(V,V-V_\alpha) - B(V-V_\alpha,V-V_\alpha)
&= B(V,V) - B(V_\alpha,V_\alpha).
\end{align}
Adding and subtracting $B(V_\alpha,V_\alpha)$ to the left-hand side of Eq.~\eqref{eq:error-eq-deb} and using Eq.~\eqref{eq:bilinearity}, we get
\begin{equation}\label{eq:error-eq-0}
    \partial_t e + \lambda e + B(e,V) + B(V,e) - B(e,e) + B(V_\alpha, V_\alpha) - \overline{(\overline{V}_\alpha \cdot \D)\overline{V}_\alpha}
= - \D p + \D p_\alpha .
\end{equation}

Denote the Stokes operator~\cite[p.~52]{Foias_Manley_Rosa_Temam_2001} by $A=-\Delta$ with eigenpairs
\begin{equation}\label{eq:stokes}
A e_j = \eta_j e_j \qquad\text{s.t.}\qquad 0<\eta_1\leq \eta_2\leq... \,\,\, .
\end{equation}
For periodic BC on $\Omega=\mathbb{T}^2$, it is well-known that $\eta_1 = 4\pi^2$ and the trilinear form in Eq.~\eqref{eq:def-trilinear-form} satisfies
\begin{equation}
    b(U,U, A U) = 0
\end{equation}
for all $U\in D(A)$ \cite[p.52,101]{Foias_Manley_Rosa_Temam_2001}. We note that it also holds
\begin{equation}\label{eq:b-filter-skew-sym}
    b(\overline{U},\overline{U}, A \overline{U}) = 0.
\end{equation}
Taking the inner product of Eq.~\eqref{eq:error-eq-0} with $e$ and applying Eq.~(\ref{eq:b-filter-skew-sym}) and Eq.~\eqref{eq:b-skew-symmetric}, we get
\begin{equation}\label{eq:error-eq-1}
    \ip{\partial_t e}{e} + \lambda |e|^2 + \ip{B(e,V)}{e} + \ip{B(V_\alpha, V_\alpha)}{e} - \ip{\overline{(\overline{V}_\alpha \cdot \D)\overline{V}_\alpha}}{e} = 0.
\end{equation}
Adding and subtracting $\ip{\overline{B(V_\alpha, V_\alpha)}}{e}$ to Eq.~\eqref{eq:error-eq-1} and grouping terms gives
\begin{equation*}
    \ip{\partial_t e}{e} + \lambda |e|^2  + \ip{B(e,V)}{e} + \ip{B(V_\alpha,V_\alpha) - \overline{B(V_\alpha,V_\alpha)}}{e} + \ip{\overline{B(V_\alpha,V_\alpha)} - \overline{B(\overline{V}_\alpha,\overline{V}_\alpha)}}{e}.
\end{equation*}
Thus,
\begin{equation}\label{eq:error-eq-Ji}
    \frac{1}{2} \partial_t \abs{e}^2 + \lambda |e|^2 \leq J_1 + J_2 + J_3
\end{equation}
where
\begin{equation}
\begin{aligned}
    J_1 &= \left|\ip{B(e,V)}{e}\right|, \\
    J_2 &= \left|\ip{B(V_\alpha,V_\alpha) - \overline{B(V_\alpha,V_\alpha)}}{e}\right|,\\
    J_3 &= \left| \ip{\overline{B(V_\alpha,V_\alpha)} - \overline{B(\overline{V}_\alpha,\overline{V}_\alpha)}}{e} \right|.
\end{aligned}
\end{equation}

We now upper-bound the terms $J_1,J_2,J_3$. Using Eq.~(A.46e) in \cite[p.100]{Foias_Manley_Rosa_Temam_2001},
\begin{equation*}
\begin{aligned}
    J_1 &\leq c_1 \abs{e}^{\frac{1}{2}} \, \norm{e}^{\frac{1}{2}} \, \abs{e}^{\frac{1}{2}} \, \norm{e}^{\frac{1}{2}} \norm{V} \\
    &= c_1 \abs{e} \, \norm{e} \, \norm{V}
\end{aligned}
\end{equation*}
for some constant $c_1>0$. By self-adjointness of $\calF_{\alpha\beta}$, Lemma~\ref{lemma:alpha-norm-bound}, and the Brezis-Gallouet inequality \cite[p.99-101]{Foias_Manley_Rosa_Temam_2001}, we have
\begin{equation*}
\begin{aligned}
    J_2 &\leq \alpha^{\frac{1}{\beta}} \abs{e} \, \norm{\D B(V_\alpha,V_\alpha)} \\
    &\leq \alpha^{\frac{1}{\beta}} \abs{e} \, 
    \left( \abs{\D V_\alpha}^2 + \norm{V_\alpha}_{L^\infty} \abs{\D^2 V_\alpha} \right) \\
    &\leq \alpha^{\frac{1}{\beta}} \abs{e} \, \left( \norm{V_\alpha}^2 +c_2 \norm{V_\alpha} \, \left(1 + \log \left( \frac{\abs{A V_\alpha}^2}{\eta_1 \norm{V_\alpha}^2}\right) \right)^{\frac{1}{2}} \abs{A V_\alpha} \right)
\end{aligned} 
\end{equation*}
where $\eta_1$ is defined in Eq.~\eqref{eq:stokes} and $c_2>0$ is a constant. Using self-adjointness of $\calF_{\alpha\beta}$ and performing a simple algebra, we have
\begin{equation}\label{eq:J3-term1}
\begin{aligned}
    J_3 &= \left| \ip{{B(V_\alpha,V_\alpha)} - {B(\overline{V}_\alpha,\overline{V}_\alpha)}}{\overline{e}} \right| \\
    &= \left| \ip{{B(V_\alpha,V_\alpha)} - {B(\overline{V}_\alpha,{V}_\alpha)} + {B(\overline{V}_\alpha,{V}_\alpha)} - B(\overline{V}_\alpha,\overline{V}_\alpha)}{\overline{e}} \right| \\
    &\leq \left| \ip{B(V_\alpha - \overline{V}_\alpha, V_\alpha)}{ \overline{e}} \right| + \left| \ip{B(\overline{V}_\alpha, V_\alpha - \overline{V}_\alpha)}{\overline{e}} \right|. 
\end{aligned}
\end{equation}
By Holder's inequality, Ladyzhenskaya's inequality, Lemma~\ref{lemma:alpha-norm-bound}, Poincaré's inequality, and Young's inequality, the first term on the right-hand side of Eq.~\eqref{eq:J3-term1} can be written
\begin{equation*}
\begin{aligned}
    \left| \ip{B(V_\alpha - \overline{V}_\alpha, V_\alpha)}{\overline{e}} \right| &\leq \abs{V_\alpha - \overline{V_\alpha}} \, \norm{\D V_\alpha}_{L^4} \, \norm{\overline{e}}_{L^4} \\
    & \leq \alpha^{\frac{1}{\beta}} \abs{V_\alpha} c_2^2 \abs{\D V_\alpha}^{\frac{1}{2}} \, \abs{A V_\alpha}^{\frac{1}{2}} \, \abs{\overline{e}}^{\frac{1}{2}} \, \abs{\D \overline{e}}^{\frac{1}{2}} \\
    &\leq \alpha^{\frac{1}{\beta}} c_2^2 \eta_1^{-\frac{1}{4}} \abs{V_\alpha} \abs{A V_\alpha} \abs{\overline{e}}^{\frac{1}{2}} \abs{\D \overline{e}}^{\frac{1}{2}} \\
    & \leq \alpha^{\frac{1}{\beta}} c_2^2 K^{\frac{1}{2}} \eta_1^{-\frac{1}{4}}  \abs{V_\alpha} \abs{A V_\alpha} \abs{\overline{e}}^{\frac{1}{2}}  \abs{e}^{\frac{1}{2}}
\end{aligned}
\end{equation*}
where the last line has used the definition of $K$ from Eq.~\eqref{eq:K-delta}. Similarly, the second term on the right-hand side of Eq.~\eqref{eq:J3-term1} satisfies
\begin{equation*}
\begin{aligned}
    \left| \ip{B(\overline{V}_\alpha, V_\alpha - \overline{V}_\alpha)}{ \overline{e}} \right| &\leq c_2^2 K^{\frac{1}{2}} \eta_1^{-\frac{1}{2}} \alpha^{\frac{1}{\beta}} \norm{V_\alpha} \abs{A V_\alpha}  \abs{e}^{\frac{1}{2}}
\end{aligned}
\end{equation*}

Inserting the upper bounds on $J_1,J_2,J_3$ into Eq.~\eqref{eq:error-eq-Ji} yields
\begin{equation}\label{eq:error-eq-Ji-bounded}
\begin{aligned}
    \frac{1}{2} \partial_t \abs{e}^2 + \lambda |e|^2  &\leq
    c_1 \abs{e} \, \norm{V} \, \norm{e}\\
    &\qquad+
    \alpha^{\frac{1}{\beta}} \abs{e} \, \Gamma(V_\alpha)\\
    &\qquad+ 
    \alpha^{\frac{1}{\beta}} c_2^2 \eta_1^{-\frac{1}{4}} K^{\frac{1}{2}}  \abs{V_\alpha} \abs{A V_\alpha} \abs{\overline{e}}^{\frac{1}{2}} \abs{e}^{\frac{1}{2}}\\
    &\qquad+ \alpha^{\frac{1}{\beta}}
    c_2^2 \eta_1^{-\frac{1}{2}} K^{\frac{1}{2}}  \norm{V_\alpha} \abs{A V_\alpha}  \abs{e}^{\frac{1}{2}}
\end{aligned}
\end{equation}
where we have defined
\begin{equation*}
    \Gamma(V_\alpha) \coloneq \norm{V_\alpha}^2 +c_2 \norm{V_\alpha} \, \left(1 + \log \left( \frac{\abs{A V_\alpha}^2}{\eta_1 \norm{V_\alpha}^2}\right) \right)^{\frac{1}{2}} \abs{A V_\alpha}
\end{equation*}

We wish to write the right-hand side of Eq.~\eqref{eq:error-eq-Ji-bounded} in terms of the error $\abs{e(t)}^2$. First, notice that $\norm{V}$ can be bounded by taking the inner product of Eq.~\eqref{eq:dEB} with $AV$, using Eq.~\eqref{eq:b-skew-symmetric} and applying Gronwall's inequality \cite{gronwall1919}, giving
\begin{equation}\label{eq:V-abs-bound}
    \norm{V(t,\cdot)}^2 \le  e^{-2 \lambda t} \norm{V_0}^2.
\end{equation}
Similarly, using Eq.~\eqref{eq:b-filter-skew-sym}, it holds
\begin{equation}\label{eq:Valpha-abs-bound}
    \norm{V_\alpha(t,\cdot)}^2 \le  e^{-2 \lambda t} \norm{V_0}^2.
\end{equation}
The triangle inequality then gives $\norm{e}\leq 2 e^{- \lambda t} \norm{V_0}$. Inserting this bound back into Eq.~\eqref{eq:error-eq-Ji-bounded} and using $\norm{\overline{e}} \leq \norm{e}$ and the AM-GM inequality with $\varepsilon_{0,1,2,3}>0$, we get
\begin{equation}\label{eq:pre-wt}
\begin{aligned}
\frac{1}{2}\partial_t \abs{e} + \lambda \abs{e}^{2}
&\le 4 \, \varepsilon_0 \, c_1^2 \, e^{- 2\lambda t} \norm{V_0}^2 \abs{V}^2  + \frac{\abs{e}^2}{\varepsilon_0} \\
&\quad
+ \Gamma(V_\alpha) \, \frac{\abs{e}^{2}}{\varepsilon_{1}}
+ \varepsilon_{1}
\Gamma(V_\alpha) \, \alpha^{\frac{2}{\beta}} \\
&\quad
+ \frac{\abs{e}^{2}}{\varepsilon_{2}}
+ \varepsilon_{2}\,\alpha^{\frac{2}{\beta}}\,
\eta_1^{-\frac{1}{2}} c_2^4 K \abs{V_\alpha}^{2}\,
\abs{A{V_\alpha}}^{2}  \\
&\quad
+ \frac{\abs{e}^{2}}{\varepsilon_{3}}
+ \varepsilon_{3}\,\alpha^{\frac{2}{\beta}}\,
\eta_1^{-1} c_2^4 K\,
\norm{V_\alpha}^{2}\,
\abs{A V_\alpha}^{2}.
\end{aligned}
\end{equation}
Setting $w(t) = \abs{e(t,\cdot)}^2$ and rearranging Eq.~\eqref{eq:pre-wt} yields
\begin{equation}\label{eq:w-dynamics}
\begin{aligned}
\frac{{d}}{{d}t} w(t)
&\le{}\;
2\left(
-\lambda
 + \frac{1}{\varepsilon_0}
 + \frac{\Gamma(V_\alpha)}{\varepsilon_1}
+ \frac{1}{\varepsilon_2}
+ \frac{1}{\varepsilon_3}
\right) w(t)
\\
&\quad
+ 2\,\alpha^{\frac{2}{\beta}}
\Bigg(
\varepsilon_1
\Gamma(V_\alpha)
+ \varepsilon_{2} \eta_1^{-\frac{1}{2}} c_2^4 K
\abs{V_\alpha}^{2}\,
\abs{A{V_\alpha}}^{2}  
+ \varepsilon_{3}\,
\eta_1^{-1} c_2^4 K
\norm{V_\alpha}^{2}\,
\abs{A V_\alpha}^{2} 
\Bigg)
\\
&\quad + 8 \, \varepsilon_0 \, c_1^2 \, \norm{V_0}^2 \, e^{- 2\lambda t} \, \abs{V}^2 \\
&= h(t)\,w(t) + \alpha^{\frac{2}{\beta}} g(t) + k(t).
\end{aligned}
\end{equation}
Applying Gronwall's Lemma as stated in \cite[p.125]{Foias_Manley_Rosa_Temam_2001} to Eq.~\eqref{eq:w-dynamics} gives
\begin{equation}\label{eq:err-upper-bound-int}
    |e(t)|^2 = w(t) \leq \alpha^{\frac{2}{\beta}} \int_{0}^t g(s) e^{\int_s^t h(\sigma) d\sigma} ds + \int_{0}^t k(s) e^{\int_s^t h(\sigma) d\sigma} ds
\end{equation}
where we have used $w(0)=0$. It follows from Eq.~\eqref{eq:Valpha-abs-bound} that
\begin{equation}\label{eq:gamma-bound}
\begin{aligned}
    \frac{1}{T}&\int_{t}^{T+t} \Gamma(V_\alpha(s)) ds \leq \norm{V_0}^2
    \frac{1}{T} \int_{t}^{T+t} e^{-2 \lambda s}ds+c_2
    \norm{V_0} \, \frac{1}{T} \int_{t}^{T+t} e^{- \lambda s} \left(1 + \log \left( \frac{\abs{A V_\alpha(s)}^2}{\eta_1 \norm{V_\alpha(s)}^2}\right) \right)^{\frac{1}{2}} \abs{A V_\alpha(s)} ds \\
    &\quad\leq \norm{V_0}^2 e^{-2 \lambda t}
    (e^{-2 \lambda T} - 1)
    +c_2
    \norm{V_0} \left( \frac{1}{T} \int_{t}^{T+t} \left(1 + \log \left( \frac{\abs{A V_\alpha(s)}^2}{\eta_1 \norm{V_\alpha(s)}^2}\right) \right) ds
    \frac{1}{T} \int_{t}^{T+t} \abs{A V_\alpha(s)}^2 ds \right)^{\frac{1}{2}} 
\end{aligned}
\end{equation}
where the second line follows from the Cauchy-Schwarz inequality. It is well-known that, if the initial velocity $V_0$ lies in $H^2(\Omega)^2$ then it remains in $H^2(\Omega)^2$ for all times \cite{OLIVER1997471}. Furthermore, it follows from \cite[Sec.3]{OLIVER1997471} that $\sup_{t\in[0,T]} \abs{A V_\alpha}^2 \leq C_T$ and so, applying Jensen's inequality to Eq.~\eqref{eq:gamma-bound}, we get $\frac{1}{T}\int_{t}^{T+t} \Gamma(V_\alpha(s)) ds \leq C_\Gamma$ for some finite $C_\Gamma>0$. In turn, there exists $\varepsilon_{0,1,2,3}>0$ such that
\begin{equation*}
\begin{aligned}
    \frac{1}{T} \int_{t}^{T+t} h(s) ds = 
-\lambda
 + \frac{1}{\varepsilon_0}
 +\frac{1}{\varepsilon_1} \frac{1}{T}\int_{t}^{T+t} \Gamma(V_\alpha(s)) ds
+ \frac{1}{\varepsilon_2}
+ \frac{1}{\varepsilon_3}
< 0.
\end{aligned}
\end{equation*}
Following similar steps to above, one can easily show $\frac{1}{T}\int_{t}^{T+t} g(s) ds$ is bounded. Finally, 
\begin{equation*}
\begin{aligned}
    \frac{1}{T} \int_{t}^{T+t} k(s) ds &=  8 \, \varepsilon_0  \, c_1^2 \, \norm{V_0}^2 \, \frac{1}{T} \int_{t}^{T+t}  e^{-2\lambda s} \abs{V(s)}^2 ds \\
    &\leq \frac{\varepsilon_0  \, c_1^2 \, \norm{V_0}^4}{\lambda T} e^{-4\lambda t} \left(1 - e^{-4\lambda T} \right)
\end{aligned}
\end{equation*}
where the last line follows from Eq.~\eqref{eq:V-abs-bound}. For sufficiently large $T>0$, the second term on the right-hand side of Eq.~\eqref{eq:err-upper-bound-int} may be made arbitrarily small, and the result follows from Gronwall's Lemma \cite[p.125]{Foias_Manley_Rosa_Temam_2001}.

\end{proof}

\subsection{Discrete dEB model}\label{subsec:discrete-deb}

We shall discretize Eq.~(\ref{eq:dEB}) in the basis of divergence-free sine waves
\be
\label{eq:sine-basis}
e_p(\xi)=\sqrt{2}\,\frac{p^\perp}{\|p\|}\sin(2\pi p\cdot \xi), \qquad \xi\in \TT^2.
\ee
Here $p=(p^1,p^2)\in \ZZ^2$ is a nonzero integer wave vector and $p^\perp=(-p^2,p^1)$. Note that $p$ and $-p$ can be identified since $e_p(\xi)=e_{-p}(\xi)$.
Choose a momentum cutoff $\Lambda$ and let $p_1,\ldots,p_N\in \ZZ^2$ be the list of all nonzero wave vectors $p$ satisfying $\|p\|\le \Lambda$
with one representative picked from each pair $p$ and $-p$. The following proposition collects all the key properties of the discrete dEB model:  
\begin{prop}
 Galerkin projection of Eq.~(\ref{eq:dEB}) onto the span of $e_{p_1},\ldots,e_{p_N}$ yields an $N$-dimensional ODE
\be
\label{Euler2}
\dfrac{dX_i(t)}{dt} = -\lambda X_i(t)  + \sum_{j,k=1}^N c_{ijk} X_j(t) X_k(t), 
\ee
with coefficients
\be\label{Euler_cijk}
% c_{ijk} =\tilde{c}_{ijk}w(p_i) w(p_j) w(p_k), \qquad w(p) = \frac1{1 + \alpha \|p\|^\beta},
c_{ijk} =
\frac{\pi}{\sqrt{2}}\,
\frac{(p_j^\perp\cdot p_k)(\|p_j\|^2-\|p_k\|^2)}
{\|p_i\|\,\|p_j\|\,\|p_k\|}\,
\delta_{ijk} \, w(p_i) w(p_j) w(p_k),
\ee
where
\be 
\label{delta_ijk}
w(p) = \frac1{1 + \alpha \|p\|^\beta}, \qquad 
\delta_{ijk}=
\left\{
\ba{rcl}
1 & \mbox{if} & p_j+p_k\in\{p_i,-p_i\},\\
-1 & \mbox{if} & p_j-p_k\in\{p_i,-p_i\},\\
0 & \mbox{otherwise}.
\ea
\right.
\ee
The tensor $c_{ijk}$ is symmetric: $c_{ijk}=c_{ikj}$. It has $s=4$ drift channels. It obeys the divergence-free conditions with $\lambda_1=\ldots=\lambda_N=\lambda$. 
In particular,
\be
\label{Euler_div_free1}
c_{ijk}=0 \quad \mbox{unless all indices $i,j,k$ are distinct}
\ee
and
\be
\label{Euler_div_free2}
c_{ijk} + c_{jki} + c_{kij} =0
\ee
for any triple $i,j,k$. Finally, define \[
K_{\alpha,\beta}:=\sup_{r\ge 0}\frac{r^2}{(1+\alpha r^\beta)^2}
\]
and
\[
S_i^{(1)}(N):=\sum_{j,k=1}^N |c_{ijk}|^2,
\qquad
S_i^{(2)}(N):=\sum_{j,k=1}^N |c_{jik}|^2.
\]
Then, for every fixed $i$,
\begin{align}
S_i^{(1)}(N)&\le 
8\pi^2 K_{\alpha,\beta}\left(9+\frac{2^{2\beta+1}\pi^2(\beta-2)}{\alpha^{2/\beta}\beta^2 \sin(2\pi/\beta)}\right) ,&\beta>1\label{eq:1beta>1}\\
S_i^{(1)}(N)&\le 8\pi^2K_{\alpha,\beta}\left(9+
\frac{8\pi}{\alpha^2}\log(1+\alpha (2\sqrt{N}+1))\right),&\beta=1\label{eq:1beta=1}\\
S_i^{(2)}(N)&\le 32\pi^2 K_{\alpha,\beta}\left(9+\frac{2^{2\beta+1}\pi^2(\beta-2)}{\alpha^{2/\beta}\beta^2 \sin(2\pi/\beta)}\right) ,&\beta>1\label{eq:2beta>1}\\
S_i^{(2)}(N)&\le 32\pi^2K_{\alpha,\beta}\left(9+
\frac{8\pi}{\alpha^2}\log(1+\alpha (2\sqrt{N}+1))\right),&\beta=1\label{eq:2beta=1}
\end{align}
so that the drift strength $J$ of the tensor $c$ is: $J=O(K_{\alpha,\beta})$ if $\beta>1$ and $J=O(\log(N))$ for $\beta=1$. In addition, $\sum_{i,j,k=1}^N |c_{ijk}|$ diverges to infinity as $N\to\infty$ polynomially in $N$ if $1<\beta<5/3$, and only logarithmically if $\beta=5/3$, and it stays bounded independent of $N$ if $\beta>5/3$. 
\end{prop}
\begin{proof}
    We approximate the velocity field by the Galerkin ansatz 
\be
\label{Vansatz}
V(t,\xi)=\sum_{i=1}^N X_i(t)e_{p_i}(\xi),\quad X_i(t)\text{ are coordinates of } V\text{ in the basis }e_{p_i}. 
\ee
%where $X_i(t)$ are  real variables. 
% \be
% \label{Euler_cijk}
% \tilde{c}_{ijk}=
% \frac{\pi}{\sqrt{2}}\,
% \frac{(p_j^\perp\cdot p_k)(\|p_j\|^2-\|p_k\|^2)}
% {\|p_i\|\,\|p_j\|\,\|p_k\|}\,
% \delta_{ijk},
% \ee

%\begin{proof}[\bf Proof of Lemma~\ref{lemma:Euler}]
Let us prove Eq.~(\ref{Euler_cijk}). Let us first do a formal derivation for classical Euler equation. 
\be
\label{eq:dEB1}
\frac{\partial}{\partial t} V +(V \cdot \nabla) V +\lambda V  = -\nabla P.
\ee
Let us take the inner product between 
both sides of Eq.~(\ref{eq:dEB1}) and $e_{p_i}\equiv e_{p_i}(\xi)$. We get
\[
\frac{\partial}{\partial t} (e_{p_i} \cdot V) +e_{p_i}\cdot (V \cdot \nabla)V = -(e_{p_i} \cdot \nabla) P.
\]
Next  take the integral over $\xi\in \TT^2$.
Integration by parts and the identity $\nabla \cdot e_{p_i}=0$ give
\[
\int_{\xi \in \TT^2} d\xi \, (e_{p_i}\cdot \nabla) P= -\int_{\xi \in \TT^2} d\xi \, P (\nabla \cdot e_{p_i})=0.
\]
Thus the pressure term does not contribute to the integral. Using the identity $\int_{\xi \in \TT^2} d\xi\, e_{p_i} \cdot e_{p_j} = \delta_{i,j}$ one gets
\[
\dot{X}_i = -\int_{\xi \in \TT^2} d\xi\, e_{p_i} \cdot (V \cdot \nabla)V,
\]
where $\dot{X}_i\equiv dX_i(t)/dt$ and 
$V$ is determined by the Galerkin ansatz Eq.~(\ref{Vansatz}). Expanding the product gives
\[
\dot{X}_i = -2\sum_{j,k=1}^N X_j X_k \frac1{\|p_j\|  \|p_k\|} \int_{\xi \in \TT^2} d\xi\, (e_{p_i} \cdot p_k^\perp) \sin{(2\pi p_j\cdot \xi)}  (p_j^\perp \cdot \nabla)\sin{(2\pi p_k\cdot \xi)}
\]
Taking the gradient of $\sin{(2\pi p_k\cdot \xi)}$, substituting the explicit formula for the sine wave $e_{p_i}$ and rearranging the terms one gets
\be
\label{dotX_i_explicit}
\dot{X}_i = -\pi \sqrt{2} \sum_{j,k=1}^N X_j X_k \frac1{\|p_i\| \|p_j\|  \|p_k\|} (p_i^\perp \cdot p_k^\perp) (p_j^\perp \cdot p_k) \gamma_{ijk},
\ee
where
\[
\gamma_{ijk}=
4 \int_{\xi \in \TT^2} d\xi\,  \sin{(2\pi p_i\cdot \xi)}  \sin{(2\pi p_j\cdot \xi)}  \cos{(2\pi p_k\cdot \xi)}.
\]
The integral can be taken using the  identities  $2\sin{(a)}\cos{(b)} = \sin{(a+b)} + \sin{(a-b)}$
and 
\[
2\int_{\xi \in \TT^2} d\xi\, \sin(2\pi p \cdot \xi)\sin{(2\pi q\cdot \xi)} = \delta_{p,q} - \delta_{p,-q}.
\]
It gives
\[
\gamma_{ijk} = \delta_{p_i,p_j+p_k} - \delta_{p_i,-p_j-p_k} + \delta_{p_i,p_j-p_k} - \delta_{p_i,p_k-p_j}.
\]
Let $\Gamma_{ijk}$
be the factor
$(p_i^\perp \cdot p_k^\perp) (p_j^\perp \cdot p_k) \gamma_{ijk}$
that appears in Eq.~(\ref{dotX_i_explicit})
symmetrized under the swap of indices $j$ and $k$, that is, 
\[
\Gamma_{ijk}=\frac12 \left[(p_i^\perp \cdot p_k^\perp) 
(p_j^\perp \cdot p_k)\gamma_{ijk}
+(p_i^\perp \cdot p_j^\perp) 
(p_k^\perp \cdot p_j)\gamma_{ikj}\right].
\]
Then Eq.~(\ref{dotX_i_explicit}) is equivalent to 
\be
\label{dotX_i_explicit1}
\dot{X}_i = -\pi \sqrt{2} \sum_{j,k=1}^N X_j X_k \frac1{\|p_i\| \|p_j\|  \|p_k\|} \Gamma_{ijk}.
\ee
Note that $p_k^\perp \cdot p_j=-p_j^\perp \cdot p_k$ and $p^\perp \cdot q^\perp = p\cdot q$ for any vectors $p,q\in \RR^2$. Thus
\[
\Gamma_{ijk}=\frac12 (p_j^\perp \cdot p_k)
\left[(p_i \cdot p_k ) \gamma_{ijk}
- (p_i  \cdot p_j) 
\gamma_{ikj}\right].
\]
We are led to consider three cases.

\noindent
{\em Case 1:} $p_i=\pm (p_j + p_k)$. Then $\gamma_{ijk}=\gamma_{ikj}=\pm 1$ and
\begin{align*}
\Gamma_{ijk}= \frac12 (p_j^\perp \cdot p_k)   (\|p_k\|^2 - \|p_j\|^2).
\end{align*}

\noindent
{\em Case 2:} $p_i=\pm (p_j -p_k)$. Then $\gamma_{ijk}=-\gamma_{ikj} =\pm 1$
and
\begin{align*}
\Gamma_{ijk}=- \frac12 (p_j^\perp \cdot p_k)   (\|p_k\|^2 - \|p_j\|^2).
\end{align*}

\noindent
{\em Case 3:} None of the above. Then $\gamma_{ijk}=\gamma_{ikj}=0$ and thus $\Gamma_{ijk}=0$.

Now, combining all these cases and using Eq.~(\ref{dotX_i_explicit1}) 
 one arrives at Eq.~(\ref{Euler_cijk}) by noting that $\bar e_{p_i} = w(p_i) e_{p_i}$. 
 
Clearly the tensor $c_{ijk}$ has $s=4$ drift channels thanks to the presence of $\delta_{ijk}$ in its definition. Let us now check the symmetry and the divergence-free conditions. Clearly, the factors $w(p_i)w(p_j)w(p_k)$ do not impact neither of those properties: i.e. if $c_{ijk}$ is symmetric and divergence free for $\alpha=0$ then so it is for any $\alpha,\beta>0$. Therefore, wlog, we set $\alpha=0$ in the following derivations. By definition,
\be
c_{ijk} \sim \gamma_{ijk}, \qquad 
\gamma_{ijk}=(p_j^\perp\cdot p_k)(\|p_j\|^2-\|p_k\|^2)\delta_{ijk},
\ee
where $\sim$ hides the factor symmetric under all permutations of indices $i,j,k$.
We have $\gamma_{ijk}=\gamma_{ikj}$ since $p_j^\perp \cdot p_k = - p_j \cdot p_k^\perp$ and $\delta_{ijk}=\delta_{ikj}$.
This proves the symmetry $c_{ijk}=c_{ikj}$.

To prove Eq.~(\ref{Euler_div_free1}), it suffices to check
that $\gamma_{ijk}=0$ unless all three indices $i,j,k$ are distinct. Consider three cases.

\noindent
{\em Case~1:} $j=k$. Then $\gamma_{ijk}=0$ since $\|p_j\|=\|p_k\|$.

\noindent
{\em Case~2:} $i=j$. Recall that $\delta_{ijk}=0$ unless $p_j+p_k \in \{p_i,-p_i\}$
or $p_j -p_k \in \{p_i,-p_i\}$, see Eq.~(\ref{delta_ijk}). Since all wave vectors are nonzero, this is possible only if $p_k=2p_i$.
Then the factor $p_j^\perp \cdot p_k=p_i^\perp \cdot p_k$ in $\gamma_{ijk}$ vanishes and thus $\gamma_{ijk}=0$.

\noindent
{\em Case~3:} $i=k$. This is the same as Case~2 since $\gamma_{ijk}=\gamma_{ikj}$. 

This proves Eq.~(\ref{Euler_div_free1}). 
From the definition Eq.~(\ref{delta_ijk}) one infers that  $\delta_{ijk}=0$ unless 
\be
\label{delta_nonzero_condition}
\sigma_i p_i + \sigma_j p_j + \sigma_k p_k = (0,0) \qquad \mbox{for some $\sigma_i,\sigma_j,\sigma_k\in \{1,-1\}$}.
\ee
In the latter case $\delta_{ijk} = \sigma_j \sigma_k$. It suffices to prove  Eq.~(\ref{Euler_div_free2}) for  any triple of distinct indices $i,j,k$ satisfying Eq.~(\ref{delta_nonzero_condition})
since otherwise all three terms in Eq.~(\ref{Euler_div_free2})  are zero. 
We have
\[
p_k^\perp \cdot p_i = -\sigma_i p_k^\perp \cdot  (\sigma_j p_j + \sigma_k p_k) = - \sigma_i \sigma_j p_k^\perp \cdot p_j = \sigma_i \sigma_j p_j^\perp \cdot p_k
\]
and
\[
p_i^\perp \cdot p_j = -\sigma_i( \sigma_j p_j^\perp + \sigma_k p_k^\perp) \cdot p_j = -\sigma_i \sigma_k p_k^\perp \cdot p_j  = \sigma_i \sigma_k p_j^\perp \cdot p_k.
\]
Noting that $\delta_{ijk}=\sigma_j \sigma_k$, $\delta_{jki}=\sigma_k \sigma_i$, and $\delta_{kij}=\sigma_i\sigma_j$ one gets
\[
\gamma_{ijk} + \gamma_{jki} + \gamma_{kij} = (p_j^\perp \cdot p_k) \sigma_j \sigma_k \left( \| p_j\|^2 - \|p_k\|^2  + \|p_k\|^2 - \|p_i\|^2 + \| p_i\|^2 - \|p_j\|^2 \right) =0.
\]

Let us  demonstrate the inequalities Eqs.~(\ref{eq:1beta>1}--\ref{eq:2beta>1}). Recall from the definition of $\delta_{ijk}$ that 
\begin{equation}\label{eq:pjpkpi}
    p_j+p_k\in\{p_i,-p_i\} \text{ or } p_j-p_k\in\{p_i,-p_i\}
\end{equation}
Equation~\eqref{eq:pjpkpi} implies that %Considering $p_j+p_k\in\{p_i,-p_i\}$ or $p_j-p_k\in\{p_i,-p_i\}$ it is evident that
\[
\bigl|\|p_j\|^2-\|p_k\|^2\bigr|
=
\bigl|\left((p_j-p_k)\cdot (p_j+p_k)\right)\bigr|
\leq \|p_j-p_k\|\|p_j+p_k\|
\leq \|p_i\|(\|p_j\|+\|p_k\|)
\]
Then, by using 
\(
|p_j^\perp\!\cdot p_k|\le \|p_j\|\,\|p_k\|,
\)
we get
\[
    \left|
    \frac{(p_j^\perp\!\cdot p_k)(\|p_j\|^2-\|p_k\|^2)}
    {\|p_i\|\,\|p_j\|\,\|p_k\|}
    \right|
    \leq
    \|p_j\|+\|p_k\|
\]

% To estimate \(S_i^{(1)}(N)\) we set 
% \(
%     c:=p_i,\; a:=p_j,\; b:=p_k,
% \)
% and use the last bound 
so that
\[
|c_{ijk}|^2
\le
\frac{\pi^2}{2}\,(\|p_j\|+\|p_k\|)^2\,w(p_i)^2w(p_j)^2w(p_k)^2.
\]
Since $w(p_i)\le 1$, and
\(
(\|p_j\|+\|p_k\|)^2\le 2\|p_j\|^2+2\|p_k\|^2,
\) 
we obtain
\[
|c_{ijk}|^2
\le
\pi^2\bigl(\|p_j\|^2 w(p_j)^2w(p_k)^2+\|p_k\|^2w(p_j)^2w(p_k)^2\bigr).
\]
By the definition of $K_{\alpha,\beta}$,
\[
\|p_j\|^2w(p_j)^2\le K_{\alpha,\beta},
\qquad
\|p_k\|^2w(p_k)^2\le K_{\alpha,\beta},
\]
so
\[
|c_{ijk}|^2
\le
\pi^2K_{\alpha,\beta}\bigl(\psi(p_j)+\psi(p_k)\bigr).
\]
Hence 
\begin{align}
S_i^{(1)}(N)&=\sum_{j,k=1}^N c_{ijk}^2 \le 
\pi^2 K_{\alpha,\beta}
\sum_{k,j=1}^N 
\bigl(\psi(p_j) + \psi(p_k)\bigr)\delta_{ijk}\\ 
&= \pi^2 K_{\alpha,\beta} \sum_{j=1}^N \bigl(4\psi(p_j)+ 2\psi(p_j-p_i)+2\psi(p_j+p_i)\bigr)\\
&\le \pi^2 K_{\alpha,\beta} 8 M_N \label{eq:S1-via-MN-both}
\end{align}
where \[
    M_N:=\max_{\substack{A\subset\mathbb Z^2\\ |A|\le N}}\sum_{p\in A}\psi(p),\qquad \psi(p):=w(p)^2,\qquad  A_N=\{p_1,\dots,p_N\}, p_i\in\mathbb Z^2.
\]
and $|A|$ denotes the number of elements of the set $A$.
% Summing over all admissible pairs, and using the fact that for each fixed \(i\) and \(a\) there are at most four admissible values of \(b\), we get
% \[
% S_i^{(1)}(N)
% \le
% \pi^2K_{\alpha,\beta}
% \sum_{a\in A_N}\sum_{\ell=1}^4
% \bigl(\psi(a)+\psi(b_\ell(a))\bigr),
% \]
% where each $b_\ell(a)$ is one of
% \(
% c-a,\; -c-a,\; a-c,\; a+c.
% \)
% Each map \(a\mapsto b_\ell(a)\) is injective on \(\mathbb Z^2\), so
% \[
% \sum_{a\in A_N}\psi(a)\le M_N,
% \qquad
% \sum_{a\in A_N}\psi(b_\ell(a))\le M_N.
% \]
% Hence
% \be
% S_i^{(1)}(N)\le 8\pi^2K_{\alpha,\beta}\,M_N.
% \label{eq:S1-via-MN-both}
% \ee
%     c:=p_i,\; a:=p_j,\; b:=p_k,
To estimate \(S_i^{(2)}(N)\) we note that 
\[
|c_{jik}|^2
\le\frac{\pi^2}{2}\,(\|p_i\|+\|p_k\|)^2\,w(p_j)^2 w(p_i)^2 w(p_k)^2.
\]
To eliminate $\|p_i\|$ in the latter we recall Eq.~\eqref{eq:pjpkpi} which implies that
\[
\|p_i\|\le \|p_j\|+\|p_k\|,
\]
and so
\[
|c_{jik}|^2
\le
\frac{\pi^2}{2}\,4(\|p_j\|+\|p_k\|)^2\,w(p_j)^2w(p_i)^2w(p_k)^2.
\]
Now, following similar logic as for \(S_i^{(1)}(N)\) we obtain
\be
    S_i^{(2)}(N)\le 32\pi^2K_{\alpha,\beta}\,M_N.
\label{eq:S2-via-MN-both}    
\ee
Now Eqs.~(\ref{eq:1beta>1}--\ref{eq:2beta>1}) are obtained by taking estimates of $M_N$ provided below by Lemma~\ref{lem:MN-by-coarea}, and plugging them into Eq.~\eqref{eq:S1-via-MN-both} and Eq.~\eqref{eq:S2-via-MN-both}. 
\begin{lemma}
\label{lem:MN-by-coarea}
\begin{align}
M_N&\le 
%9+ \frac{2^{2\beta+1}\pi}{\beta}\alpha^{-2/\beta}\Gamma\!\left(\frac{2}{\beta}\right)\Gamma\!\left(2-\frac{2}{\beta}\right)=
9+\frac{2^{2\beta+1}\pi^2(\beta-2)}{\alpha^{2/\beta}\beta^2 \sin(2\pi/\beta)},&\beta>1\label{eq:beta>1}\\
M_N&\le 9+
% 2^{2\beta}I(2\sqrt N+1)
\frac{8\pi}{\alpha^2}\log(1+\alpha (2\sqrt{N}+1)),&\beta=1\label{eq:beta=1}
\end{align}
\end{lemma}
\begin{proof}
% We first estimate $M_N$ in terms of the integral $I(R)$.
% \textbf{Step 1: A maximal weighted lattice sum.}
Since $\psi(p)$ is radial and non-increasing in $\|p\|$, the maximum in the definition of $M_N$ is attained at a set $A_N$ composed of $N$ lattice points $p_i$ which are near the origin. Thus, taking $R_N:=2\sqrt N$ we get
\be
M_N\le \sum_{\substack{p\in\mathbb Z^2\\ \|p\|\le 2\sqrt N}}\psi(p).
\label{eq:MN-disk-both}
\ee
For each lattice point $p\in\mathbb Z^2$, define the unit square
\[
Q_p:=p+\left[-\frac12,\frac12\right]^2.
\]
The squares $Q_p$ are pairwise disjoint and have area $1$. If $\|p\|\ge 2$ and $x\in Q_p$, then
\[
\|x\|\le \|p\|+\frac{\sqrt2}{2}\le 2\|p\|.
\]
so that \( 1+\alpha \|p\|^\beta \ge 1+\alpha 2^{-\beta}\|x\|^\beta \).
Using the inequality
\[
\frac{1}{(1+\lambda t)^2}\le \lambda^{-2}\frac{1}{(1+t)^2},
\qquad 0<\lambda\le 1,\quad t\ge 0,
\]
with \(\lambda=2^{-\beta}\) and \(t=\alpha\|x\|^\beta\), we obtain
\begin{equation}
    \label{eq:psipx}
    \psi(p)= \frac{1}{(1+\alpha\|p\|^\beta)^2}\le 2^{2\beta}\frac{1}{(1+\alpha\|x\|^\beta)^2}.
\end{equation}
Integrating both sides of Eq.~\eqref{eq:psipx} over $Q_p$ and recalling that $\int_{Q_p}dx=1$ yields
\[
\psi(p)\le
2^{2\beta}\int_{Q_p}\frac{dx}{(1+\alpha\|x\|^\beta)^2}.
\]
Summing over all lattice points with \(2\le \|p\|\le 2\sqrt N\), and using
\[
Q_p\subset \{|x|\le 2\sqrt N+1\}\Rightarrow \cup_{2\le \|p\|\le 2\sqrt N} Q_p \subset \{|x|\le 2\sqrt N+1\}
\]
we obtain
\begin{align*}
    \sum_{\substack{p\in\mathbb Z^2\\ 2\le \|p\|\le 2\sqrt N}}
    \psi(p)
&\le 2^{2\beta}\int_{\bigcup_{\substack{2\le \|p\|\le 2\sqrt N}} Q_p}\frac{dx}{(1+\alpha\|x\|^\beta)^2}\\
&\le 
2^{2\beta}\int_{|x|\le 2\sqrt N+1}\frac{dx}{(1+\alpha\|x\|^\beta)^2}
=
2^{2\beta}I(2\sqrt N+1).
\end{align*}
where we defined 
\[
I(R):=\int_{x_1^2+x_2^2\le R^2}
\frac{dx_1\,dx_2}{\bigl(1+\alpha(x_1^2+x_2^2)^{\beta/2}\bigr)^2}
=
\int_{|x|\le R}\frac{dx}{(1+\alpha\|x\|^\beta)^2}.
\]

Since $\sum_{\substack{p\in\mathbb Z^2\\\|p\|<2}} \psi(p) \le 9$ we get: 
\be
\sum_{\substack{p\in\mathbb Z^2\\ \|p\|\le 2\sqrt N}}\psi(p)
\le
9+2^{2\beta}I(2\sqrt N+1).
\label{eq:lattice-by-I-both}
\ee
Combining Eq.~\eqref{eq:MN-disk-both} and Eq.~\eqref{eq:lattice-by-I-both} yields
\be
M_N\le 9+2^{2\beta}I(2\sqrt N+1).
\label{eq:MN-by-I-both}
\ee

% \textbf{Step 4: Explicit estimates for \(I(R)\).}
Now we bound the integral $I(R)$, for which we first use the coarea formula (equivalently, polar coordinates) to get
\[
I(R)=2\pi\int_0^R \frac{r\,dr}{(1+\alpha r^\beta)^2}.
\]

Then, consider two cases:

\smallskip
\emph{Case \(\beta>1\).}
With the substitution \(t=\alpha r^\beta\), we get
\[
r\,dr=\frac1\beta \alpha^{-2/\beta} t^{\frac{2}{\beta}-1}\,dt,
\]
and therefore
\be
    I(R)
    =
    \frac{2\pi}{\beta}\alpha^{-2/\beta}
    \int_0^{\alpha R^\beta}\frac{t^{\frac{2}{\beta}-1}}{(1+t)^2}\,dt
    % \le
    % \frac{2\pi}{\beta}\alpha^{-2/\beta}
    % \Gamma\!\left(\frac{2}{\beta}\right)
    % \Gamma\!\left(2-\frac{2}{\beta}\right).
    % \label{eq:I-beta>1-both}
\ee
Using the substitution
\[
    u=\frac{t}{1+t} \iff t=\frac{u}{1-u} 
\]
we get (here $\Gamma$ denotes the standard Euler function): 
\be\begin{aligned}
    I(R)
    &=
    \frac{2\pi}{\beta}\alpha^{-2/\beta}
    \int_0^{\frac{\alpha R^\beta}{1+\alpha R^\beta}}u^{\frac{2}{\beta}-1}(1-u)^{1-\frac{2}{\beta}}\,du \\
    &\le
    \frac{2\pi}{\beta}\alpha^{-2/\beta}
    \int_0^{1}u^{\frac{2}{\beta}-1}(1-u)^{1-\frac{2}{\beta}}\,du \\
    &=
    \frac{2\pi}{\beta}\alpha^{-2/\beta}
    \Gamma\!\left(\frac{2}{\beta}\right)
    \Gamma\!\left(2-\frac{2}{\beta}\right)\\
    &=\frac{2\pi}{\beta}\alpha^{-2/\beta}(1-2/\beta)\Gamma\!\left(\frac{2}{\beta}\right)
    \Gamma\!\left(1-\frac{2}{\beta}\right) = \frac{2\pi}{\beta}\alpha^{-2/\beta}(1-2/\beta)\frac{\pi}{\sin(2\pi/\beta)}
    \label{eq:I-beta>1-both}
\end{aligned}\ee

% \\
% &\le
% \frac{2\pi}{\beta}\alpha^{-2/\beta}
% \Gamma\!\left(\frac{2}{\beta}\right)
% \Gamma\!\left(2-\frac{2}{\beta}\right) \\
% &=
% \frac{2\pi}{\beta}\alpha^{-2/\beta}
% \int_0^{\alpha R^\beta}\frac{t^{\frac{2}{\beta}-1}}{(1+t)^2}\,dt
% \le
% \frac{2\pi}{\beta}\alpha^{-2/\beta}
% \Gamma\!\left(\frac{2}{\beta}\right)
% \Gamma\!\left(2-\frac{2}{\beta}\right).

\smallskip
\emph{Case \(\beta=1\).}
Here
\[
    I(R)=2\pi\int_0^R \frac{r\,dr}{(1+\alpha r)^2}.
\]
With \(u=1+\alpha r\), we find
\be
    I(R)=
    \frac{2\pi}{\alpha^2}
    \left(
        \log(1+\alpha R)+\frac{1}{1+\alpha R}-1
    \right)
    \le
    \frac{2\pi}{\alpha^2}\log(1+\alpha R).
\label{eq:I-beta=1-both}
\ee
Substituting Eq.~\eqref{eq:I-beta>1-both} and Eq.~\eqref{eq:I-beta=1-both} into the bound Eq.~\eqref{eq:MN-by-I-both} gives the stated estimates.
\end{proof}
Finally define
\[
\Lambda_R:=\{p\in\mathbb Z^2\setminus\{0\}:\|p\|\le R\}.
\] Partitioning integer lattice points into dyadic shells it is not hard to show that there exists a constant $c_{\alpha,\beta}>0$ such that, for all
sufficiently large $R$,
\[
\sum_{p_i,p_j,p_k\in \Lambda_R}|c_{ijk}|
\ge c_{\alpha,\beta} R^{5-3\beta}.
\]
which diverges to infinity as $R\to\infty$ whenever $1<\beta<5/3$. Moreover \[
\sum_{p_i,p_j,p_k\in \Lambda_R}|c_{ijk}|\ge \log(R)\,\qquad \beta=5/3.
\]
On the other hand \[
\sum_{p_i,p_j,p_k\in \Lambda_R} |c_{ijk}|
\le
\begin{cases}
C_{\alpha,\beta}, & \beta>\frac53,\\[0.7ex]
C_{\alpha}\log R, & \beta=\frac53.
\end{cases}
\]
This completes the proof. 
\end{proof}

Finally, by adding Wiener noise we convert the ODE Eq.~\eqref{Euler2} to an SDE of the form Eq.~(\ref{SDE})
that can be simulated by our quantum algorithm:
\be
\label{eq:deb-sde}
dX_i(t) = -\lambda X_i(t)dt  + \sum_{j,k=1}^N c_{ijk} X_j(t) X_k(t) dt + \sqrt{q} \, dW_i(t).
\ee

\subsection{Numerical experiment}\label{subsec:monte-carlo-experiment}

Here we provide a complete description of the numerical experiment from Section~\ref{sec:methods}. 
For the convenience of the reader, we recall the definitions and results therein.  We classically simulate the dynamics of
\begin{equation}\label{eq:numerics-v(t,x)}
v(t,x)=\EE_{z} u(t,x+z) = \EE_z(\EE_W\,u_{0,\xi^*}(X(t) \,)
\end{equation}
for $X(t)$ solving the 2D discrete dEB model Eq.~\eqref{eq:deb-sde} over $N=220$ basis functions with initial condition $X(0)=x$, tensor $c$ defined by Eq.~\eqref{Euler_cijk}, and an observable function $u_{0,\xi^*}$.

\paragraph{Euler-Maruyama method.} For classical simulation of the dEB model, we use the Euler-Maruyama (EM) method \cite{Maruyama1955ContinuousMP}, which is a generalization of the Euler method for SDEs. We outline the method as follows. Consider a general $N$-dimensional SDE of the form
\begin{equation}
\label{eq:general-sde}
    {d} X(t) = a(X(t),t) {d}t + b(X(t),t) {d}W(t), \qquad\quad a,b:\mathbb{R}^N\to\mathbb{R}^N,
\end{equation}
with initial condition $X(t_0) = x$. The functions $a$ and $b$ constitute the deterministic and stochastic parts of the SDE, respectively. Fix some final evolution time $T>0$ and consider $D$ equally-spaced timepoints
\begin{equation*}
    t_0< t_1 < t_2 < \hdots < t_{D} = T, \qquad \Delta t = \frac{T}{D}.
\end{equation*}
Following the EM-method, we set $Z_0 = x$ and obtain approximations $Z^{(i)}$ of $X(t_i)$ via
\begin{equation}\label{eq:euler-maruyama}
    Z^{(i)} = Z^{(i-1)} + a\left(Z^{(i-1)},t_{i-1}\right) \Delta t + b\left(Z^{(i-1)},t_{i-1}\right) \Delta W^{(i-1)},
\end{equation}
for $i=1,...,D$, where $\Delta W^{(i)}\sim\mathcal{N}(0, I_N \Delta t)$ are Brownian increments. Notice if $b\equiv 0$ then Eq.~\eqref{eq:general-sde} is fully deterministic and the time-stepping scheme Eq.~\eqref{eq:euler-maruyama} reduces to the classical Euler method for ODEs. For the dEB model Eq.~\eqref{eq:deb-sde}, $a$ and $b$ are given component-wise by
\begin{equation*}
    a_i(X(t),t) = - \lambda X_i(t) + \sum_{j,k=1}^N c_{ijk} X_j(t) X_k(t), \qquad \quad
    b_i(X(t),t) = \sqrt{q},
\end{equation*}
for $c_{ijk}$ defined in Eq.~\eqref{Euler_cijk}.

The EM-method is weakly first-order convergent \cite{KloedenPlaten1992}, namely
\begin{equation}\label{eq:weak-convergence}
    \left| \mathbb{E}[ \varphi(X(T)) ] - \mathbb{E}[\varphi(Z^{(D)})] \right| = O\left({\Delta t}\right)
\end{equation}
for continuous and bounded $\varphi: \mathbb{R}^N\to \mathbb{R}$. However, in practice, the expectation $\mathbb{E}[\varphi(Z^{(D)})]$ is approximated by the sample mean
\begin{equation*}
    \widehat{\varphi}_P = \frac{1}{P} \sum_{i=1}^P \varphi(Z^{(D,i)})
\end{equation*}
where $Z^{(D,1)},...,Z^{(D,P)}$ are obtained from $P$ runs of the EM-method over different noise realizations. Following the standard Monte Carlo procedure \cite{fishman1995montecarlo}, for $Z^{(D)}$ with sample variance $s^2$, it holds
\begin{equation*}
    \left| \widehat{\varphi}_P - \mathbb{E}[\varphi(Z^{(D)})] \right| \leq \frac{s z}{\sqrt{P}},
\end{equation*}
with probability determined by the selected $z$-score. Applying the triangle inequality and taking timestep $\Delta t = C_1/ \sqrt{P}$ we get the error estimate
\begin{equation}\label{eq:em-mcmc-error}
\begin{aligned}
    \big| \mathbb{E}[ \varphi(X(T)) ] -  \widehat{\varphi}_P  \big| &\leq \big| \mathbb{E}[ \varphi(X(T)) ] - \mathbb{E}[\varphi(Z^{(D)})] \big| + \big| \mathbb{E}[\varphi(Z^{(D)})] - \widehat{\varphi}_P \big| \\
    &\leq  \frac{C_1 C_2  + s z}{\sqrt{P}} \\
    &=:\varepsilon_P.
\end{aligned}
\end{equation}
The constant $C_1>0$ is chosen such that the EM-method remains stable and the constant $C_2>0$ arises from Eq.~\eqref{eq:weak-convergence} and depends on the underlying SDE in Eq.~\eqref{eq:general-sde}. 

\paragraph{Turbulent initial condition.} To illustrate the growing simulation complexity due to the impact of a strong quadratic nonlinearity, we generate a turbulent initial condition (IC) $x$ by simulating the dEB model with deterministic Kolmogorov forcing. The Kolmogorov forcing term, defined by
\begin{equation}\label{eq:k-forcing}
    F(t,\xi)
    = 
    F_0
    \begin{pmatrix}
     \sin \left(2 \pi k_f \, \xi_2\right), & 0 
\end{pmatrix}
^T
, \qquad 
F_0 \in \mathbb{R}, \,\, k_f \in \mathbb{Z},
\end{equation}
was originally introduced in the context of the Navier-Stokes equation (NSE) \cite{meshalkin1961investigation} and serves as a prototypical system for the study of turbulent dynamics. 

In the noiseless case ($q=0$), the dEB model with forcing Eq.~\eqref{eq:k-forcing} and zero initial condition $V(t=0,\xi)=(0,0)^T$ has the analytical solution
\be
\label{eq:deb-analytical-sol}
V(t,\xi) =
\frac{F_0}{\lambda} ( 1 - e^{-\lambda t})
\,
\begin{pmatrix}
     \sin \left(2 \pi k_f \, \xi_2\right),
    & 0
\end{pmatrix}
^T
,
\ee
which can be viewed as the analogue of ``Kolmogorov flow" in the 2D NSE \cite{meshalkin1961investigation}. The solution Eq.~\eqref{eq:deb-analytical-sol} appears qualitatively as horizontal ``bands'', which evolve towards the steady state
\[
\lim_{t\to\infty} V(t,\xi) = \frac{F_0}{\lambda} \big( \sin (2 \pi k_f \, \xi_2), 0 \big)^T.
\]
We note further that the solution in Eq.~\eqref{eq:deb-analytical-sol} satisfies $\overline{(\overline V\cdot \nabla)\overline V} = 0$, and so dEB simplifies to a linear model in this case. However, in the noisy case ($q>0$), the nonlinearity $\overline{(\overline V\cdot \nabla)\overline V}$ becomes increasingly pronounced over time due to the Wiener noise until it eventually dominates the linear part of dEB. 

\begin{figure}[htbp]
\centering\includegraphics[width=\textwidth]{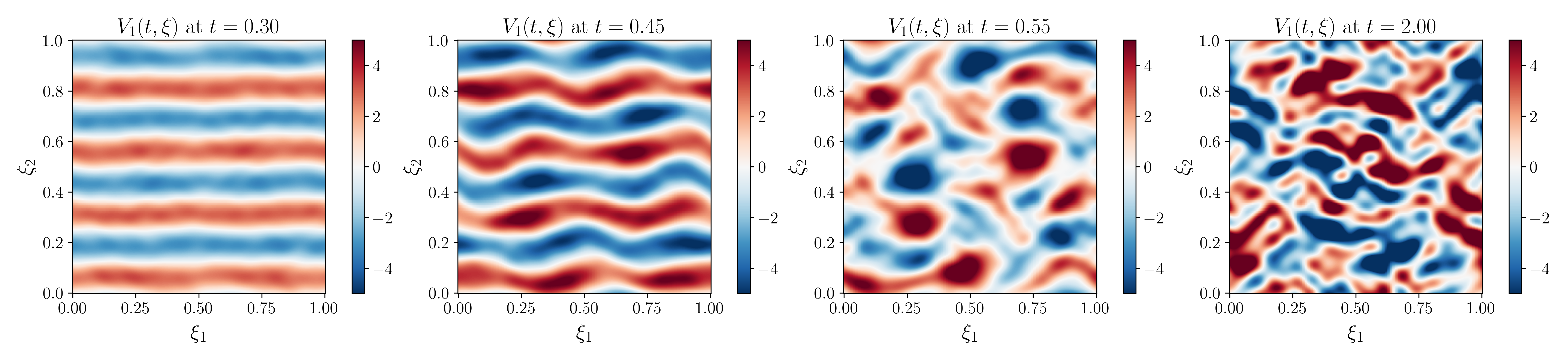}
        \caption{{\em Generation of the turbulent initial condition}. Profiles of the first component of the velocity $V=(V_1,V_2)$ for the dEB model Eq.~\eqref{eq:deb-sde} with Kolmogorov forcing Eq.~\eqref{eq:k-forcing} over $t\in[0,t_{\text{turb}}]$.}
    \label{fig:turbulence}
\end{figure}

The behavior of the dEB model with Kolmogorov forcing is illustrated in \autoref{fig:turbulence} (see also the first two rows of \autoref{fig:dEB}). We fix a single realization of Wiener noise with rate $q=5\times 10^{-4}$ and integrate the system forward in time using the EM-method Eq.~\eqref{eq:euler-maruyama}, from the zero state $X(0)=0$ at $t=0$ until turbulence is observed at $t=t_{\text{turb}}=2.0$, using a timestep $\Delta t = 10^{-4}$ to ensure stability of EM. The resulting state $x=X(t_{\text{turb}})$ is taken as the \textit{deterministic} turbulent IC. We take damping $\lambda = 10^{-4}$, Kolmogorov forcing with amplitude $F_0=10$ and frequency $k_f=4$, and coefficients $c_{ijk}$ with regularization parameters $\alpha=10^{-2}$ and $\beta=1.6$. We note that this choice of $\alpha,\beta$ falls into the natural target regime for quantum algorithms (see Conjectures \ref{conj-sparsity} and \ref{conj-quadratic-lambda}). In the first panel of \autoref{fig:turbulence}, at $t=0.3$, the velocity only slightly deviates from the analytical solution Eq.~\eqref{eq:deb-analytical-sol}, suggesting that we are essentially in a linear regime. In the second panel, at $t=0.45$, the velocity begins to diverge from the analytical solution, and finally, at $t=t_{\text{turb}}=2.0$, turbulence has fully developed. The complexity of $V$ grows with $t$: at $t=0.3$, $V$ can be represented by just a few basis functions as it behaves approximately like Eq.~\eqref{eq:deb-analytical-sol}, which would require only one basis function. However, as $t$ grows, increasingly many basis functions are required to resolve $V$ with fixed precision.

\paragraph{Averaged observable $v(t,x)$.} To simulate the dynamics of $v(t,x)$ in Eq.~\eqref{eq:numerics-v(t,x)}, we remove the Kolmogorov forcing term and simulate the dEB model Eq.~\eqref{eq:deb-sde} from the noisy initial conditions $x+z$ where $x=X(t_{\text{turb}})$ with $z_i\sim \calN(0,q/(2\lambda))$. We reduce the Wiener noise rate from $q=5\times 10^{-4}$ to ${q}=2.5\times 10^{-4}$, as this in turn allows for a reduced noise variance $q/(2\lambda)=1.25$ in the turbulent initial condition. The parameters $\lambda,\alpha,\beta,N$ are all taken as in the previous section. We simulate from $t=t_{\text{turb}}=2.0$ to $t=T=3.0$ over $5,000$ realizations of the noise in the initial condition using the EM-method. Following the EM error estimate in Eq.~\eqref{eq:em-mcmc-error}, the timestep is taken as $\Delta t=C_1/ \sqrt{P}$ where we take $C_1=1.2\times 10^{-3}$ to ensure stability of the EM-method, thus giving a timestep of $\Delta t=C_1 / \sqrt{P} = 1.6\times 10^{-5}$ for $P=5,000$ samples. \autoref{fig:sample-variance} shows the sample variance in the coefficients $X_i(T)$ and $X_i^2(T)$ for $i=1,...,N$, which have converged on average to the values $\sigma^2(X(T))=1.97$ and $\sigma^2(X^2(T))=7.93$, respectively. Empirical observation suggests that $C_2=O(1)$ and so, following Eq.~\eqref{eq:em-mcmc-error}, to achieve an absolute error of $\varepsilon_P = 0.1$ in the observable $\varphi(X) = X_i$ with $98\%$ probability\footnote{We take $z$-score $z=2.326$ for a two-tailed confidence level of $98\%$.} one requires $P= \lceil \sigma^2(X(T)) z^2 / \varepsilon^2 \rceil = 1,066$ samples. Similarly, one requires $P= \lceil \sigma^2(X^2(T)) z^2 / \varepsilon^2 \rceil = 4,291$
samples to achieve the same error $\varepsilon_P$ in the observable $\varphi(X) = X_i^2$ with $98\%$ probability.

\begin{figure}[htbp]
\centering\includegraphics[width=0.9\textwidth]{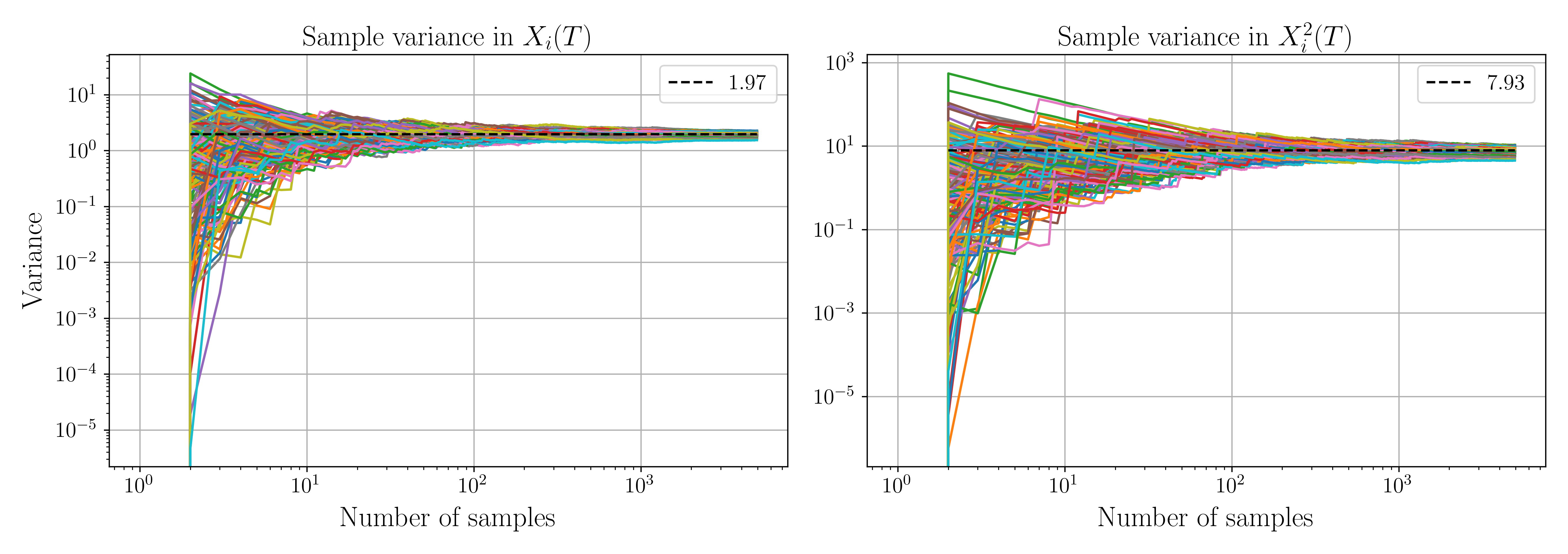}
        \caption{{\em Sample variance in projection coefficients}. Unbiased sample variance in all $N$ coefficients $X_i(T)$ and $X_i^2(T)$ at the final time $T=3.0$.
        }
    \label{fig:sample-variance}
\end{figure}

For the observable, we consider
\begin{equation}\label{eq:numerics-u0}
u_{0,\xi^*}(X(t)) = \sum_{i=1}^N X_i(t) \big( \, (1,0) \cdot \left( w(p_i) e_{p_i}(\xi^*) \right) \big),
\end{equation}
defined at the spatial points  $\xi^*\in l=\{(0.5,0),...,(0.5,1)\}$. That is, $u_{0,\xi^*}(X(t))$ approximates the first component of the velocity field $V_1(t,\xi^*)$ along the vertical line $\xi^*\in l$. The observable is a weighted linear combination of $N$ linear observables $x_i$ (i.e., $\prod_{i=1}^N x_i^{d_i}$ for $d_i=1$) with state encoding
\[
|\psi(0)\ra = \sum_{i=1}^N \sqrt{\frac{q}{2\lambda_i}} (1,0) \cdot \left( w(p_i) e_{p_i}(\xi^*) \right) |\mm_i\ra.
\] 
It follows from Eq.~\eqref{eq:I-beta>1-both} that, for $\beta>1$,
\begin{equation*}
\begin{aligned}
    \|u_{0,\xi^*}\|^2_2  = \la \psi(0)|\psi(0)\ra = \sum_{i=1}^N \frac{1}{( 1+\alpha \|p_i\|^\beta )^2} \leq \frac{2^{2\beta+1} \pi^2 \alpha^{-2/\beta}(\beta-2)}{\beta^2 \sin(2\pi/\beta)}
\end{aligned}
\end{equation*}
and so we conjecture that $u_{0,\xi^*}$ can be efficiently prepared as a quantum state. As a corollary of Lemma~\ref{lemma:alpha-norm-bound}, it can be further shown that the observable $u_{0,\xi^*}$ converges to the ``unfiltered" observable (defined as in Eq.~\eqref{eq:numerics-u0} but without $w(p_i)$ factor) at the rate $O(\alpha^{1/\beta})$. 

Rows 3, 4 and 5 of~\autoref{fig:dEB} show the expectation values $v(t,x)$ and associated averaged velocity fields obtained from $5,000$ realizations of the noisy initial condition and Wiener noise. The observable functions $u_{0,\xi^*}$ are defined at 100 equally-spaced points $\xi^*\in l=\{(0.5,0),...,(0.5,1)\}$. \autoref{fig:v-truncation} shows a comparison of the expectation $v(t,x)$ for the filtered and unfiltered $u_{0,\xi^*}$. It is clear that the filtered observable serves as a good approximation of the unfiltered observable, i.e., that $u_{0,\xi^*}(X(t))$ is indeed a good approximation of $V_1(t,\xi^*)$ along $\xi^*\in l$.

\begin{figure}[htbp]
\centering\includegraphics[width=\textwidth]{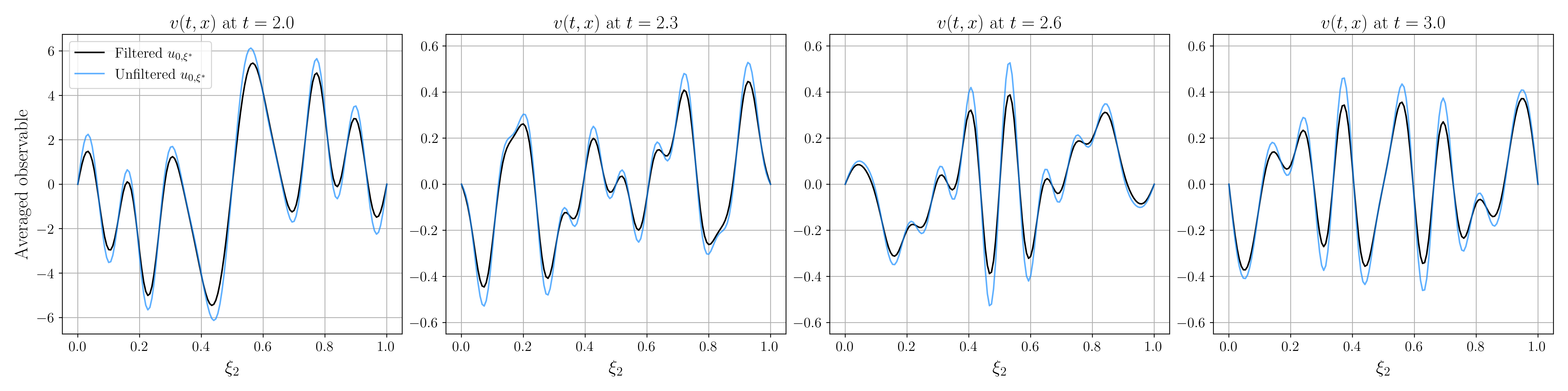}
        \caption{{\em Unfiltered expectation $v(t,x)$}. Comparison of the expected value $v(t,x)$ for observable $u_{0,\xi^*}(X(t))$ with the filter $w(p_i)$ for $(\alpha,\beta)=(10^{-2},1.6)$ versus the ``unfiltered" observable with $(\alpha,\beta)=(0,0)$, over $t\in[t_{\text{turb}},T]$.
        }
    \label{fig:v-truncation}
\end{figure}

\autoref{fig:obserable-variance} reports the sample variance in the averaged observable $v(t,x)$ for the above parameters. Namely, we compute the sample variance $\sigma^2_k$ in $v(t,x)$ at each $\xi_k^*\in l$ and plot the average across all $\sigma^2_k$ values, which we denote $\widehat{\sigma}^2$. That is, $\widehat{\sigma}^2$ represents the sample variance in the expectation value $v(t,x)$, averaged across the points $\xi^*$ along the vertical line $l$. The variance in the observable has clearly converged for each of the plotted timepoints $t\in[t_{\text{turb}},T]$.

\begin{figure}[htbp]
\centering\includegraphics[width=0.5\textwidth]{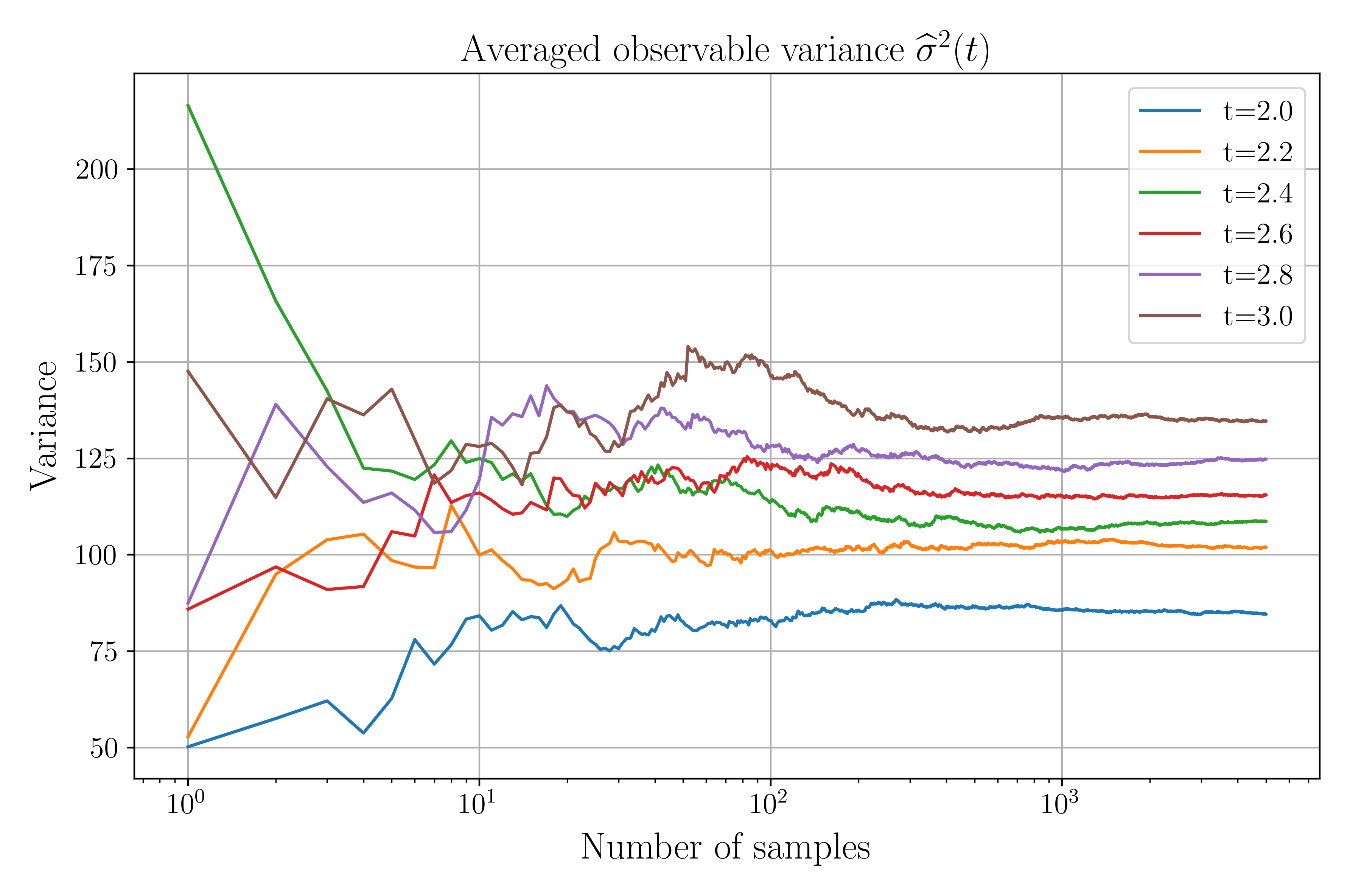}
        \caption{{\em Sample variance in averaged observables}. Sample variance $\widehat{\sigma}^2$ in the expectation $v(t,x)$ averaged across the spatial points $\xi^*\in l$, plotted at timepoints $t\in[t_{\text{turb}},T]$.
        }
    \label{fig:obserable-variance}
\end{figure}

As noted above, the noise rate in the IC is $q/(2\lambda)=1.25$ in the present configuration. The corresponding norm of the readout vector $\psi_{out}(x)$ is $\exp(\lambda \|x\|_2^2/q)=7.3\times 10^{4}$. If we took an increased Wiener noise rate of $q=5\times 10^{-4}$, this would result in an increased noise variance in the IC of $q/(2\lambda)=2.5$ and a reduced readout norm of $\exp(\lambda \|x\|_2^2/q)=2.7\times 10^{2}$. A comparison of the observable $v(t,x)$ obtained for these two Wiener noise rates is presented in \autoref{fig:v-IC-noise}, which highlights that we can still obtain nontrivial dynamics for a small norm in the readout state, but at the expense of larger noise in the IC.

\begin{figure}[htbp]
\centering\includegraphics[width=\textwidth]{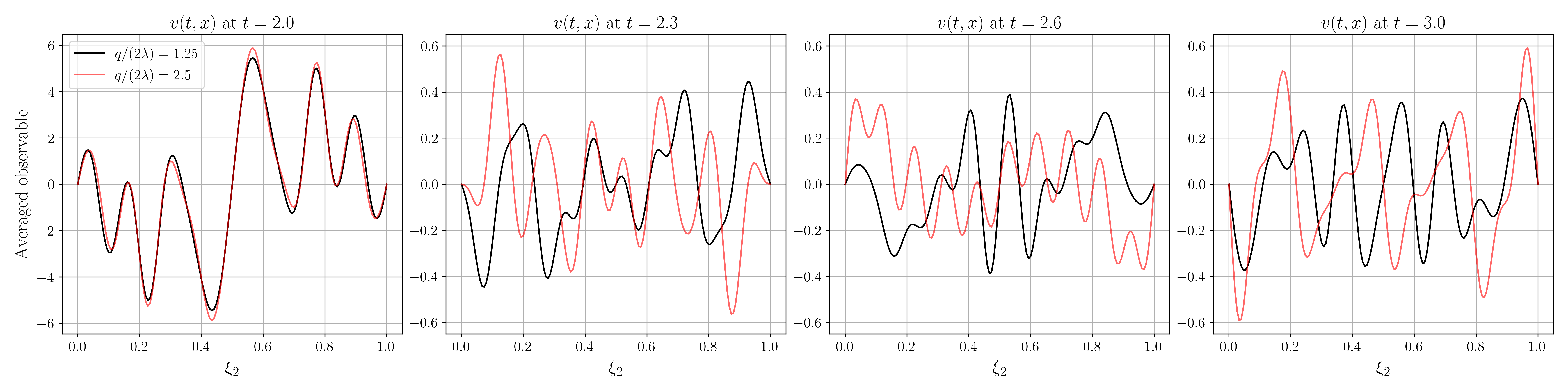}
        \caption{\emph{Expectation $v(t,x)$ for different noise rates $q$.} Comparison of the expected value $v(t,x)$ for Wiener noise $q=2.5\times 10^{-4}$ versus with $q=5\times 10^{-4}$ over $t\in[t_{\text{turb}},T]$, which  correspond to initial condition noise rates of $q/(2\lambda)=1.25$ and $q/(2\lambda)=2.5$, respectively.
        }
    \label{fig:v-IC-noise}
\end{figure}

\autoref{fig:Xi-exp-vals} shows the expectation $v(t,x)$ for each choice of observable function representing the first mode $u_0(X(t))=X_i(t)$, second mode $u_0(X(t))=X_i^2(t)$ and nearest-neighbor terms $u_0(X(t))=X_i(t) X_{i+1}(t)$. Each of these observables can be efficiently prepared as a quantum state (see examples in Section~\ref{sec:init}).

\begin{figure}[htbp]
\centering\includegraphics[width=0.9\textwidth]{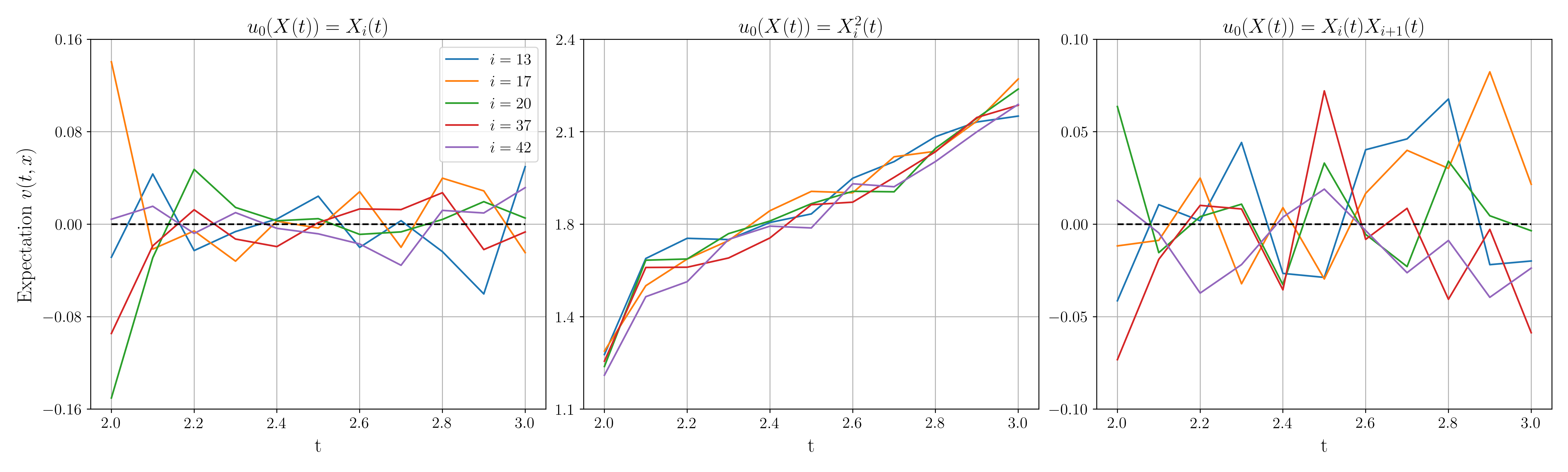}
        \caption{\emph{Expectation $v(t,x)$ for different choices of observable $u_0(X(t))$}. Expected value of the observables $X_i(t)$, $X_i^2(t)$ and $X_i X_{i+1}(t)$ over a subset of indices $1\leq i \leq N$, for $q=2.5\times 10^{-4}$, over $t\in[t_{\text{turb}},T]$.}
    \label{fig:Xi-exp-vals}
\end{figure}

\section{Classical simulations of non-dissipative quadratic ODEs} 
\label{sec:qdrift-euler}
In what follows we discuss classical simulation of the discrete dEB model. 
\begin{conj}\label{conj-sparsity}
Assume that $\beta\ge 5/3$ (so that $J=\sum_{i,j,k=1}^N |c_{ijk}|=O(\polylog{N})$). Further assume that there is an efficient way to sample from a probability distribution induced by the normalized tensor $\frac{c}{J}$. Then the system of ODEs given by 
\begin{equation}
    \label{eq:Fx}
\dfrac{dX_i}{dt} = f_i(X)=\sum_{i,j,k=1}^N c_{ijk} X_jX_k
\end{equation}
with divergence-free drift $c$, and sparse initial condition $X(0)=X_0$ can be solved directly in ODE' state space by a classical randomized time-evolution algorithm with polynomial runtime.     
\end{conj}
\begin{proof}
    For \(1\le i,j,k\le N\), let \(e_k\in\mathbb R^N\) denote the \(k\)-th canonical basis vector, and define
\[
S_{ijk}(x):=\langle e_k,x\rangle\,(E_{ij}-E_{ji}),
\qquad x\in\mathbb R^N,
\]
where \(E_{ij}\) is the canonical matrix with \(1\) in position \((i,j)\) and zeros elsewhere. We introduce the matrix-valued map
\[
H(X):=\frac{2}{3}\sum_{i,j,k=1}^N c_{ijk} S_{ijk}(X)
=
\frac{2}{3}\sum_{i,j,k=1}^N c_{ijk}\,\langle e_k,X\rangle\,(E_{ij}-E_{ji}),
\qquad x\in\mathbb R^N.
\]
We also define the quadratic vector field \( F(X)=H(X)X\). Then \[
F_i(X)=\sum_{j,k=1}^N c_{ijk}X_jX_k = f_i(X).
\]
Clearly $H(x)^\top=-H(x)$ and hence \[
\langle x,H(x)x\rangle=0
\qquad\text{for every }x\in\mathbb R^N.
\]
Now fix an integer \(M\ge 1\), set
\[
h:=\frac{T}{M},
\qquad
t_n:=nh,
\qquad n=0,1,\dots,M,
\]
and define the left-endpoint projection
\[
\eta(t):=t_{n-1}
\qquad\text{for }t\in[t_{n-1},t_n).
\]
Consider the following time-varying piece-wise linear approximation of Eq.~\eqref{eq:Fx} given by: 
\begin{equation}\label{eq:linHx}
\dfrac{dY}{dt} = H(Y(\eta(t)))Y(t)
\qquad
Y(0)=X_0.
\end{equation}
Clearly, the solution of Eq.~\eqref{eq:linHx} can be represented as follows: 
\begin{equation}
    \label{eq:Yexp}
Y(t) = e^{H(Y(t_{n-1}))(t-t_{n-1})}Y(t_{n-1})=
e^{H(Y(t_{n-1}))(t-t_{n-1})}e^{H(Y(t_{n-2}))\,h}\times\dots\times e^{H(Y(0))\,h}Y(0) ,\qquad t\in (t_n,t_{n-1}]
\end{equation}
It can be demonstrated that $\int_0^T\|X(t)-Y(t)\|^2dt\to 0$ as $M\to\infty$, see e.g.~\cite{Zhuk2015EulerEstimation} . Clearly, Eq.~\eqref{eq:Yexp} suggest that Q-Drift like algorithm~\cite{qDriftCampbell2019} can be used to approximate each $e^{H(Y(t_{n-2}))\,h}$ building upon the fact that $H(X):=\frac{2}{3}\sum_{i,j,k=1}^N c_{ijk} S_{ijk}(X)$ and \[
\exp\!\left(S_{ijk}(X)\,h\right) = 
\exp\!\left(X_k\, h\,(E_{ij}-E_{ji})\right)
=
I+\sin(h\,X_k)\,(E_{ij}-E_{ji})+(\cos(h\,X_k)-1)(E_{ii}+E_{jj}).
\] 
For example, to approximate \[
\exp^{H(Y(0))\,h} = \exp\left(\frac{2}{3}\sum_{i,j,k=1}^N |c_{ijk}| \tilde S_{ijk}(X(0))\right),\qquad \tilde S_{ijk}(X) = \operatorname{sign}(c_{ijk}) X_k(E_{ij}-E_{ji})
\]
we first note that the largest singular value of each of $\tilde S_{ijk}(X)$ is $1$ as long as $\|X_0\|\le 1 $, then we define a probability distribution
\[
P_{ijk}:=\frac{|c_{ijk}|}{J},
\qquad 1\le i,j,k\le N.
\]
and assume that 
\[
(i_1,j_1,k_1),\dots,(i_W,j_W,k_W)
\]
are i.i.d. samples drawn from \(P\). Following~\cite{qDriftCampbell2019} we define
\begin{align}
Z(t_1)&:=
\prod_{v=1}^W
\exp\!\left(
\frac{J h}{W}\,
\tilde S_{i_v j_v k_v}(X_0)
\right)X_0\\
&=  \prod_{v=1}^W \left(I+\sin(\frac{J h \operatorname{sign}(c_{i_vj_vk_v})}{W}\,X_{k_v})\,(E_{i_vj_v}-E_{j_vi_v})+(\cos(\frac{J h\operatorname{sign}(c_{i_vj_vk_v})}{W}\,X_{k_v})-1)(E_{i_vi_v}+E_{j_vj_v})
\right)X_0
\end{align}
If $W=\left\lceil \frac{4 J^2 h^2}{\epsilon} \right\rceil$ and $M'$ is the number of different realizations of $Z(t_1)$ then slightly reformulating result of~\cite{qDriftImportanceSampling} we get \[
P(\| Y(t_1) - Z(t_1) \| > \epsilon/2) < 2^{\log_2(N) +1} \exp\{- \frac{W M' \epsilon^2}{11t_1^2 J^2 C}\}
\] where $C$ depends on the probability distribution $P_{ijk}$ (importance sampling). Since $W$ is only polynomial in $\log(N)$, and each $\exp{}$ in $Z(t_1)$ increases the number of non-zero elements of $X_0$ by at most two (two diagonal terms leave the support unchanged and the only increase comes from the skew-symmetric term $(E_{i_vj_v}-E_{j_vi_v})$), we have that $Z(t_1)$ support is at most $\polylog{N}$. Repeating this process $M$ times with $M=O(\polylog{N})$ we obtain the required claim provided there is an efficient sampling routine to sample from $P_{ijk}$. We conjecture that this sampling can be done classically and efficiently using MCMC-machinery, and leave this for the future research.  
\end{proof}
\section{Classical simulations of strongly dissipative quadratic SDEs} 
\label{sec:strongly}

Here we consider a special case when the dissipation rates follow the power law scaling
$\lambda_i \propto i^p$ for some constant $p\ge 1$. Let $k>0$ be the regularization cutoff introduced in Section~\ref{sec:regul-sz}.
We show that the low-dissipation subspace $\calH_k$
simulated by our quantum algorithm has dimension  $\mathrm{dim}(\calH_k) =O(e^{k^\frac{1}{1+p}})$, and outline a potential proof of the following conjecture: despite the fact that resolving classically the entire state $\psi$ of Kolmogorov equation with given error tolerance $\epsilon$ requires $O(e^{k^\frac{1}{1+p}})$ operations with $k=\operatorname{poly}(\epsilon^{-1},J,\lambda_1^{-1})$, in contrast estimating $u(t,x)$ at a point $x$, might be done classically by Euler-Maruyama-type numerical algorithm with runtime just polynomial in $k$ and $\epsilon^{-1}$. We also provide a practically relevant example of ``strong" dissipation: the eigen-values $\lambda_i$ of Stokes operator determining the dissipative part of NSE in 2D.

Let us demonstrate that for $\lambda_i\propto i^p$, $p\ge1$ the dimension of the low-dissipation subspace, $\calH_k$ spanned by Fock basis vectors $|\mm\ra$ satisfying $\la \mm|A|\mm\ra\le k$ is in fact independent of $N$, and depends only on the cutoff parameter $k$: $\mathrm{dim}(\calH_k)=O(e^{k^\frac{1}{1+p}})$. To this end let $\Gamma$ denotes Euler's Gamma function and $\zeta$ the Riemann zeta function:
\[
\Gamma(s)=\int_0^\infty t^{s-1}e^{-t}\,dt \qquad (\Re(s)>0),
\]
\[
\zeta(s)=\sum_{n=1}^\infty n^{-s} \qquad (\Re(s)>1).
\]
\begin{lemma}\label{lem:Gk-asymptotic}
Let $p\ge 2$ be fixed, and let $N\ge k$. Define
\[
G_k:=\left\{m=(m_1,\dots,m_N)\in \mathbb{N}_0^N:\ \sum_{i=1}^N i^p m_i\le k\right\}.
\]
Set
\[
A_p:=\left(\frac{\Gamma\!\left(1+\frac1p\right)\zeta\!\left(1+\frac1p\right)}{p}\right)^{\!p/(p+1)}.
\]
Then, for any $\epsilon>0$ there is a constant $C_0$ such that as $k\to\infty$,
\[
\operatorname{dim}(\calF_0)=|G_k|
\le C_0+ (1+\epsilon)
\frac{A_p}{(2\pi)^{(p+1)/2}\sqrt{1+\frac1p}}\,
\exp\!\left((p+1)A_p\,k^{1/(p+1)}\right).
\]
\end{lemma}

\begin{proof}
For $n\ge 0$, let
\[
Q_p(n):=\#\left\{m=(m_1,\dots,m_N)\in \mathbb{N}_0^N:\ \sum_{i=1}^N i^p m_i=n\right\}.
\]
Then
\[
|G_k|=\sum_{n=0}^k Q_p(n).
\]

Since $N\ge k$, the equality
\[
\sum_{i=1}^N i^p m_i=n,\qquad n\le k,
\]
implies that $m_i=0$ for $i>n$. Therefore $Q_p(n)$ is exactly the number of partitions of $n$
into parts from the set
\[
\mathcal A_p:=\{1^p,2^p,3^p,\dots\}.
\]
%Equivalently, $Q_p(n)=p_{\mathcal A_p}(n)$ in the notation of \cite{BMZ}.

Recall that $p>1$. Then, by the asymptotic formula for partitions into perfect $p$-th powers
(see \cite{BMZ,TWL}),
\[
Q_p(n)\sim B_p\,n^{-\alpha_p}\exp\!\bigl(C_p n^{\delta_p}\bigr),
\qquad n\to\infty,
\]
where
\[
\delta_p=\frac1{p+1},\qquad
\alpha_p=\frac{3p+1}{2(p+1)},
\]
and
\[
B_p=\frac{A_p}{(2\pi)^{(p+1)/2}\sqrt{1+\frac1p}},
\qquad
C_p=(p+1)A_p.
\]
Hence, for every $\varepsilon>0$, there exists $n_0$ such that for all $n\ge n_0$,
\[
Q_p(n)\le (1+\varepsilon)B_p\,n^{-\alpha_p}\exp\!\bigl(C_p n^{\delta_p}\bigr).
\]
Therefore, for $k\ge n_0$,
\[
|G_k|
\le
\sum_{n=0}^{n_0-1}Q_p(n)
+
(1+\varepsilon)B_p\sum_{n=n_0}^k n^{-\alpha_p}e^{C_p n^{\delta_p}}.
\]
Since $n^{\delta_p}\le k^{\delta_p}$ for $n\le k$, we obtain
\[
|G_k|
\le
C_0+(1+\varepsilon)B_p e^{C_p k^{\delta_p}}\sum_{n=n_0}^k n^{-\alpha_p}.
\]
Because $\alpha_p>1$ when $p>1$, the series $\sum_{n=1}^\infty n^{-\alpha_p}$ converges, and thus
\[
|G_k|=O\!\left(e^{C_p k^{\delta_p}}\right).
\]
\end{proof}

\begin{rem}\label{rem:p=1}
Similar asymptotic formula holds for $p=1$, provided one sets
\[
A_1=\left(\Gamma(2)\zeta(2)\right)^{1/2}=\frac{\pi}{\sqrt6}.
\]
Indeed, in this case $Q_1(n)$ is the ordinary partition function with asymptotics determined by the classical Hardy--Ramanujan formula:
\[
Q_1(n)\sim \frac{1}{4\sqrt3\,n}\exp\!\left(\pi\sqrt{\frac{2n}{3}}\right)
\qquad (n\to\infty),
\]
Repeating verbatim the summation
argument from the proof of Lemma~\ref{lem:Gk-asymptotic}, one obtains
\[
|G_k|\le C_0 + (1+\epsilon)
\frac{\log(k)}{4\sqrt3}\,
\exp\!\left(\pi\sqrt{\frac{2k}{3}}\right),
\qquad (k\to\infty),
\]
\end{rem}
Lemma~\ref{lem:Gk-asymptotic} suggests that the dimension of the projected Fock space is independent of $N$, and in turn motivates the following.
\begin{conj}\label{conj-quadratic-lambda}
Assume that $u_0(x_1\dots x_N)=u_0(x_1\dots x_k,0,\dots,0)$. If $\lambda_i\propto i^p$, $p=1,2,..$ then there exists a classical Euler-Maruyama-type numerical algorithm that estimates the expected value $v(t,x)$ with the runtime which is polynomial in $k$ and $\epsilon^{-1}$.
\end{conj}
\begin{proof}
  Let us outline a potential proof. Assume that $\lambda_1\le\dots\lambda_N$ and consider a system of SDEs
with $N$ real variables $\vec X_N=(X_1,\ldots,X_N)$ of the form
\be
\label{SDE_N}
dX_i(t) = -\lambda_i X_i(t) dt  + f_i(X(t))dt  +
\sqrt{q}\, dW_i(t),  \qquad t\ge 0.
\ee
Let \be
\label{uN(t,x)}
u_N(t,x) =
  \EE_{W} u_0(\vec X_N(t))
    \quad \mbox{subject to $X_N(0)=x$}.
 \ee
Here the expected value is over Wiener noise $W(t)$ in Eq.~(\ref{SDE_N}).
Recall that $u_N$ solves Kolmogorov equation
\begin{equation}
  \label{eq:KEN}
\frac{\partial}{\partial t} u_N(t,x) = \sum_{i=1}^N (-\lambda_i x_i + f_i(x)) \frac{\partial}{\partial x_i} u_N(t,x)
+\frac{q}2 \sum_{i=1}^N  \frac{\partial^2}{\partial x_i^2} u_N(t,x)
\end{equation}
with the initial condition $u_N(0,x)=u_0(x)$.  Let us introduce a system of
SDEs with $m$ real variables, $\vec X_m=(X_1,\ldots,X_m)$
\be
\label{SDE_k}
dX_i(t) = -\lambda_i X_i(t) dt  + f_i(X(t))dt  +
\sqrt{q}\, dW_i(t),  \qquad t\ge 0.
\ee
and let \be
\label{uk(t,x)}
u_m(t,x) =
  \EE_{W} u_0(\vec X_m(t))
    \quad \mbox{subject to $\vec X_m(0)=x$}.
    \ee
Here $\vec X_m(0)=x$ is understood for the first $m$-components of $x$.     
Similarly $u_m$ solves Kolmogorov equation
\begin{equation}
  \label{eq:KEk}
\frac{\partial}{\partial t} u_m(t,x) = \sum_{i=1}^m (-\lambda_i x_i + f_i(x)) \frac{\partial}{\partial x_i} u_m(t,x)
+\frac{q}2 \sum_{i=1}^m  \frac{\partial^2}{\partial x_i^2} u_m(t,x)
\end{equation}
with the initial condition $u_m(0,x)=u_0(x_1\dots x_m,0,\dots,0)$. Consider the regularized bKE defined in Section~\ref{sec:regul-sz} with regularization cutoff $m$: 
\be
\label{KE_reg}
\frac{d}{dt} |\psi_m(t)\ra = (-A + \Pi_m C\Pi_m ) |\psi_m(t)\ra
\ee
with the initial condition $|\psi_m(0)\ra = |\psi(0)\ra$.
The key observation here is that the regularized bKE~\eqref{KE_reg}, obtained from Eq.~\eqref{eq:KEN}, coincides with the regularized bKE obtained from Eq.~\eqref{eq:KEk}. 
This follows immediately from the fact that the inequality $\sum_{i=1}^N\lambda_i m_i \le m $ forces all $m_i$ corresponding to $\lambda_i>m$ to $0$. The latter implies that the regularized solution is composed of linear combinations of $m$-variate polynomials
\be
\label{Hermite_normalizedk}
\mathbb{H}_{\multi{m}}(x) = \prod_{i=1}^m  \frac1{\sqrt{(m_i)!}} \herm_{m_i} \left(x_i\sqrt{2\lambda_i/q}\right),
\ee
by the definition of $\calJ_m$. Hence approximation of $u_N$ obtained from $\psi_m$ is independent of variables $x_{m+1}\dots$. 

Let $\calM$ denote the mapping of $u$ to its representation in the Fock space, and let $L^2_\mu(R^N)$ denote space of square-integrable functions with Gaussian weight $\mu$. Extend $u_m$ from $R^m$ to $R^N$ by adding $0$-order univariate Hermite polynomials for variables $x_{m+1}\dots$.  Then $\| u_N(t,\cdot)-u_m(t,\cdot)\|_{L^2_\mu}\le \|\calM(u_N(t,\cdot))-\psi_m(t)\|+\|\calM(u_m(t,\cdot))-\psi_m(t)\|\le 2\epsilon$ by Theorem~\ref{thm:regul}. Now assuming that $u_N$ and $u_m$ are smooth enough (e.g. Lipschitz continuity is enough for pointwise convergence of Hermite series) to have that $|u_N(x,t)-u_m(x,t)|\le C \| u_N(t,\cdot)-u_m(t,\cdot)\|_{L^2_\mu}\le 2\epsilon C$ for some $C>0$ we deduce that \[
  |u_N(x,t)-EM(x,t)|\le  |u_N(x,t)-u_m(x,t)|+|u_m(x,t)-EM(x,t)| \le  2\epsilon C + \epsilon_1
  \] where $EM(x,t)$ is Euler-Mariyama based approximation of $u_m(t,x) =
  \EE_{W} u_0(\vec X_m(t))$ with error $\epsilon_1$. What remains to be shown is that weak EM algorithm for NSE-like systems with strong dissipation is at most polynomial in $m$ and $\epsilon_1^{-1}$ for polynomial $u_0$, which amounts to estimating variance of the random variable $u_0(\vec X_m(t))$ and applying Chebyshev-type inequality to estimate the required number of samples to reach the required level of precision, $\epsilon_1$. We leave this to the future research.

\end{proof}

Finally we provide a practically relevant example of ``strong" dissipation: the eigen-values $\lambda_i$ of Stokes operator determining the dissipative part of NSE in 2D.
Recall that the classical NSE in 2D is a system of two PDEs defining dynamics of the scalar pressure field $p(x,y)$ and the viscous fluid velocity vector-field $\vec v(t,x,y)=[v_1(t,x,y),v_2(t,x,y)]$ which depends on the initial condition $\vec v(0)=\vec v_0$, forcing $\vec g$, and Boundary Conditions (BC), e.g. periodic BC $v_{1,2}(t,x+\ell_1,y) = v_{1,2}(t,x,y)$, $v_{1,2}(t,x,y+\ell_2) = v_{1,2}(t,x,y)$. In the vector form it reads as follows:
\begin{equation}
  \label{eq:NSE-vector}
  \begin{split}
    \dfrac{d\vec{v}}{dt} &+ (\vec{v} \cdot \nabla) \vec{v} - \nu\Delta\vec{v}+\nabla p = \vec g,\quad \nabla\cdot\vec{v} = 0
  \end{split}
\end{equation}
Eigen-values of the Stokes operator $A=-\Delta$ take the following form: (see~\cite{flandoli1998kolmogorov})
\begin{align*}
    A\vec{e}_{\vec k} &= \Lambda_{\vec k} \vec{e}_{\vec k}, & \Lambda_{\vec k} &= 4\pi^2|\vec k|^2_2,&\vec k &= (k_1,k_2)^\top, k_{1}\in\mathbb{N}, k_{2}\in\mathbb{Z}\\
    \vec{e}_{\vec k}&= \frac{\vec{k}^\perp}{|\vec k|_2}\sqrt{2}\sin(2\pi\vec{k}\cdot\vec{\xi}),
    & \vec{k}^\perp&=(k_2,-k_1)^\top,&\vec{k}\cdot\vec{\xi}&=k_1\xi_1+k_2\xi_2, |\vec k|_2^2=\vec k\cdot\vec k
\end{align*}
Let $\lambda_i$ denote the $i$-th element of the set of all eigen-values $\Lambda_{\vec k}$ written in nondecreasing order and counted with multiplicity. Then Galerkin projection onto $\vec e$ results in SDE with dissipation given by $\lambda_iX_i$. The following lemma demonstrates that $\lambda_i\sim i$.
\begin{lemma}
Let
\[
N(\Lambda):=\#\{(k_1,k_2)\in \mathbb N\times\mathbb Z:\ \Lambda_{\vec k}\le \Lambda\}.
\]
Then, as $\Lambda\to\infty$,
\[
N(\Lambda)=\frac{\Lambda}{8\pi}+O(\sqrt{\Lambda}).
\]
In particular, if $\lambda_i$ denotes the $i$-th element of the multiset
$\{\Lambda_{\vec k}:k_1\in\mathbb N,\ k_2\in\mathbb Z\}$
written in nondecreasing order and counted with multiplicity, then
\[
\lambda_i\sim 8\pi\, i.
\]
\end{lemma}
\begin{proof}
By definition,
\[
\Lambda_{\vec k}\le \Lambda
\iff
4\pi^2(k_1^2+k_2^2)\le \Lambda
\iff
k_1^2+k_2^2\le R^2, \quad R:=\frac{\sqrt{\Lambda}}{2\pi}.
\]
Hence
\[
N(\Lambda)
=
\#\{(k_1,k_2)\in \mathbb N\times\mathbb Z:\ k_1^2+k_2^2\le R^2\}.
\]

This is the number of lattice points in the half-disk
\[
\{(x,y)\in\mathbb R^2:\ x>0,\ x^2+y^2\le R^2\},
\]
up to the contribution of the boundary line $x=0$, which contains only $O(R)$ lattice points.
Therefore Gauss' area law gives
\[
N(\Lambda)=\frac{\pi R^2}{2}+O(R).
\]
Substituting $R=\sqrt{\Lambda}/(2\pi)$ yields
\[
N(\Lambda)
=
\frac{\pi}{2}\cdot\frac{\Lambda}{4\pi^2}
+
O(\sqrt{\Lambda})
=
\frac{\Lambda}{8\pi}+O(\sqrt{\Lambda}).
\]
Define \[
N(\lambda_i^-):=\#\{(k_1,k_2):\Lambda_{\vec k}<\lambda_i\}.
\]
Now by definition of the counting function,
\[
N(\lambda_i)\ge i
\qquad\text{and}\qquad
N(\Lambda_i^-)\le i,
\]
because at most the first $i-1$ eigenvalues can be strictly smaller than $\lambda_i$, while the first $i$ eigenvalues are all $\le \lambda_i$.
So asymptotically $N(\lambda_i)\sim i$. Since
\[
N(\Lambda)=\frac{\Lambda}{8\pi}+O(\sqrt{\Lambda}),
\]
we obtain
\[
i\sim \frac{\lambda_i}{8\pi},
\]
and hence
\[
\lambda_i\sim 8\pi\, i.
\]
\end{proof}

\section*{Acknowledgments}
The authors thank Andrew Childs, Arkopal Dutt, Jay Gambetta, Hari Krovi, Martin Mevissen, Juan Bernabe Moreno,
Kristan Temme, and Barry Sanders for helpful discussions. 

\bibliographystyle{unsrtnat}
\bibliography{mybib}

\end{document}